\documentclass[11pt]{article}

\usepackage{amssymb,amsmath,amsthm,amsfonts}
\usepackage[T1]{fontenc}
\usepackage[tt=false, type1=true]{libertine}
\usepackage[varqu]{zi4}
\usepackage[libertine]{newtxmath}
\usepackage[margin=1in]{geometry}
\usepackage{enumitem}
\setlist{nosep,topsep=0pt,leftmargin=*}

\usepackage{hyperref,titlesec,subcaption}
\hypersetup{
    colorlinks,
    allcolors=teal
}

\usepackage[sortcites,sorting=nyt,style=alphabetic,maxbibnames=10]{biblatex}
\usepackage[ruled]{algorithm2e}
    
    \SetAlFnt{\small}
    \SetAlCapFnt{\small}
    \SetAlCapNameFnt{\small}
    \SetAlCapHSkip{0pt}
    \IncMargin{-\parindent}

\usepackage{tikz,xifthen,physics,xcolor,suffix,nicefrac,enumitem,thm-restate,wrapfig,comment}
\usepackage[dvipsnames]{xcolor}
\usetikzlibrary{shapes,positioning,arrows.meta}
\usetikzlibrary{trees}
\usepackage[bb=dsserif]{mathalpha}

\usepackage[capitalize]{cleveref}

\usepackage{subcaption}
\usepackage[normalem]{ulem}

\newcommand{\itparagraph}[1]{\smallskip\noindent\textit{#1}}
\newcommand{\bfparagraph}[1]{\smallskip\noindent\textbf{#1}}

\newcommand{\orderT}[1]{\tilde{\mathcal{O}}\qty(#1)}
\WithSuffix\newcommand\orderT*[1]{\tilde{\mathcal O}(#1)}
\DeclareMathOperator*{\argmax}{arg\,max}

\newcommand{\R}{\mathbb{R}}
\newcommand{\N}{\mathbb{N}}
\newcommand{\floor}[1]{\left\lfloor{#1}\right\rfloor}
\WithSuffix\newcommand\floor*[1]{\lfloor{#1}\rfloor}
\newcommand{\ceil}[1]{\left\lceil{#1}\right\rceil}
\WithSuffix\newcommand\ceil*[1]{\lceil{#1}\rceil}

\let\e\varepsilon

\newcommand{\Line}[4]{%
    #1&%
    \ifthenelse{\isempty{#2}}{\phantom{=}}{#2}%
    #3%
    \ifthenelse{\isempty{#4}}{}{&&\qquad\left(\big.\substack{#4}\right)}%
}

\newcommand{\gl}{\ensuremath{\mathtt{gl}}}

\newcommand{\parent}{\pi}
\newcommand{\ch}{\mathtt{ch}}

\newcommand{\desc}{\mathtt{desc}}
\newcommand{\subTree}{\mathtt{SubT}}

\newcommand{\agentTree}{\mathcal{T}}
\newcommand{\convGraph}{\mathcal{C}}
\newcommand{\inEdges}{\mathtt{in}}
\newcommand{\outEdges}{\mathtt{out}}

\newcommand{\ind}[1]{\mathbb{1}\left[#1\right]}

\newcommand{\figwidth}{0.48\textwidth}

\titlespacing{\subsection}{0pt}{5pt plus 1pt minus 0.5pt}{3pt plus 1pt minus 0.5pt}
\newcommand{\ignore}[1]{}
\newcommand{\Description}[1]{}
\newtheorem{theorem}{Theorem}[section]

\newtheorem{lemma}[theorem]{Lemma}
\newtheorem{corollary}[theorem]{Corollary}
\newtheorem{proposition}[theorem]{Proposition}

\newtheorem{definition}[theorem]{Definition}
\newtheorem{remark}[theorem]{Remark}
\theoremstyle{definition}
\newtheorem{assumption}[theorem]{Assumption}

\newcommand{\citet}[1]{\textcite{#1}}
\newcommand{\citep}[1]{\parencite{#1}}

\title{Resource Allocation and Conversion along the Org Chart}

\author{
Yuan Deng\\
Google Research\\
\texttt{dengyuan@google.com}
\and
Giannis Fikioris%
\thanks{Supported by the Google PhD Fellowship and the ONR MURI grant N000142412742.}\\
Cornell University\footnote{Conducted while part-time Student Researcher at Google Research.}\\
\texttt{gfikioris@cs.cornell.edu}
\and
Chido Onyeze%
\thanks{Supported by AFOSR grant FA9550-23-1-0068.}\\
Cornell University\footnote{Conducted while part-time Student Researcher at Google Research.}\\
\texttt{chidoonyeze@cs.cornell.edu}
\and
Renato Paes Leme\\
Google Research\\
\texttt{renatoppl@google.com}
\and
Balasubramanian Sivan\\
Google Research\\
\texttt{balusivan@google.com}
\and
Mihai Tiuca\\
Google\\
\texttt{mtiuca@google.com}
}

\date{\vspace{-25pt}}

\begin{document}

\let\oldabovedisplayskip\abovedisplayskip
    \setlength{\abovedisplayskip}{\oldabovedisplayskip-3pt}
\let\oldbelowdisplayskip\belowdisplayskip
    \setlength{\belowdisplayskip}{\oldbelowdisplayskip-3pt}
\let\oldtextfloatsep\textfloatsep
    \setlength{\textfloatsep}{\oldtextfloatsep-10pt}

\maketitle{}
\thispagestyle{empty}

\begin{abstract}
    We consider the allocation of multiple heterogeneous resources to agents who are organized according to an organizational hierarchy.
In a company those correspond to business units, departments, and engineering teams.
In government it corresponds to federal, state, and municipal levels as well as various agencies within each.
In a university it corresponds to schools, departments, and research groups.
The resources are also distributed along the same organizational tree: part of the supply is available only to certain sub-trees since it is purchased for the exclusive use of certain departments or units, and some supply is available to the entire tree.
We consider the allocation with multiple types of resources where there is a possibility of converting between certain pairs of resources.

We formulate this allocation problem as a market equilibrium problem and derive the necessary conditions to find a feasible solution.
We prove that such a feasible solution always exists and provide an algorithm to compute it.
Finally, we test our algorithm on a real dataset of the allocation of computing resources, such as TPUs and GPUs at Google.
Our results rely on a novel two-step process, differentiating it from previous approaches.
First, we show how to solve ``easy'' (harmonic) instances of our problem that satisfy certain structural properties.
Then, we show how to efficiently approximate general instances by a series of easier harmonic instances, whose solutions converge efficiently to the solution of the original problem.

\end{abstract}

\clearpage
\setcounter{page}{1}

\section{Introduction}

Any large organization (e.g., company, government agency, university) must decide how to divide a shared, often heterogeneous, pool of resources among the many teams or subunits that sit beneath it in its hierarchy, and do so in a way that respects both the chain of authority and each unit's priorities.
We introduce a new resource management problem motivated by a setting that arises in several organizations when distributing resources along the organizational chart (org chart).
For example, the system deployed at Google for allocating AI training accelerators (such as TPUs and GPUs) uses a market-based mechanism where different teams within the organization receive credits to purchase access to different accelerators in an internal market.
This system was designed and deployed very recently by~\cite{sivanquota}, implementing a simplified version (see~\cref{sec:related_work} for details).
In this work, we formalize this as a new class of combinatorial market-design problems.
We show that there always exists a \emph{unique} market-clearing allocation that can be computed \emph{efficiently} using a novel iterative algorithm we call the \emph{Budget Descent Algorithm}.

Our focus is on markets with two defining characteristics:
\itparagraph{Feature 1: Hierarchical Structure.}
The market mirrors an organization's hierarchy: beneath a single root lie distinct units, each comprising multiple teams (in our running Google example, the units are business divisions and the teams are individual product teams).
Artificial currency, in the form of credits, is endowed to individual teams who can expend these credits to procure resources (e.g., computer chips) to meet their demands (e.g., computational).
On the supply side, resources are provisioned at various levels in the hierarchy.
Some resources are accessible to all teams, while others are purchased for the exclusive use of specific units.
Teams themselves may also procure resources for their own exclusive use.

\itparagraph{Feature 2: Resource Conversions.} Many resources are heterogeneous.
For example, accelerator chips vary by location (e.g., North America, Europe), by generation (newer generations are more capable), etc.
Some compute jobs demand specific chip types---for instance, requiring a North American chip due to data locality or regulatory constraints, or a newer generation chip for specific hardware features.
Conversely, other jobs are flexible and can utilize chips in any location (which we term ``global'' chips, although of course chips are always physically present in some specific location) or of any generation.
This substitutability can be modeled as a resource conversion graph (see Figure \ref{fig:resource_conversion}).

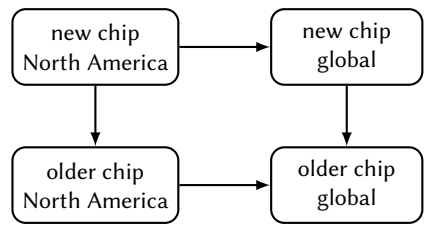
\begin{wrapfigure}{r}{0.4\textwidth}%
    \begin{center}%
        \begin{tikzpicture}[
    node distance=0.8cm and 1.2cm,
    auto,
    block/.style={
        rectangle,
        draw=black,
        thick,
        rounded corners=5pt,
        minimum width=2.cm,
        minimum height=1.cm,
        align=center,
        font=\sffamily\footnotesize
    },
    arrow/.style={
        ->,
        thick,
        >=latex 
    }
]

    \node (nn) [block] {new chip \\ North America};
    \node (on) [block, below=of nn] {older chip \\ North America};
    \node (ng) [block, right=of nn] {new chip \\ global};
    \node (og) [block, right=of on] {older chip \\ global};

    \draw [arrow] (nn) -- (on);
    \draw [arrow] (nn) -- (ng);
    \draw [arrow] (on) -- (og);
    \draw [arrow] (ng) -- (og);

\end{tikzpicture}%
    \end{center}%
    \vspace{-12pt}%
    \caption{Resource Conversion Graph}
    \label{fig:resource_conversion}
    \Description{A directed graph whose nodes are resource types and whose edges indicate the permissible conversions between them.}
\end{wrapfigure}

More formally, this market structure is characterized by two graphs.
First, a tree representing the hierarchy of the participants, where each node may have resource demands and available supply.
A node's supply is available exclusively to the subtree rooted at that node, and conversely, a subtree first consumes resources provisioned at the nodes inside the subtree before purchasing resources from outside the subtree.
The second graph is a resource-conversion graph describing the permissible conversions between resource types.
These two structures, coupled with the natural market forces of supply and demand, give rise to a new class of combinatorial market design problems.

While our running example will be the Google application described above (for which we provide simulational data in \cref{sec:simulations}) this structure is not particular to only this example.
The same structure---a participant hierarchy over a resource-conversion graph---recurs in many other markets.
Examples include energy distribution (federal-to-state-to-county regulators converting among electricity, gas, and solar), international trade (export quotas reserving domestic commodity supply before imports are allowed), multi-echelon manufacturing (raw materials converted downstream into finished goods along a technological graph), and waste management (municipal-to-regional hierarchies sorting commingled streams into recyclables or energy).

Of these examples, we focus on those where resources are allocated using artificial currencies, i.e., the agents derive no utility from holding the currency outside the specific market environment.
The artificial-currency assumption fits several of these settings exactly.
In large-scale government procurement and internal corporate markets like Google's, an exogenous process hands out fixed budgets earmarked for specific activities---such as purchasing energy---that function purely as priority signals, with no value outside the market.
Our primary contribution is a unified, abstract combinatorial market model capturing these hierarchical and networked constraints, along with algorithms to compute a market equilibrium.

When designing our desiderata for such markets, we have two goals in mind:

\begin{enumerate}
    \item As many resources as the central agent has available should be consumed.
    \item The amount of resource each agent gets should be proportional to the amount of currency.
\end{enumerate}

This second goal is a notion of fairness, but not one of social equity across the organization.
Rather, the budget allocation itself encodes the goals of the exogenous process: in government settings, budgets promote equitable access across agencies or regions; in corporate settings like Google, they reflect long-term strategic priorities and profit potential.
Our market then translates these credit endowments faithfully into resource allocations in proportionally to each team's budget. This ensures the realized allocation always mirrors the organization's stated priorities, regardless of which resources happen to be scarce.

\itparagraph{Why a market?}
Why run a market instead of \emph{direct priority-weighted allocation}, where the central agent hands each team a share of every resource proportional to its priority?
Such a rule preserves the intended priorities only in the simplest cases: a team's share of a resource it cannot use is wasted, and any change in supply or demand are handled by recomputing and reissuing allocations.
A market avoids both problems by decoupling physical resources from entitlements: credits and prices reduce the central agent's task to setting one priority per team, and prices translate these into allocations that automatically respect relative scarcity.
But designing such a market is not straightforward: it falls outside the scope of the standard tools for computing market equilibria, an obstruction we make precise (and overcome) in the contributions below.

\itparagraph{From simple problems to more complex.}
To build intuition, consider a single resource (no conversions) and a flat organization: multiple agents reporting to one central agent.
Each leaf agent $i$ is endowed with a budget $B_i$ of artificial currency, representing both its purchasing power and its intended share of resources.
The central agent sets a fixed unit price $p$, and each agent buys as much as its budget allows: too low a price causes excess demand, too high leaves supply unallocated.
The central agent therefore seeks the unique market-clearing price at which all resources are exactly consumed.

We now generalize to an arbitrary directed rooted tree of $n$ agents, where each agent has a local supply in addition to the resources it receives from its parent.
To generalize fairness, the resource each agent receives from its parent should be proportional to the total budget of the subtree rooted at that agent.
The solution is similar---each parent $i$ sets a unit price $p_i$---but the decision is harder: each node must balance its own usage against what it sends to its children, drawing on both its local supply and its parent's resources, and on the extra budget passed up from its children.

The most general model assumes $m$ different resources.
Were they independent, we could split the problem into $m$ separate instances.
Instead, each agent can convert one type of resource into another, either to meet its own demand or to sell to its children.
Thus, each agent $i$ picks a price $p_i^r$ for every resource $r$, balancing the transfers of $r$ with other agents against conversions of $r$ into other resources.

One last complication is that an agent's interest in a resource is not unlimited.
Beyond the budget $B_i^r$ that encodes its priority for resource $r$, each agent $i$ may also have a cap $d_i^r$ on how much of $r$ it actually wants.
This leaves us with two intertwined tasks.
First, we must articulate a set of natural axioms that pin down what the right solution should even be.
Second, we must efficiently compute an allocation that reconciles all of these budgets, caps, conversions, and hierarchical constraints.

\subsection{Contributions}

We summarize our main results below.
Because their proofs are cleanest in the language of the formal model, we defer a detailed overview of our techniques to \cref{ssec:technical_overview}, immediately after the model is introduced in \cref{sec:model}.

\bfparagraph{Problem formulation.}
Our first contribution is formalizing a set of conditions that imply that a proposed collection of prices, resource flows between agents, and conversions between resource types is a good solution.
We have three sets of conditions:
\begin{enumerate}
    \item \textbf{Flow Conservation}: For every agent-resource pair, the resources coming in (local supply, flow from the parent, conversions to this resource) should be equal to the amount coming out (local demand, flows to children, conversions from this resource). Resources in can be more than resources out, only if the corresponding price is $0$.

    \item \textbf{Money Conservation}: For every agent resource pair, the money generated (payment from local usage and payments from children) should be equal to the money given to the parent.
    
    \item \textbf{Conversion Complementarity}: 
    To avoid arbitrage (buying at a low price and selling to a higher one), if the resource conversion graph permits converting resource $r$ to resource $r'$, then (a) the price of $r$ should always weakly exceed that of $r'$, and, (b) conversion from $r$ to $r'$ happens only if their prices are equal.
    For example, an agent flexible about where a resource is sourced can draw on supply from any region.
    Rule (a) says the price of the resource at any specific region weakly exceeds that of the flexible variant (whose price is simply the lowest across all regions), while rule (b) says a flexible agent buys only from the regions offering the lowest price.
\end{enumerate}

A set of prices, flows, and conversions that satisfy the above three conditions will be referred to as a \textit{feasible solution} henceforth.
We formally motivate and define these in \cref{sec:model}.
Our solution concept is related to, but technically very distinct from, classical Arrow-Debreu general equilibrium with production~\cite{arrow1954existence}.
The two key departures are: (i)~\emph{zero-profit intermediation}: unlike the profit-maximizing firms of Arrow-Debreu, nodes retain no surplus---every credit a node collects is passed up to its parent; and (ii)~\emph{non-fungible budgets}: unlike a consumer's single, freely-spendable wealth, each agent's currency is earmarked per resource and cannot be moved across these (agent, resource) accounts.
These two features together appear to place our solution outside the reach of a single global Eisenberg-Gale-type convex program (such as the one that captures Fisher markets~\cite{eisenberg1959consensus}).
Providing an efficient solution despite this difficulty is the primary source of technical novelty in our work.

\bfparagraph{Existence and computation of solutions: the harmonic demand case.}
Our next contribution is to show that feasible solutions exist and can be computed.
First, in \cref{sec:solution}, we examine the above questions in the setting where the agents are allocated a certain budget for every resource and wish to get as many resources as possible given that budget.
This means that the demand of agent $i$ for resource $r$ as a function of the price $p_i^r$ is harmonic $D_i^r(p) = B_i^r / p$, where $B_i^r$ is the budget the agent $i$ has allocated for resource $r$.
In \cref{thm:harmonic:main}, we show how to compute a feasible solution by solving $n$ convex programs, starting from the root of the hierarchy and recursively moving towards the leaves.
We also prove that solutions are \textit{effectively unique}: while there might be multiple prices that satisfy our axioms, in every such solution the agents' allocations are the same (see \cref{def:model:uniqueness} for the formal definition).

\bfparagraph{Existence and computation of solutions: the general demand case.}
Next, in \cref{sec:capped}, we examine a more challenging setting, with more general demands.
Our main practical motivation is an agent $i$ that not only sets a budget for each resource $r$, but also a limit for the amount she cares for, $d_i^r$.
This makes her demand as a function of the price to be $D_i^r(p) = \min\{ d_i^r, B_i^r / p \}$.
However, our results in \cref{sec:capped} extend beyond this function form.
We consider any demand functions with the following properties (\cref{assumption:demands}):
For all $i,r$, $D_i^r(p)$ is non-increasing and $p\cdot D_i^r(p)$ is non-decreasing and bounded.
This setting generalizes the previous one and captures more realistic scenarios where the agent wants to allocate different payments for different prices, as long as these payments increase with the price $p$.
This function form captures the earlier, simpler form, $\min\{ d_i^r, B_i^r / p \}$, which is our focus in our simulation section, \cref{sec:simulations}, in a market used by Google to allocate AI accelerators.

To solve the more general demand case, we build on our previous results for harmonic demand.
First, we show that a feasible solution for the more general problem corresponds to a feasible solution to a harmonic problem.
In particular, we show that any solution to a non-harmonic instance of our problem with demands $\qty\big( D_i^r(\cdot) )_{i,r}$ corresponds to a solution to the harmonic problem where agents have budgets $\qty\big( \hat B_i^r )_{i,r}$ for some optimal ``virtual budgets'' $\qty\big( \hat B_i^r )_{i,r}$.
Specifically, these optimal virtual budgets are such that, for all agents $i$ and resources $r$:
\begin{equation}\label{eq:virtual_budget_condition}
    \frac{\hat B_i^r}{p_i^r}
    =
    D_i^r\qty\big( p_i^r )
\end{equation}
where $p_i^r$ are the prices of the corresponding harmonic problem with budgets $\qty\big( \hat B_i^r )_{i,r}$. In \cref{thm: budget bound}, we show that any prices satisfying \cref{eq:virtual_budget_condition} are feasible prices for the non-harmonic instance. This reduces the problem to finding a set of optimal ``virtual budgets''.

To calculate and prove the existence of these \textit{virtual budgets} $\hat B = \qty\big(\hat B_i^r)_{i,r}$, we employ the following fixed-point iterative process, which we call the \textit{Budget Descent Algorithm} (\cref{algo:BDA}):
\begin{enumerate}
    \item Initially, we overestimate the virtual budgets by setting them to the maximum possible payments: $\tilde B_i^r[0] = \lim_{p \to \infty} p \cdot D_i^r(p)$ for all $i, r$.
    
    \item For every iteration $t$, we calculate the prices $p(\tilde B[t])$ of the harmonic problem with budgets $\tilde B[t]$ (via the convex problem of \cref{sec:solution}).
    Then, we re-estimate the virtual budgets by setting them equal to the payments under the prices $p(\tilde B[t])$ and the real demands $D$: $\tilde B_i^r[t+1] = p(\tilde B[t]) \, D_i^r\qty\big( p(\tilde B[t]) )$.
\end{enumerate}

We show that the virtual budgets converge to the desired condition of \cref{eq:virtual_budget_condition} (\cref{thm: main correctness thm}).
This also proves the existence of virtual budgets that satisfy \cref{eq:virtual_budget_condition}, along with the fact that there always exist feasible prices for demand functions $D$ where $D(p)$ is non-increasing and $p\,D(p)$ is non-decreasing and bounded.
Finally, we show that any feasible solution is effectively unique, i.e., each agent receives the same amount of resource in any feasible solution.

\bfparagraph{Computational efficiency.}
In the harmonic setting of \cref{sec:solution}, the convex program formulation directly implies a polynomial-time algorithm.
For the general demand setting, we measure the error of a candidate allocation by its total absolute deviation, summed over all agent-resource pairs, from the (effectively unique) feasible allocation.
We show that driving this error below $\epsilon$ (up to an instance-dependent constant) takes $O\left(\frac{1}{\epsilon}\log\left(\frac{1}{\epsilon}\right)\right)$ iterations of the Budget Descent Algorithm.

\bfparagraph{Simulations}
In \cref{sec:simulations}, we run our algorithm from \cref{sec:capped} on real-world data from a hierarchical Google market for allocating AI accelerators, where each team declares a maximum demand and budget.
We run our Budget Descent Algorithm on $4$ instances of this market and show convergence in three ways: \cref{eq:virtual_budget_condition} is increasingly satisfied, over-allocation from demand overestimation converges to $0$, and the payment mismatch from infeasibility also converges to $0$.

\subsection{Related Work}
\label{sec:related_work}

Our work builds upon a rich literature applying tools from market design to resource allocation in distributed systems.
\citet{shneidman2005markets} explored this direction early on, arguing that while market-based mechanisms offered theoretical promise for distributed resource allocation, significant practical and technical challenges hindered their immediate adoption at the time.

However, as the scale and complexity of distributed systems grew, these market-based ideas gained substantial traction.
\citet{stokely2009market} demonstrated the viability of this approach by implementing a market economy to provision compute resources across Google's planet-wide clusters, effectively using pricing to balance supply and demand.
Concurrently, the focus on fairness in multi-tenant clusters led to the development of Dominant Resource Fairness (DRF) by \citet{ghodsi2011dominant}, which generalized max-min fairness to multiple resource types and became a standard for fair allocation in cluster schedulers.
\citet{vuppalapati2023karma, fikioris2024incentives} introduced Karma, a credit-based mechanism designed to handle dynamic demands, further bridging the gap between theoretical market principles and the practical requirements of modern production environments.
Most closely related to our application, \citet{sivanquota} very recently deployed a market that allocates AI training accelerators by pricing supply and demand, letting teams express the value of their workloads through spending, focusing on the practical implementation details.
They provide theoretical results only in the case where resource types are independent; in the absence of conversions, the market decouples into $m$ separate single-resource allocations\footnote{In this case, the prices of each market can be found by simple nested binary search algorithms.}, the elementary case of our development.
Their algorithm for the general case is a simple greedy algorithm with no theoretical guarantees.
Other works that study simple credit-based economies are \cite{elokda2024self,elokda2025carma}.
The use of credit-based systems for resource allocation is at this point both a reality in practice and a well-developed research agenda.

We contribute to this line of work by adding two features to the model: resource conversions and hierarchical allocation.
We build on top of the classic Eisenberg-Gale convex program, which can centrally compute equilibria in standard Fisher markets~\cite{eisenberg1959consensus, devanur2008market}.
Resource conversions in our model are related to the conversion of inputs into outputs subject to technological constraints, a feature present in classical Arrow-Debreu general equilibrium models with production~\cite{arrow1954existence}.

Our model, however, differs in that our agents do not seek to maximize profit.
In standard general equilibrium theory, firms choose production plans to maximize the difference between revenue and cost.
In contrast, agents in our hierarchy act as zero-profit intermediaries subject to strict financial autonomy: the virtual credits they collect from downstream agents must exactly equal the credits paid to upstream agents.
Furthermore, our economy operates entirely on artificial currency, similar to the ``pseudomarkets'' of Hylland and Zeckhauser~\cite{hylland1979efficient}, where budgets are exogenous endowments rather than wealth derived from profit shares.

The physical flow of resources and budgets down the organizational chart shares structural similarities with trading networks and supply chain models \cite{ostrovsky2008stability, hatfield2013stability}.
However, the trading network literature typically assumes quasi-linear utilities and real currency, focusing on discrete matching, bilateral contracts, or profit-maximizing intermediaries.
In contrast, our model operates purely on artificial currency and enforces a strict financial autonomy axiom, where internal agents act as sovereign, zero-profit proxies on behalf of their subtrees.

More distantly related, recent work from the Algorithmic Game Theory community studies repeated resource allocation for indivisible resources, e.g., \cite{lin_et_al, gorokh2021remarkable, banerjee2023robust}.
Their main technical tool is allocating every agent a specific amount of virtual currency and then running repeated auctions where the agents bid every round for the right to use the resource.
Their main focus is fairness across all agents, proving that each agent receives a certain amount of value, either in equilibrium or against adversarial environments.

\section{Model of the Hierarchical Market with Convertible Resources}
\label{sec:model}

\bfparagraph{Complementarity Notation.}
We use the $\bot$ symbol to denote complementary inequalities.
Two inequalities are said to be complementary if at least one of them must hold with equality: $(a \geq 0) \;\bot\; (b \geq 0)$ denotes that $a \geq 0$, $b \geq 0$, and $ab=0$.

\bfparagraph{Instance Description.}
An \emph{instance} of our problem consists of $n$ agents (indexed by $i$) and $m$ resource types (indexed by $r$), both of which are distributed along an organizational graph.
Agents are endowed with artificial currency (credits) that can be used to purchase resources in an internal market economy.
Instead of modeling the credit endowments directly, we model them indirectly via demand functions that specify how many units of each resource an agent wants to buy given the local market prices.
An instance is formally specified by the following components: 

\begin{enumerate}
    \item {\bf Organizational tree:} A directed rooted tree $\agentTree$ with nodes representing agents in $[n]$. The structure of this tree specifies the parent-child relationships between the agents. Specifically, a (directed) edge $(i, j) \in \agentTree$ implies that the resources available to $i$ are also available to $j$.
    For an agent $i$, let $\parent(i)$ denote the unique parent of $i$ in the tree $\agentTree$. Furthermore, let $\ch(i)$ denote the set of children of $i$ in $\agentTree$. Finally, let $\desc(i)$ be the set of all strict descendants of $i$ in $\agentTree$.

    \item {\bf Resource-conversion graph:} A weakly connected directed graph $\convGraph$ with nodes representing resources in $[m]$. The structure of this graph represents the set of feasible conversions, i.e., a conversion of resource $r_1$ to resource $r_2$ is feasible if and only if $(r_1, r_2)$ is an edge in $\convGraph$.
    While in many practical applications this graph is naturally a directed acyclic graph (DAG), e.g., representing irreversible hardware downgrades, our model and solution concept naturally accommodate cycles; any such cycles simply resolve into price equivalence classes in any feasible solution (\cref{lem:dag}).
    For any $r \in [m]$, let $\inEdges(r) = \{r' \in [m]: (r', r) \in \convGraph \}$ and $\outEdges(r)= \{r' \in [m]: (r, r') \in \convGraph \}$ denote the set of resources $r'$ adjacent to $r$ and adjacent from $r$ respectively.

    \item {\bf Resource Supply:} A set of values $s_i^r \geq 0$ for $i \in [n]$ and $r \in [m]$ where $s_i^r$ represents the initial, exogenous supply of resource $r$ available directly to agent $i$.
    
    \item {\bf Demands:} A set of demand functions%
    \footnote{\emph{Remark on Separable Demand:} We define independent demand functions for each resource, implicitly assuming that agents' preferences are separable. While a fully combinatorial demand model---where an agent maps a global price vector $\vec{p}_i$ to a bundle of requested resources---would be a natural extension, our separable formulation is both analytically tractable and practically expressive. In our framework, the substitutability between resources is elegantly captured on the \emph{supply side} via the resource-conversion graph, rather than on the demand side. For instance, if an agent is indifferent between two specific resources, the agent can express demand for a ``generic'' resource, and the conversion graph allows either specific resource to fulfill this demand. This mirrors strict corporate budgeting silos where funds are earmarked for specific categories, while system-level technological constraints govern substitutions.}
    $D^r_{i}: (0, \infty) \rightarrow \R_+$ for $i \in [n]$ and $r \in [m]$ where $D_{i}^r(p)$ represents the amount of demand of agent $i$ for resource $r$ at price $p$. 
    The most general demand model we examine is in \cref{sec:capped}, where $D_i^r(p)$ is weakly decreasing, while $p \, D_i^r(p)$ is weakly increasing and bounded (\cref{assumption:demands}).
\end{enumerate}

\bfparagraph{Economic Principles and Solution Concept.}
Our goal is to define a solution concept that achieves efficiency (clearing available supply) and fairness (allocating proportionally to credit endowments) within the constraints of the organizational hierarchy. To make the economic intuition rigorous, our solution concept generalizes a Walrasian equilibrium, guided by three standard microeconomic principles: \emph{Local Market Clearing}, \emph{Zero-Profit Intermediation} (or \emph{Financial Autonomy}), and \emph{No-Arbitrage}. 
To motivate these principles, it is instructive to consider a few simple examples.

\begin{figure}[htbp]
    \centering
    \begin{minipage}[b]{0.33\textwidth}%
        \centering
        \begin{tikzpicture}[
            level distance=1cm,
            sibling distance=1cm,
            dot/.style={circle, fill=black, inner sep=1.5pt},
            edge from parent/.style={draw, ->, >=latex},
            every label/.style={font=\small}
          ]
          \node [dot, label=above:$s$] {}
            child { node [dot, label=below:$D_1(p)$] {} }
            child { node [dot, label=below:$D_2(p)$] {} }
            child { node {\dots} edge from parent[draw=none] }
            child { node [dot, label=below:$D_n(p)$] {} };
        \end{tikzpicture}
        \caption{Ex. 1, Centralized Market.}
        \Description{A star-shaped tree: a root node holding all supply $s$, with $n$ leaf children whose demands are $D_1(p),\dots,D_n(p)$.}
        \label{fig:example1}
    \end{minipage}%
    \begin{minipage}[b]{0.33\textwidth}%
        \centering
        \begin{tikzpicture}[
            level 1/.style={level distance=.7cm, sibling distance=1.8cm},
            level 2/.style={level distance=.7cm, sibling distance=1.8cm},
            filled dot/.style={circle, fill=black, inner sep=1.5pt},
            hollow dot/.style={circle, draw=black, fill=white, inner sep=1.5pt},
            edge from parent/.style={draw, ->, >=latex},
            every label/.style={font=\small}
          ]
          \node [filled dot, label=above:{$s=10$}] {}
            child { node [filled dot, label=below:{$D_1(p) = \frac{2}{p}\qquad$}] {} }
            child { node [filled dot, label=right:{$s=6$}] {}
                child { node [filled dot, label=below:{$D_2(p) = \frac{2}{p}$}] {} }
                child { node [filled dot, label=below:{$D_3(p) = \frac{1}{p}$}] {} }
            };
        \end{tikzpicture}
        \caption{Ex. 2, Hierarchical budgets}
        \Description{A two-level tree: a root with supply $10$ whose children are a leaf with demand $2/p$ and an internal node with exclusive supply $6$; that internal node has two leaves with demands $2/p$ and $1/p$.}
        \label{fig:example2}
    \end{minipage}%
    \hfill
    \begin{minipage}[b]{0.29\textwidth}%
        \raggedleft
        \begin{tikzpicture}[
            grow=left,
            edge from parent/.style={draw, <-, >=latex},
            level 1/.style={level distance=1.5cm, sibling distance=1.5cm},
            filled dot/.style={circle, fill=black, inner sep=1.5pt},
            hollow dot/.style={circle, draw=black, fill=white, inner sep=1.5pt},
            every label/.style={font=\small}
          ]
          \node [filled dot, label=above:{$\qquad D^G(p) = \frac{10}{p}\qquad$}] {}
            child { node [filled dot, label=above:{$D^A(p) = \frac{10}{p}\qquad$}, label=below:{$s^A = 10\qquad$}] {} }
            child { node [filled dot, label=below:{$D^B(p) = \frac{30}{p}\qquad$}, label=above:{$s^B = 10\qquad$}] {} };
        \end{tikzpicture}
        \caption{Ex. 3, Resource Conversions}
        \Description{A single agent demanding a generic resource $G$, drawing on supplies of $10$ units each of resources $A$ and $B$, both of which can convert into $G$.}
        \label{fig:example3}
    \end{minipage}%
\end{figure}

\itparagraph{Example 1: Centralized Market Clearing.}
In the first example (\cref{fig:example1}), there is a single resource type ($m = 1$) and only the root has supply.
All the demand is at the leaves, representing a structure with all the supply centralized at the root.
Walrasian supply and demand matching dictates that the market clears when setting a global clearing price $p$ such that $s = \sum_i D_i(p)$, or that $s > \sum_i D_i(0)$ (supply exceeds maximum demand).
In this case, all nodes face the same uniform price $p$, and exactly $D_i(p)$ units of resource flow from the root to leaf $i$.

\itparagraph{Example 2: Hierarchical Budgets and Financial Autonomy.}
Consider now a slightly more complex example with a single resource ($m=1$) distributed across two hierarchical levels (\cref{fig:example2}).
There are $10$ units of resource at the root that can be consumed by all agents in the tree, and $6$ units of resource that are exclusive to the right subtree.
The leaves have endowments $B_i$ of $2$, $2$, and $1$ credits, respectively.
Assuming they have no cap on what they want to acquire, their demands take the harmonic form $D_i(p) = B_i/p$.

Because budgets are strictly siloed, we can view the right subtree as acting as an autonomous market intermediary representing the aggregate interests of its descendants.
The right subtree sets an internal local price $p_R$, which allows it to collect a virtual revenue of $p_R \cdot \qty\big( D_2(p_R) + D_3(p_R) ) = 2 + 1 = 3$ credits from its leaves---note that due to the demand functions being harmonic, this virtual revenue is independent of $p_R$.
The subtree acts as a zero-profit pass-through: it uses exactly these $3$ collected credits as its budget to acquire supplementary supply from the root.

If the root sets a price $p_0$, Agent $1$ buys $D_1(p_0) = 2/p_0$ resources, while the right subtree uses its $3$ collected credits to buy $3/p_0$.
Clearing the market at the root requires $10 = 2/p_0 + 3/p_0$, which yields $p_0 = 1/2$.
The right subtree, having purchased $3 / (1/2) = 6$ units from the root, can now pool this with its exclusive local supply of $6$ units, allocating $6+6=12$ total units to its children.
Clearing this local market requires $3/p_R = 12$, yielding $p_R = 1/4$.
The agents are thus allocated $4, 8$, and $4$ units respectively.
Note that prices strictly decrease along the tree ($1/2 \to 1/4$) because downstream nodes enjoy exclusive access to local supply.
 
\bfparagraph{Axiomatizing the Solution for a Single Resource.} 
The preceding examples illustrate the first two core axioms of our solution concept.
For a single resource $r$, a solution consists of a local internal price $p_i^r$ for each node $i \in [n]$ together with an upstream-to-downstream physical flow of resources $f_i^r$ corresponding to the transfer $\parent(i) \rightarrow i$.
This is subject to two economic conditions:

The first axiom is \textbf{Local Market Clearing} (or \textbf{Flow Conservation}), which dictates that every node acts as a localized Walrasian market.
Its total available supply (internal supply plus resources acquired from its parent) must equal its total allocation (its own consumption plus the resources sent to its children):
$$D_i^r(p_i^r) + \sum_{j \in \ch(i)} f_{j}^r  = s_i^r + f_{i}^r.$$
This equality must hold whenever the resource is scarce and has strict economic value ($p_i^r > 0$).
In the event of strict oversupply, the price falls to zero ($p_i^r = 0$), allowing supply to exceed demand (i.e., free disposal).
We will summarize this standard condition using complementarity notation.

The second axiom is \textbf{Zero-Profit Intermediation} (or \textbf{Money Conservation}).
Because currency is artificial and intermediating agents derive no utility from hoarding it, each agent operates at zero profit, preserving the financial autonomy of their subtree.
The virtual revenue an agent collects from its subtree at its local price must exactly equal the virtual payment it makes to its parent at the parent's price.
For any node $i$, clearing its local market at price $p_i^r$ generates a total revenue of $p_i^r\qty\big( D_i^r(p_i^r) + \sum_{j \in \ch(i)} f_{j}^r )$, where  $p_i^r \cdot D_i^r(p_i^r)$ is the credit endowment from node $(i, r)$ to fulfill its demand $D_i^r(p_i^r)$.
To purchase its imported flow $f_i^r$, it must pay the parent's rate $p_{\parent(i)}^r$, leading to the condition:
$$p_i^r\;\left(D_i^r(p_i^r) + \sum_{j \in \ch(i)} f_{j}^r\right) = p_{\parent(i)}^r\;f_{i}^r.$$
In instances where there is only one resource, the above two axioms fully characterize our market (or fully characterize $m$ independent markets if there are $m$ resources with no conversions).

\itparagraph{Example 3: Resource Conversions and No-Arbitrage.} 
While the hierarchical structure is an interesting market feature to analyze, an important defining characteristic of many practical applications (such as the Google internal market for AI accelerators) is the ability to convert resources.
Let us consider a simple one-agent topology as in Figure \ref{fig:example3} with three resource types $A, B, G$ where $G$ stands for a ``generic'' resource.
There is an initial supply of $s^A = 10$ and $s^B = 10$.
While there is no supply of $G$, we can convert $A \rightarrow G$ or $B \rightarrow G$.
Now, consider the demand functions for $A$, $B$ and $G$:
$$D_1^A(p_1^A) = \frac{10}{p_1^A} \qquad
D_1^B(p_1^B) = \frac{30}{p_1^B} \qquad
D_1^G(p_1^G) = \frac{10}{p_1^G}$$
Without the generic demand, resources $A$ and $B$ would be priced independently at $1$ and $3$ respectively.
However, the generic demand introduces a substitution effect: the system must decide whether to satisfy it by converting resource $A$ or resource $B$.
An optimizing agent would request the substitute with the lowest price---in this case, $A$.
This spatial arbitrage increases demand for $A$, driving its price up until marginal values equalize: $p^A = p^G = 2 < p^B = 3$.
This leads to the feasible allocation: $D_1^A(2) = 5, D_1^B(3) = 10, D_1^G(2) = 5$, achieved by converting $5$ units of $A$ to $G$.

This highlights the \textit{No-Arbitrage} principle governing resource conversions, which imposes two desirable properties on feasible prices: (i) If we can convert $A \rightarrow G$, then we must have $p_i^G \leq p_i^A$; otherwise, an arbitrageur could generate infinite virtual wealth by converting the cheaper $A$ into the strictly more expensive $G$ (or equivalently, demand for $A$ would perfectly substitute to $G$). (ii) If the conversion $A \rightarrow G$ is actively utilized, then the marginal cost must equal to the marginal benefit, meaning the prices must perfectly equalize ($p_A = p_G$). This observation forms the basis of our conversion complementarity condition.

\subsection{Solution Concept with Conversions} \label{ssec:model:full_model}

Consolidating these economic principles, in its most general version, a solution consists of the following three variables:
\begin{enumerate}
    \item A set of flow values for edge $(\parent(i), i) \in \agentTree$ and resource $r \in [m]$, denoted $f_i^r \geq 0$.
    \item A set of local clearing prices for each agent $i \in [n]$ and resource $r \in [m]$, denoted $p_{i}^r \geq 0$.
    \item A set of conversion amounts for each edge $(r_1, r_2) \in \convGraph$ and agent $i \in [n]$, denoted $c_i^{(r_1, r_2)} \geq 0$.
\end{enumerate}

We say a solution is \emph{feasible} (equivalently, an \emph{equilibrium}) if it satisfies the following complementary economic axioms:
\begin{enumerate}
    \item \textbf{(Flow Conservation / Local Market Clearing)} For all agent and resource pairs $(i, r)$:
    \begin{equation}\label{eq:model:flow}
        D_i^r(p_i^r) + \sum_{j \in \ch(i)} f_{j}^r + \sum_{r' \in \outEdges(r)} c_i^{(r, r')} \le s_i^r + f_{i}^r + \sum_{r' \in \inEdges(r)} c_i^{(r', r)} \quad\quad\bot\quad\quad p_i^r \ge 0    
    \end{equation}
    where we assume $f_{i}^r = 0$ if $i$ is the root node.
    
    \item  \textbf{(Money Conservation / Zero-Profit Intermediation)} For all agent and resource pairs $(i, r)$ where $i$ is not the root:
    \begin{equation}\label{eq:model:money}
        p_i^r\;\left(D_i^r(p_i^r) + \sum_{j \in \ch(i)} f_{j}^r\right) = p_{\parent(i)}^r\;f_{i}^r.        
    \end{equation}

    \item \textbf{(Conversion Complementarity / No-Arbitrage)} For agents $i$ and resources $r, r'$ such that $(r, r') \in \convGraph$:
    \begin{equation}\label{eq:model:conversion}
        p_i^{r} \geq p_i^{r'} \quad\quad\bot\quad\quad c_i^{(r, r')} \geq 0.
    \end{equation}
\end{enumerate}

There are examples where there are multiple feasible solutions.
Such a simple example is when there is a unique agent and resource with supply $s = 1$ and demand function $D(p) = \min\{1, 1/p\}$; any price $p \le 1$ is feasible, but all feasible prices induce the same allocation.
Another example is two agents and one resource, where both agents have demands $D_1 = D_2 = 0$; there are feasible solutions with zero prices and positive flow, or solutions with zero flows and positive prices.
While these examples are pathological, there are multiple ways to induce non-uniqueness.
However, in both these examples, the resource allocations are the same.
We introduce the following definition of \textit{effective uniqueness} that formalizes and generalizes this idea.

\begin{definition}[Effective uniqueness] \label{def:model:uniqueness}
    Fix an instance $(\agentTree, \convGraph, D, s)$.
    We call a solution $(f, p, c)$ \textit{effectively unique} if for any other solution $(\hat f, \hat p, \hat c)$ it holds that $D_i^r(p_i^r) = D_i^r(\hat p_i^r)$ for all $i, r$.
\end{definition}

\subsection{Technical Overview} \label{ssec:technical_overview}

Having defined the model, we now expand the high-level contributions from the introduction into a full technical overview of our approach.

\bfparagraph{Departures from Arrow-Debreu and Eisenberg-Gale.}
Our setting departs from classical general-equilibrium theory in two fundamental ways.
First, intermediate agents are zero-profit pass-throughs rather than profit-maximizing firms: they collect revenue from their subtrees and spend it entirely on resources from their parent, with no residual profit motive.
Second, all currency is artificial and strictly siloed by the org chart, so wealth cannot flow freely across organizational boundaries.
These departures also break the standard algorithmic toolkit.
For classical Fisher markets, equilibrium prices can be found by a \emph{single} global convex program---the celebrated Eisenberg-Gale program---but our hierarchy resists any such joint formulation: each node's optimization depends on its parent's price, which is itself the output of another optimization.
We are therefore forced to solve $n$ convex programs sequentially, top-down along the tree, and, for general demands, to wrap this pipeline as the inner step of an iterative outer loop (our Budget Descent Algorithm, below).
Proving correctness of an iterative algorithm whose inner step is already a non-trivial sequential computation over the tree is what makes our analysis technically demanding, and motivates the solution concept and algorithms we develop next.

\bfparagraph{Structural Reductions.}
Before computing feasible solutions, we establish three structural reductions that simplify the problem significantly.
First, any cyclic resource conversion graph reduces to a DAG: we show that strongly connected components collapse into equivalence classes of equal-priced resources, with no loss of generality (\cref{lem:dag}).
Second, agent-resource pairs that collectively have enough supply to satisfy maximum demand---which we call \emph{degenerate subgraphs}---decouple entirely from the rest of the market.
We show these can be identified and removed efficiently via a series of min-cut computations (\cref{thm:degenerate_and_prices,thm:computation_of_degenerate}), leaving an instance with strictly positive prices where maximum demand cannot be met.
Third, under positive prices the Zero-Profit Intermediation axiom allows us to eliminate the flow variables entirely: each flow $f_i^r$ can be written in closed form as a function of prices alone (\cref{lem: problem reformulation}).
Together, these reductions leave a clean system in prices and conversions only.

\bfparagraph{Harmonic Demands: $n$ Independent Convex Programs.}
For harmonic demands $D_i^r(p) = B_i^r/p$, the critical observation is that the payment $p_i^r D_i^r(p_i^r) = B_i^r$ is a constant, independent of prices.
This makes the Modified Flow Conservation equation (\cref{eq:harmonic:flow}) for agent $i$ depend only on $i$'s own prices and its parent's prices---not on any descendant's prices.
Prices can therefore be computed top-down, one agent at a time, by sequentially solving $n$ convex programs where each depends only on the previous ones (\cref{eq: optimization problem}).
Each program maximizes a separable concave objective over the no-arbitrage constraints $p_i^r \geq p_i^{r'}$ for $(r, r') \in \convGraph$. Via the KKT conditions, we observe that the optimal prices and the optimal conversion amounts correspond to optimal primal and dual solutions respectively (\cref{thm:harmonic:main}).
Effective uniqueness of the solution follows from a careful characterization of when the concave objective fails to be strongly concave, showing that non-uniqueness in prices never translates to non-uniqueness in allocations (\cref{cor:harmonic:uniqueness}).

\bfparagraph{General Demands: Virtual Budget Reduction.}
More general demand functions, and even capped harmonic demands $D_i^r(p) = \min\{B_i^r/p,\, d_i^r\}$, break the constant-payment property, invalidating the top-down decomposition.
Our key insight is a reduction back to harmonic.
Let $\hat B = (\hat B_i^r)_{i,r}$ be a vector of \emph{virtual budgets} and $p_i^r(\hat B)$ be the prices of a feasible solution of the problem when replacing demand function $D_i^r(p)$ with the harmonic $\hat B_i^r / p$.
If $\hat{B}$ satisfies the fixed-point condition $\hat{B}_i^r / p_i^r(\hat{B}) = D_i^r(p_i^r(\hat{B}))$ for all $(i,r)$, then the harmonic solution at $\hat{B}$ is also a feasible solution for the general demand problem (\cref{thm: budget bound}).
Crucially, the converse also holds: every feasible solution of the general problem arises this way.
This bijection between fixed points and feasible solutions reduces the existence question entirely to finding a fixed point of a well-structured operator (as we show next).

\bfparagraph{Budget Descent Algorithm.}
We find such a fixed point via the Budget Descent Algorithm (\cref{algo:BDA}).
Starting from the overestimate $\tilde{B}_i^r[0] = \lim_{p\to\infty} p\,D_i^r(p)$, each iteration solves the harmonic subproblem and updates $\tilde{B}_i^r[t+1] = p_i^r(\tilde{B}_i^r[t])\,D_i^r(p_i^r(\tilde{B}_i^r[t]))$.
The algorithm is parameter-free, requiring no step-size or learning-rate tuning, and maintains throughout the invariant that it always overestimates the virtual budgets.
Establishing convergence requires two independent technical results.

The first is a \emph{monotonicity} result: if virtual budgets decrease componentwise, the resulting harmonic prices also decrease (\cref{lem: key structural lemma}).
Although intuitive---less money in the system should lower prices---the proof is subtle because prices are defined by a recursive chain of $n$ coupled convex programs, one per agent, where each child's program takes its parent's optimal price as an input parameter.
A price decrease at the root must therefore propagate monotonically all the way to the leaves through these dependencies.
We establish this via a careful sequence of applications of Topkis's Theorem (\cref{thm:topkis}), which requires verifying supermodularity of each local objective in the relevant parameters---a non-trivial task given the coupling introduced by the no-arbitrage constraints across resources.
This monotonicity validates the overestimation invariant and implies that the virtual budgets decrease across iterations.

The second is a \emph{continuity} result: the price map $\hat{B} \mapsto p(\hat{B})$ is continuous (\cref{lem: key structural lemma; price continuity}).
We prove this via Berge's Maximum Theorem (\cref{thm:berge's}), again applied recursively from root to leaves.
The key difficulty is that each node's input parameter is its parent's price, which is only known to be continuous after Berge has already been applied at the parent; the proof therefore proceeds level by level from root to leaves.
Together, monotone decrease and continuity guarantee that the virtual budgets converge to the desired fixed point, yielding a constructive, algorithmic proof that a feasible solution exists, alongside the correctness guarantee for the algorithm (\cref{thm: main correctness thm}).

\bfparagraph{Effective Uniqueness.}
We also use the Budget Descent Algorithm to prove the effective uniqueness of any feasible solution, as long as the demands satisfy \cref{assumption:demands}.
In \cref{general:uniqueness}, we prove that our Budget Descent Algorithm approximates any solution from above, i.e., the corresponding virtual budgets and prices of any feasible solution are at most the ones generated by our algorithm.
This implies that any other feasible solution differing in allocations would provide agents with strictly more resources in total than the solution found by the Budget Descent Algorithm, violating the fact that every feasible solution allocates all the available resources.

\bfparagraph{Efficient Convergence of the Budget Descent Algorithm.}
Beyond convergence in the limit, we offer bounds for the error after a finite number of steps.
For any $\epsilon > 0$, after $O\qty(\tfrac{1}{\epsilon}\log\tfrac{1}{\epsilon})$ iterations the allocation is within $\epsilon$ of the (effectively unique) feasible one.
This error is measured by the total absolute deviation in allocation across all agent-resource pairs, up to an instance-dependent constant (\cref{thm: main convergence result}).

\subsection{Reducing the problem to a simpler one}

We now apply three simplifications that reduce the problem to a more convenient form; all omitted proofs are in \cref{app:proof from sec 2}.

\bfparagraph{Turning the resource graph into a DAG.}
First, we may assume the resource graph $\convGraph$ is a DAG.
Any strongly connected component shares a single price in every feasible solution, so it can be merged into one resource (with the aggregate supply and demand) with no loss of generality.

\begin{restatable}{lemma}{LemmaDAG}\label{lem:dag}
    Fix an instance $(\agentTree, \convGraph, D, s)$.
    Consider a new instance $(\agentTree, \hat\convGraph, \hat D, \hat s)$ with $\hat m < m$ resources, where every resource $\hat r \in [\hat m]$ corresponds to a maximal strongly connected component in $\convGraph$.
    Denote with $R(\hat r)$ the resources in $\convGraph$ that correspond to the new resource $\hat r \in [\hat m]$.
    Then, for every $i \in [n]$ and $\hat r \in [\hat m]$ we define
    \begin{equation*}
        \hat s_i^{\hat r} = \sum_{r \in R(\hat r)} s_i^r
        \quad\text{ and }\quad
        \hat D_i^{\hat r}(p) = \sum_{r \in R(\hat r)} D_i^r(p)
    \end{equation*}

    Then, every feasible solution of $(\agentTree, \convGraph, D, s)$ corresponds to a feasible solution of $(\agentTree, \hat \convGraph, \hat D, \hat s)$ and vice versa.
\end{restatable}

\bfparagraph{Splitting the problem into independent sub-problems.}
Second, we may assume all prices are strictly positive.
Agent-resource pairs that can be priced at $0$ turn out to decouple from the rest of the market, so we can compute and remove them.
To make this precise, we work on the \emph{product graph} $\agentTree \times \convGraph$ on node set $[n]\times[m]$, with edges $(i, r) \to (i', r)$ whenever $i = \parent(i')$ and $(i, r) \to (i, r')$ whenever $r' \in \outEdges(r)$; a node set $\mathcal G$ is \emph{absorbing} if it has no outgoing edges.
An absorbing $\mathcal G$ is \emph{degenerate} if it has at least as much supply as demand and is minimal with this property.
Minimality is essential: a subgraph can have more aggregate supply than demand yet still fail to be self-sufficient, because a surplus at one node cannot always be routed to a deficit at another (\cref{rem:degenerate_minimality}).

\begin{definition}[Degenerate subgraph] \label{def:model:degenerate}
    An absorbing subgraph $\mathcal G \subseteq [n] \times [m]$ is called \textit{degenerate} if, $\sum_{(i, r) \in \mathcal{G}} D_i^r(0) \leq \sum_{(i, r) \in \mathcal{G}} s_i^r$ and no strict subset of $\mathcal G$ satisfies this condition\footnote{Since the demand function may not be well-defined at 0, we take $D_i^r(0)$ to be $\lim_{x \rightarrow 0^+} D_i^r(x)$ (which may be $\infty$). This exists by monotone convergence.}.
\end{definition}

A degenerate subgraph decouples from the rest of the market: setting its prices to $0$ and zeroing its flows and conversions from the outside preserves feasibility and, moreover, every agent's allocation.
Removing degenerate subgraphs iteratively thus leaves an instance whose prices are all strictly positive.

\begin{restatable}{theorem}{TheoremDegeneratePrices}\label{thm:degenerate_and_prices}
    Consider an instance of our problem $(\agentTree, \convGraph, D, s)$ and a feasible solution $(f, p, c)$.
    If no degenerate subgraph exists, then all the prices must be positive.
    On the other hand, if there exists a degenerate $\mathcal{G}$, then generating a new solution where we set to $0$ the price of the nodes in $\mathcal{G}$, along with the flows and conversions between $\mathcal G$ and $([n]\times[m]) \setminus \mathcal G$ retains feasibility and every pair $(i, r) \in \mathcal G$ gets the same allocation.
\end{restatable}

Finally, a minimal degenerate subgraph can be found (or ruled out) efficiently via a sequence of min-cut computations (\cref{sec: identify degen graph}), so degenerate subgraphs can be removed recursively until all prices are positive.

\begin{restatable}{theorem}{TheoremComputationOfDegenerate} \label{thm:computation_of_degenerate}
    We can efficiently find a minimal degenerate subgraph or decide that one does not exist.
\end{restatable}

\begin{remark}[Positive Prices] \label{remark:positive_prices}
    For the rest of this paper, we assume there is no degenerate subgraph, making all the prices positive.
    This would require us to define a more general model, where not every agent resource pair $(i, r)$ appears.
    However, to keep the model and notation more understandable, for simplicity, we will assume that the entire $\agentTree \times \convGraph$ graph is present.
\end{remark}

\bfparagraph{Removing the Flow Variables.}
Third, under positive prices the Money Conservation axiom lets us write each flow $f_i^r$ in closed form as a function of the prices: the payment a child makes to its parent for resource $r$ equals the money generated by the demands of the whole subtree $\subTree(i) = \{i\} \cup \desc(i)$, by induction along the tree (\cref{lem: isolate flow lemma}). This eliminates the flow variables and the Money Conservation axiom entirely.

\begin{lemma} \label{lem: problem reformulation}
    Fix an instance $(\agentTree, \convGraph, D, s)$. Consider a feasible solution $(f, p, c)$ for the instance with $p_i^r > 0$ for all $(i, r) \in [n] \times [m]$.
    Then, for all $(i, r) \in [n] \times [m]$, $(p, c)$ satisfies:
    \begin{equation}\label{eq: payment eq}
        \left(\frac{1}{p_i^r} - \frac{\ind{i \text{ is not root}}}{p_{\parent(i)}^r}\right)\sum_{j \in \subTree(i) } p_{j}^r\;D_j^r(p_{j}^r) = s_i^r + \sum_{r' \in \inEdges(r)} c_i^{(r', r)} - \sum_{r' \in \outEdges(r)} c_i^{(r, r')}
    \end{equation} Furthermore, for any $(p, c)$ satisfying \cref{eq: payment eq} and Conversion Complementarity, there exists unique $f$ such that $(f, p, c)$ is a feasible solution.
\end{lemma}

We have thus reduced the problem to one in the variables $(p, c)$ alone, subject to Conversion Complementarity (\cref{eq:model:conversion}) and the Modified Flow Conservation \cref{eq: payment eq}.

\section{Harmonic Demands}
\label{sec:solution}

In this section, we examine the case when agents' demands are harmonic functions, i.e., for every $i,r$, $D_i^r(p) = \frac{B_i^r}{p}$ for some $B_i^r \ge 0$.
While simple, these demands capture an agent $i$ who has allocated $B_i^r$ budget for resource $r$, and wishes to get as many resources as possible given this restriction.
In addition, solving this simpler problem will be a key step for calculating the solution for more complicated demand functions.
Specifically, in \cref{sec:capped} we use the algorithms and techniques from this section as a subroutine to solve the problem for more general demand functions.
These include the capped demands $D_i^r(p) = \min\big\{ \frac{B_i^r}{p}, d_i^r \big\}$, where an agent allocates a budget $B_i^r$ for a resource but also caps at $d_i^r$ the amount it wants.

We will show that there exist effectively unique feasible prices and that we can calculate them via a convex optimization approach. Throughout this section we assume, following \cref{remark:positive_prices}, that all prices are strictly positive; degenerate subgraphs, where prices may be zero, have already been removed in \cref{sec:model}.
Due to \cref{lem: problem reformulation} and the form of each $D_i^r(\cdot)$, we need to find prices and conversions $(p, c)$ satisfying:
\begin{enumerate}
    \item \textbf{(Modified Flow Conservation)}
    For all agent and resource pairs $(i, r)$:
    \begin{equation} \label{eq:harmonic:flow}
        \left(\frac{1}{p_i^r} - \frac{\ind{i \neq \text{root}}}{p_{\parent(i)}^r}\right)\sum_{j \in \subTree(i) } B_j^r = s_i^r + \sum_{r' \in \inEdges(r)} c_i^{(r', r)} - \sum_{r' \in \outEdges(r)} c_i^{(r, r')}.    
    \end{equation}

    \item \textbf{(Conversion Complementarity)} For agents $i$ and resources $r$ and $r'$ such that $(r, r') \in \convGraph$:
    \begin{equation} \label{eq:harmonic:conversion}
        p_i^{r} \geq p_i^{r'} \quad\quad \bot \quad\quad c_i^{(r, r')} \geq 0.
    \end{equation}
\end{enumerate}

We shall show that we can exactly find a feasible solution to this problem using a sequence of convex optimization problems.
To see this, we make the critical observation that for the root agent, \cref{eq:harmonic:flow,eq:harmonic:conversion} depend only on its prices.
For non-root agents $i$, given $p_{\parent(i)}^r$ for all $r\in [m]$ in a feasible solution, \cref{eq:harmonic:flow,eq:harmonic:conversion} do not depend on prices of any other agents.
Hence, if we find prices from the root going downwards, such that at every step \cref{eq:harmonic:flow,eq:harmonic:conversion} hold, we can fix those prices and deduce the price of the children.
The key idea is then to show that we can indeed find $(p_i, c_i)$ satisfying the \cref{eq:harmonic:flow,eq:harmonic:conversion} for agent $i$, having fixed the prices of $\parent(i)$.
Specifically, we will show that the following optimization problem induces $(p_i, c_i)$ satisfying \cref{eq:harmonic:flow,eq:harmonic:conversion}:
\begin{equation}
    \label{eq: optimization problem}
    \max_{p_i^1, \cdots, p_i^m \geq 0}  F_i(p_i)\quad \text{s.t.} \quad p_i^r \geq p_i^{r'}\;\; \forall (r, r') \in \convGraph
\end{equation}
where
\begin{equation}\label{eq: optimization objective}
    F_i(p_i) = \sum_{r \in [m]}\log(p_i^r)\sum_{j \in \subTree(i) } B_j^r- \sum_{r \in [m]}p_i^r\left(\frac{\ind{i \neq \text{root}}}{p_{\parent(i)}^r}\sum_{j \in \subTree(i) } B_j^r+ s_i^r\right)
\end{equation}
which is a concave function since budgets are non-negative.
The main theorem of this section is that any optimal solution $p$ to this concave program, along with the optimal dual variables, satisfies \cref{eq:harmonic:flow,eq:harmonic:conversion}.

\begin{theorem} \label{thm:harmonic:main}
    Fix an agent $i \in [n]$, along with prices $p_{\parent(i)}^r \in \R_+$ for all $r$ if $i$ is not the root.
    Then, given a set of prices $p_i$ solving program \eqref{eq: optimization problem}, there exist conversions $c_i$ satisfying \cref{eq:harmonic:flow,eq:harmonic:conversion} for agent $i$.
    In addition, each $p_i^r$ is unique if $\sum_{j \in \subTree(i) } B_j^r > 0$ or $s_i^r > 0$.
\end{theorem}

The first part of this result follows from a careful analysis of the KKT conditions. In doing so, we will see that the optimal conversions are exactly the dual variables from the corresponding constraints. To show uniqueness of a solution, we analyze when the optimization objective is strongly concave, and characterize the form it takes when it is not.

\begin{corollary} \label{cor:harmonic:uniqueness}
    In the harmonic demand setting, every feasible solution $(f, p, c)$ is effectively unique (\cref{def:model:uniqueness}): for every $i, r$, either $B_i^r > 0$ which makes $p_i^r$ and $D_i^r(p_i^r)$ unique, or $B_i^r = 0$ which makes $D_i^r(p) = 0$ for all $p$.
\end{corollary}

\section{General Demands}
\label{sec:capped}
\label{subsec: proof of structural result}

In this section, we consider the general setting in which demand functions are not necessarily harmonic.
Specifically, we consider demand functions that satisfy the following conditions.

\begin{assumption} \label{assumption:demands}
    For every $i \in [n]$ and $r \in [m]$ we assume that the demand function $D_i^r(\cdot)$ is such that:
    (a) $D_i^r(p)$ is weakly decreasing in $p$,
    (b) $p \cdot D_i^r(p)$ is weakly increasing in $p$, and 
    (c) $p \cdot D_i^r(p)$ is bounded, i.e., $\lim_{p \to \infty} p \cdot D_i^r(p) < \infty$.
\end{assumption}

The above conditions generalize the harmonic setting to demand functions.
Specifically, we allow demands that are weakly decreasing (similar to the harmonic case), but we also require that the payment generated ($p \cdot D(p)$) is weakly increasing.
This captures the natural condition that the higher the price, the more the node's payment will be.
Another way to explain the last condition is that $D(p)$ cannot be decreasing too fast, e.g., it cannot be that $D(p) \propto 1/p^2$ around some price $p$.
Our third condition simply states that while the payment increases with the price, this cannot be unbounded, capturing that no player has infinite budget.

\cref{assumption:demands} captures the natural demand model we study in \cref{sec:simulations} where the demand of $(i, r)$ is described by two parameters, a budget $B_i^r$ and a maximum demand $d_i^r$: $D_i^r(p) = \min\big\{ B_i^r/p, d_i^r \big\}$.
These capped harmonic demands capture the interplay between agents wanting as much of each resource as their budget will allow, up to a maximum amount that they are able to gain value from.

The general demand setting possesses additional technical difficulties compared to the harmonic setting.
Specifically, in the harmonic setting, we have that $\sum_{j \in \desc(i) } p_{j}^r\;D_j^r(p_{j}^r)$ is a constant independent of the agents' prices.
In this general setting, this fails to hold in general.
Thus, our technique of solving a sequence of independent optimization problems from the root down no longer works.
Hence, we introduce a new algorithm for the problem that uses a black-box solver for the harmonic setting as a subroutine. 

Consider an instance with agent tree $\agentTree$, conversion graph $\convGraph$ and supplies $s_i^r$, along with capped harmonic demands of the form $D_i^r(p) = \min\big\{ B_i^r / p, d_i^r\big\}$.
In addition, we consider a parametric instance with harmonic demands, parametrized by \textit{virtual budgets} $\hat B = \qty\big( \hat{B}_i^r )_{i, r}$.
Let $\qty\big( p(\hat{B}), c(\hat{B}) )$ be solutions to this harmonic  problem\footnote{One might wonder what happens when there are non-unique solutions to the harmonic problem, i.e., if $p(\hat{B})$ is a set with more than one element. As we showed in \cref{thm:harmonic:main}, the prices are unique, unless they do not appear in the optimization objective. This makes them either unique or non-consequential. Therefore, for simplicity, we assume that $p(\hat{B})$ is unique. We offer a detailed discussion of this in \cref{subsec: justify uniqueness}.}.
Consider what would happen if we set the virtual budgets equal to the actual budgets, $\hat B = B$, and calculate the solution to the harmonic problem $(p(B), c(B))$.
We examine two cases.

First, consider that for all $i,r$, it holds that $\hat B_i^r / p_i^r(B) \leq d_i^r$.
Then, these prices (and their conversions and flows) would also induce a feasible solution in the capped harmonic setting with budgets $B$ and maximum demands $d$, because the demands ``agree:'' $B_i^r / p_i^r(B) = D_i^r\qty\big( p_i^r(B) )$.
In this ``good'' scenario, the maximum demands are not binding, and we can treat the whole problem as one with harmonic demands.

Now consider that for some $i, r$, it holds that $\hat B_i^r / p_i^r(B) > d_i^r$.
Then, we can interpret this as agent $i$ allocating \emph{too much budget} for resource $r$.
Hence, a natural remedy would be to reduce this budget and use a virtual budget $\hat B_i^r < B_i^r$.
However, this may affect the prices of other agents and resources, creating complicated dependencies on how one should reduce budgets that are too high.
That said, if we could find a new budget vector $\hat B$ such that in the induced solution it holds that $\hat B_i^r / p_i^r(\hat B) = D_i^r\qty\big( p_i^r(\hat B) )$ for all $i,r$, then this solution to the harmonic problem is also a solution for the general one, as we show next.
We formalize and generalize this to any demands with our next theorem.
We present this theorem under our standard assumption of \cref{remark:positive_prices} that our instance has positive prices.

\begin{theorem}\label{thm: budget bound}
    Fix a problem with demands $D_i^r(p)$ for $i,r$.
    Fix virtual budgets $\hat B = (\hat B_i^r)_{i,r}$ and let $\qty\big( p(\hat B), c(\hat B) )$ be optimal solutions to the harmonic problem with budgets $\hat B$.
    Assume the prices $p(\hat B)$ are positive and that the allocations between the two problems match: $\hat B^r_i / p_i^r(\hat B) = D_i^r\qty\big(p_i^r(\hat B))$ for all $i,r$.
    Then, $\qty\big(p(\hat{B}), c(\hat{B}))$ are feasible price and conversion vectors for the original problem.
    Furthermore, for any feasible positive prices $p$ for the general problem, setting $\hat{B}^r_i = p^r_i \cdot D_i^r(p^r_i)$ satisfies $\hat{B}^r_i / p_i^r(\hat B) = D_i^r(p_i^r(\hat B))$.
\end{theorem}

\begin{proof}
    The two problems differ only in the demand functions, which, when evaluated at $p(\hat B)$, are the same.
    This means that at positive prices $p(\hat B)$, all axioms are satisfied for both problems.

    To see the reverse result, we note that the prices $p$ from the solution to the capped harmonic problem satisfy the Modified Flow Conservation \eqref{eq:harmonic:flow} for $\tilde{D}_i^r(x) = \frac{p_{i}^r\;D_i^r(p_{i}^r)}{x} = \frac{\hat{B}_i^r}{x} $.
    Hence, these prices are equal to $p^r_i(\hat{B})$.
\end{proof}

\begin{algorithm}[t]
\DontPrintSemicolon
\caption{Budget Descent Algorithm}
\label{algo:BDA}
\KwIn{Problem Instance $(\agentTree, \convGraph, s, D)$, maximum iterations $T$}

Let $p(\tilde B)$ be the optimal prices in the harmonic problem with budgets $\tilde B$: $(\agentTree, \convGraph, s, \tilde  B)$

For every $i, r$ set the initial virtual budgets $\tilde B_i^r[0] = \lim_{p \to \infty} \qty\big( p\cdot D_i^r(p) )$

\For{$t = 1, 2, \cdots, T$}
{
    For every $i,r$ update the virtual budgets
    \begin{minipage}{.913\textwidth}
    \begin{equation} \label{eq:update_rule}
        \tilde B_i^r[t]
        = E_i^r\qty( p_i^r\qty\big(\tilde B[t-1]) ) \quad\quad\quad \text{where } E_i^r(p) = p\cdot D_i^r(p)
    \end{equation}
    \end{minipage}
}

\Return prices $p\qty\big(\tilde B[T])$ and the corresponding conversions $c$

\end{algorithm}

Following the cue of this result, we introduce the \textit{Budget Descent Algorithm} (\cref{algo:BDA}).
In this algorithm, we initialize virtual budgets to the highest payment $\tilde B_i^r[0] = \lim_{p \to \infty} \qty\big( p\cdot D_i^r(p) )$ (which is bounded by \cref{assumption:demands}).
We then iteratively decrease them based on the calculated prices $p(\tilde B[t])$ and the general demands (we prove next that the virtual budgets decrease across iterations).
In every iteration, the new virtual budgets equal the payment generated if the current prices were used, where $E_i^r(p) = p \cdot D_i^r(p)$ denotes the payment at price $p$.
This update rule offers multiple advantages.
First, unlike many update rules in iterative algorithms (such as Gradient Descent), it is parameter-free, making it easily implementable without having to tune hyperparameters; we explore this for real datasets in \cref{sec:simulations}.
Second, as we prove next, it guarantees that there is never resource underutilization.
Specifically, we prove that $\tilde B_i^r[t] / p_i^r(\tilde B[t]) \ge D_i^r\qty\big(p_i^r(\tilde B[t]))$, i.e., under the calculated prices, the allocation in the harmonic problem is at least the allocation in the general problem.
Combining this inequality with the update rule \cref{eq:update_rule}, we also get that the virtual budgets decrease over time, as promised in the beginning.
Both of these facts follow from the following technical result, which proves that if the virtual budgets decrease, their corresponding prices also decrease.

\begin{lemma}\label{lem: key structural lemma}
    For $\hat B$ and $\tilde B$ such that $\hat B\preceq \tilde B$, it holds that $p(\hat{B}) \preceq p(\tilde B)$.\footnote{We say $x \preceq y$ if $x_i \leq y_i \;\forall i$.} 
\end{lemma}

Our main technical tool is Topkis's theorem; we state a simplified version of it \cite[Theorem 5]{milgrom1994monotone}.

\begin{theorem}[Topkis's Theorem (simplified)] \label{thm:topkis}
    Let $X \subseteq \R^n$ be a lattice\footnote{A set $X \subseteq \R^n$ is called a lattice if $x_1, x_2 \in X$ implies $x_1 \wedge x_2, x_1 \vee x_2 \in X$, where $\wedge$ and $\vee$ are the piece-wise minimum and maximum operators, respectively.} and $\Theta \subseteq \R^m$.
    For $x \in X$ and $\theta \in \Theta$ let $f(x ; \theta)$ be a twice continuously differentiable function such that $\frac{\partial^2 f}{\partial x_i\partial x_j} \ge 0$ and $\frac{\partial^2 f}{\partial x_i\partial \theta_k} \ge 0$ for any $x\in X$, $\theta \in \Theta$, $i \ne j$ and $k \in [m]$.
    Then $\argmax_{x \in X} f(x ; \theta)$ is weakly increasing in $\theta$. Specifically, when $\argmax_{x \in X} f(x ; \theta)$ has a unique solution for  $\theta \in \Theta$, for any $\theta \preceq \theta'$, $\argmax_{x \in X} f(x ; \theta) \preceq \argmax_{x \in X} f(x ; \theta')$.
\end{theorem}

For the root $i$, one checks directly that $F_i( p_{i} ; \hat B)$ (\cref{eq: optimization objective}) satisfies the conditions of \cref{thm:topkis}, so the claim follows.
For a non-root $i$, the objective $F_i( p_{i}; \hat B)$ depends on $p_{\pi(i)}(\hat B)$, which need not be differentiable; we instead apply \cref{thm:topkis} to a sequence of parameters---the parent prices $p_{\pi(i)}$, then the budgets $\hat B$ (after restricting the feasible region to a ``good'' set), and finally the supplies $s_i^r$, which removes this restriction.

To apply \cref{thm:topkis}, it will be useful to define the following sets:

\begin{definition}
Let $\mathcal{K} = \{x \in \R^m_+: x^r \geq x^{r'}\;\forall(r,r') \in \convGraph\}$ be the feasible region in the optimization programs defining $p_{i}$, for any $i$.
Let $\mathcal{B}$ be the set of budgets $\hat B$ such that $\sum_{j \in \subTree(i) } \hat B_j^r > 0$ for all $(i, r)$ such that $s_i^r = 0$.
We note that these sets are convex.
For the rest of this section, we will assume $B[t],\bar B \in \mathcal{B}$ for all $t$.
Note that \cref{thm:harmonic:main} tells us that, for $\hat{B} \in \mathcal{B}$, the program in \cref{eq: optimization problem} has a unique solution.
In \cref{subsec: justify uniqueness}, we argue that this is without loss of generality.
\end{definition}

\begin{proof}[Proof of \cref{lem: key structural lemma}]
We first notice that $\mathcal{K}$ is a lattice:
$$\text{If }p_{i}^{r_1} \geq p_{i}^{r_2}\text{ and }\tilde p_{i}^{r_1} \geq \tilde p_{i}^{r_2}\text{, then }\min\{ p_{i}^{r_1}, \tilde p_{i}^{r_1}\} \geq \min\{ p_{i}^{r_2}, \tilde p_{i}^{r_2}\}\text{ and }\max\{ p_{i}^{r_1}, \tilde p_{i}^{r_1}\} \geq \max\{ p_{i}^{r_2}, \tilde p_{i}^{r_2}\}$$

We will prove the result inductively on the tree.
For the root, $F_i( p_{i} ; \hat B)$, the optimization function for the program defining $p^r_{i}(\hat{B})$ parametrized by $\hat B$ (\cref{eq: optimization objective}) is clearly differentiable in all the variables.
By taking appropriate derivatives, we see that the conditions of Topkis's Theorem hold, proving the base case.

We now consider a non-root agent $i$ and assume the result holds for their parent.
Fix $\hat B$ and $\bar B$ such that $\hat B \preceq \bar B$.
Our goal is to show that $p_{i}(\hat B) \preceq p_{i}(\bar B)$.
A natural hope would be to use the exact same idea as in the root case.
However, the objective function $F_i( p_{i} ; \hat B)$ depends on $p_{\pi(i)}(\hat B)$ which will not be differentiable everywhere in general.
Hence, we change our approach by taking $p_{\pi(i)}$ as an independent parameter of its own.
Then, $F_i( p_{i} ; \hat B, p_{\pi(i)})$ is differentiable in all its variables.

We first show that, taking $p_{\pi(i)}$ as the varying parameter, the conditions of Topkis's Theorem hold.
Thus, the optimal prices are weakly increasing in the function $p_{\pi(i)}$.
Furthermore, by our induction hypothesis $p_{\pi(i)}(\hat B) \preceq p_{\pi(i)}(\bar B)$.
Hence,
$$\argmax_{x \in \mathcal{K}} F_i(x; \bar B; p_{\pi(i)}(\hat B)) \preceq \argmax_{x \in \mathcal{K}} F_i(x; \bar B; p_{\pi(i)}(\bar B))  = p_i(\bar B).$$

We then fix $p_{\pi(i)}$ to $p_{\pi(i)}(\hat B)$ and vary the budget.
We show that the conditions of Topkis's Theorem hold if $p_{\pi(i)}(\hat B) \succeq p_i$ on the feasible region of the problem.
By \cref{lem: prices are decreasing}, this would be true if the budget parameter were $\hat B$.
However, when it is not, the inequality will not be true in general.
To remedy this, we restrict the feasible region of the problem to be contained in $\mathcal{H}(\hat B) = \{x \in \R^m: x^r \leq p^r_{\pi(i)}(\hat B)\}$.
Hence, Topkis's theorem holds on the problem $\argmax_{x \in \mathcal{K} \cap \mathcal{H}(\hat B)} F_i(x;\; \cdot \;; p_{\pi(i)}(\hat B))$.
Thus,
$$p_i(\hat B) = \argmax_{x \in \mathcal{K} \cap \mathcal{H}(\hat B)} F_i(x; \hat B; p_{\pi(i)}(\hat B)) \preceq \argmax_{x \in \mathcal{K} \cap \mathcal{H}(\hat B)} F_i(x; \bar B; p_{\pi(i)}(\hat B))$$
where the first equality follows by noticing that when the budget parameter is $\hat B$, the optimal price will be contained in $\mathcal{H}(\hat B)$ regardless of whether we have that constraint.
Hence, it remains to show that
$$\argmax_{x \in \mathcal{K} \cap \mathcal{H}(\hat B)} F_i(x; \bar B; p_{\pi(i)}(\hat B)) \preceq \argmax_{x \in \mathcal{K}} F_i(x; \bar B; p_{\pi(i)}(\hat B)).$$
This inequality follows by another application of Topkis's Theorem, taking the supplies $s_i^r$ as variables, and relating the KKT conditions of the former problems to that of a related optimization program.
We detail this proof step in \cref{lem: topkis application to supply}.
\end{proof}

This lemma guarantees that if we decrease the virtual budgets, the prices also decrease.
Because the update rule in \cref{eq:update_rule} is weakly increasing in the prices (smaller prices lead to smaller virtual budgets), we get that the virtual budgets decrease over the rounds of the algorithm.
The formal condition we need to guarantee that is that in every round $t$, we overestimate the demands.
We formally state this invariant in the next lemma.

\begin{restatable}{lemma}{BDAInvariant} \label{lem:capped:invariant}
    Assume that the demands satisfy \cref{assumption:demands}.
    In every round $t$ of \cref{algo:BDA} and for all $i \in [n]$ and $r \in [m]$ it holds
    \begin{equation}\label{eq:invariant}
        \tilde B_i^r[t] \ge  E_i^r\qty\big(p_i^r(\tilde B[t]))
        .
    \end{equation}
\end{restatable}

\begin{proof}
    We prove the lemma inductively.
    We notice that \cref{eq:invariant} holds for $t = 0$ by the definition of $\tilde B_i^r[0]$ which is set to $\lim_{p \to \infty} p D_i^r(p)$, the maximum value $p D_i^r(p)$ can take (due to the monotonicity of this function).
    
    Fix round $t \ge 1$ and assume \cref{eq:invariant} holds at round $t-1$:
    \begin{equation*}
        \tilde B_i^r[t-1]
        \ge
        E_i^r(p_i^r(\tilde B[t-1]))
    \end{equation*}
Since $\tilde B_i^r[t]$ is set to the r.h.s., we get that $\tilde B_i^r[t-1] \ge \tilde B_i^r[t]$ for all $i,r$.
    By \cref{lem: key structural lemma}, this implies that $p(\tilde{B}[t-1]) \succeq p(\tilde B[t])$.
    Since the function $x \mapsto x \, D_i^r(x)$ is weakly increasing, this proves that
    \begin{equation*}
        \tilde B_i^r[t]
        =
        E_i^r(p_i^r(\tilde B[t-1]))
        \ge
        E_i^r(p_i^r(\tilde B[t]))
    \end{equation*}
    which proves the lemma.
\end{proof}

We now proceed to prove that our algorithm converges to a solution.
Specifically, we will show that the virtual budgets $\tilde B[t]$ converge to the desired values of \cref{thm: budget bound}.
In order to prove this, we have to show two things.
First, we show that the sequence $\tilde B[t]$ does converge, which is fairly straightforward by the fact that $\tilde B[t]$ is non-increasing and non-negative.
Second, we have to show that the sequence converges to the desired value according to \cref{thm: budget bound}.
For this, \cref{eq:invariant} is required to hold with equality, which intuitively is true: we decrease the virtual budgets whenever this condition is not true, so it should converge to an equality.
This, however, requires that the limit sequence of prices $p(\tilde B[t])$ converges, and in particular converges to the optimal price vector of the limit of $\tilde B[t]$.
This is exactly what the following lemma guarantees.

\begin{lemma}\label{lem: key structural lemma; price continuity}
    Fix a sequence of virtual budgets $ \tilde B[1],  \tilde B[2], \ldots$ such that $\lim_{t \to \infty} \tilde B[t] = \bar B$. Then, $p(\bar B) = \lim_{t\rightarrow \infty}p( \tilde B[t])$.
\end{lemma}

We now show that \cref{lem: key structural lemma; price continuity} follows from Berge's Maximum Theorem.
We state a simplified version of \cite[Theorem 1]{walker1979generalization}.

\begin{theorem}[Berge's Maximum Theorem (simplified)] \label{thm:berge's}
    Let $X \subseteq \R^n$ be a compact set and $F_i(x; \theta)$ be a function that is continuous in $x$ and $\theta$. Then, $\arg\max_{x \in X}F_i(x; \theta)$ is upper hemicontinuous. Hence, if $F_i(x; \theta)$ always has a unique maximizer $x^*(\theta) = \arg\max_{x \in X}F_i(x; \theta)$ for $\theta \in Y$, then $x^*(\theta)$ is continuous on $Y$.
\end{theorem}

\begin{proof}[Proof of \cref{lem: key structural lemma; price continuity}]
        For the root $i$, we see that $F_i(x; \hat{B})$ is clearly continuous in the variables.
        Since there is a unique $p_i(\hat{B}) = \arg\max_{x \in \mathcal{K}}F_i(x; \hat{B})$, by \cref{thm:berge's}, $p_i(\hat{B})$ is continuous.
        We now proceed inductively.
        Assuming $p_{\pi(i)}(\hat{B})$ is continuous, we conclude that $F_i(x; \hat{B})$ is continuous.
        Thus, by \cref{thm:berge's}, $p_i(\hat{B})$ is continuous on $\mathcal{B}$ (where it has a unique maximizer).
        Since $B[t], \bar{B} \in \mathcal{B}$ for all $t$, we conclude  \cref{lem: key structural lemma; price continuity}.
\end{proof}

We now formally state the main result of this section, that \cref{algo:BDA} converges to a solution for the problem with general demands.

\begin{restatable}{theorem}{CappedMainTheorem}\label{thm: main correctness thm}
   Fix an instance $(\agentTree, \convGraph, s, D)$ with general demands that satisfy \cref{assumption:demands}.
   Let $p(\tilde  B[T])$ be the solution that the Budget Descent Algorithm (\cref{algo:BDA}) returns after $T$ rounds.
   Then, the limit $\lim_{T \to \infty} p(\tilde  B[T])$ exists and is a feasible price vector for the general instance.
\end{restatable}

The theorem is proven by \cref{lem: key structural lemma; price continuity} and the arguments we highlighted above it. For completeness, we give a full detailed proof in \cref{app: proofs for key lemmas}.

\subsection{Convergence Rate of the Budget Descent Algorithm}

In this section, we prove a bound on the number of iterations required to approximate the optimal allocation for any instance of the problem. Informally, we will show that to achieve an allocation that is $\epsilon$ close to the optimal allocation, $O\left(\frac{1}{\epsilon}\log\left(\frac{1}{\epsilon}\right)\right)$ iterations of the Budget Descent Algorithm suffice, where the hidden constant depends on instance-specific quantities that we make precise below.

\itparagraph{Perturbed Demand Function.} To prove our convergence bound, we will run the Budget Descent Algorithm on a regularized set of demand functions.
For a demand function $D_i^r$ satisfying \cref{assumption:demands}, let $\overline{E_i^r} = \lim_{p \rightarrow \infty} p \cdot D_i^r(p)$ be the maximum payment that can be generated, which is finite by \cref{assumption:demands}.
We call $$D_i^r(p; \delta) = D_i^r(p) + \delta\cdot \frac{ \overline{E_i^r}}{p}$$ the \textit{$\delta$-perturbed demand function}.
We note that the perturbed demand function satisfies \cref{assumption:demands}.
Let $E_i^r(p; \delta) = p\cdot D_i^r(p; \delta)$ be the payment generated by this function at price $p$.

\begin{restatable}{theorem}{ConvergenceResult}\label{thm: main convergence result}
    Let $\tilde{B}_i^{r}[t;\delta]$ be the virtual budgets produced at the $t$-th iteration of the Budget Descent Algorithm on the problem instance with $\delta$-perturbed demand functions. Fix $\epsilon > 0$. Then, for $t = \frac{1 + \epsilon}{\epsilon}\log\left(\frac{1}{\epsilon}\right)$,
    \begin{equation}
\sum_{i,r}\left|D_i^r(p_i^r(\bar{B})) - D_i^r(p_i^r(\tilde{B}[t; \epsilon]))\right| \leq \epsilon\left[\left(\max_{i,r}\frac{\tilde{B}_i^r[0; \epsilon]}{\bar{B}_i^r[\epsilon]}-1\right)\left(\sum_{i,r} s_i^r \right) + \sum_{i,r}\frac{\overline{E_i^r}}{p_i^r(\bar{B})} \right]
    \end{equation}
\end{restatable}

This theorem will follow from two results. The first result bounds the distance in allocations between the limiting budgets in the perturbed problem $\bar{B}[\delta] = \lim_{t\rightarrow \infty} \tilde{B}[t; \delta]$ and the limiting budgets in the original problem $\bar{B} = \bar{B}[0]$. The error is proportional to an instance-specific constant that depends on the ratio between the maximum possible payment $\overline{E_i^r}$ and the optimal price $p_i^r(\bar{B})$.

\begin{restatable}{lemma}{BoundingLimitings}
\label{lem: bounding limiting values}
    For all $\delta > 0$, $$\sum_{i,r}\left|D_i^r(p_i^r(\bar{B})) - D_i^r(p_i^r(\bar{B}[\delta]))\right| \leq \delta \sum_{i,r}\frac{\overline{E_i^r}}{p_i^r(\bar{B})}.$$
\end{restatable}

The second result bounds the distance in allocations between $\bar{B}[\delta]$ and $\tilde{B}[t; \delta]$. Here, the bound again depends on an instance-specific constant, depending on how far the initial budgets $\tilde B_i^r[0; \delta]$ are from the limiting ones $\bar B_i^r[\delta]$ and the total supply in the system $\sum_{i,r} s_i^r$.

\begin{restatable}{lemma}{BoundingIterate}
\label{lem: bounding limiting value and iterate}
    For all $t$ and $\delta > 0$, 
    $$\sum_{i,r}\left|D_i^r(p_i^r(\bar{B}[\delta])) - D_i^r(p_i^r(\tilde{B}[t;\delta]))\right|\leq \left(\max_{i,r}\frac{\tilde{B}_i^r[0; \delta]}{\bar{B}_i^r[\delta]}-1\right)\left(1 - \frac{\delta}{1 + \delta}\right)^t\left(\sum_{i,r} s_i^r  \right).$$
\end{restatable}

The correctness of \cref{lem: bounding limiting values} will follow from the monotonicity of $\tilde{B}[t; \delta]$ with respect to $\delta$ ie. for all $\delta \leq \delta'$, $\tilde{B}[t; \delta] \preceq \tilde{B}[t; \delta']$.

Proving \cref{lem: bounding limiting value and iterate} is more involved. We define $\mu_t(\delta)$ as the smallest value such that $$\tilde{B}[t; \delta] \preceq \mu_t(\delta)\cdot \bar{B}[\delta].$$ We note that since $\bar{B}[\delta] \preceq \tilde{B}[t; \delta]$, $\mu_t(\delta) \geq 1$ is a measure of the distance of $\tilde{B}[t; \delta]$ from $\bar{B}[\delta]$ ie. $\mu_t(\delta) \approx 1$ implies that $\tilde{B}[t; \delta] \approx \bar{B}[\delta]$. We show that $\eta_t(\delta) = \mu_t(\delta) - 1$ decays exponentially with decay rate $\left(1 - \frac{\delta}{1 + \delta}\right)$, as shown in  \cref{lem: bounding limiting value and iterate}.
Interestingly, the only structural properties about the optimal $p_i^r(B)$ generated by our algorithm required for this result are that $p_i^r(B)$ is monotonically increasing (\cref{lem: key structural lemma}) and $p_i^r(B)$ is positively homogeneous of degree 1:
\begin{proposition}\label{prop: positively homogeneous}
    For all $B \in \R^{n\times m}$ and $\alpha > 0$, $p_i^r(\alpha\cdot B) = \alpha\cdot p_i^r(B)$ for all $i,r$.
\end{proposition}
\begin{proof}
    It should be immediate that, given any prices that satisfy the Modified Flow Conservation (\cref{eq:harmonic:flow}), scaling both the prices and the budgets by the same positive factor does not affect the correctness of the equality.
\end{proof}

We detail the proof of \cref{thm: main convergence result} in \cref{app: convergence rate} of the appendix.

\subsection{Uniqueness of the Optimal Allocation} \label{general:uniqueness}

We now prove that all $\hat{B}$ such that $E_i^r(p_i^r(\hat{B})) = \hat{B}_i^r$ induce the same allocation. This result follows from the fact that the Budget Descent Algorithm converges to a fixed point that dominates all other fixed points, i.e., the limiting virtual budgets $\bar B$ are weakly bigger than the budgets of any other fixed point.

\begin{proposition}
    Let $\hat{B}$ such that $E_i^r(p_i^r(\hat{B})) = \hat{B}_i^r$. Then, $\hat{B} \preceq \bar{B}$.
\end{proposition}
\begin{proof}
    We show inductively that $\hat{B} \preceq \tilde{B}[t]$. The base case holds since $$\hat{B}_i^r = E_i^r(p_i^r(\hat{B})) \leq \max_pE_i^r(p) = \tilde{B}_i^r[0].$$ We now assume  $\hat{B} \preceq \tilde{B}[t]$. Since $E_i^r$ and $p_i^r$ are monotone: $$\hat{B}_i^r = E_i^r(p_i^r(\hat{B})) \leq E_i^r(p_i^r(\tilde{B}[t])) = \tilde{B}_i^r[t+1].$$ Hence, $\hat{B}$ lower bounds the sequence $\tilde{B}[t]$. Hence, $\hat{B} \preceq \lim_{t \rightarrow \infty}\tilde{B}[t]  = \bar{B}$.
\end{proof}

We now prove the uniqueness result:
\begin{theorem}
    Any solution to an instance satisfying \cref{assumption:demands} is effectively unique, i.e., has unique allocations.
\end{theorem}

\begin{proof}
    Since $\hat{B} \preceq \bar{B}$ and $D_i^r$ is non-increasing and $p_i^r$ is non-decreasing, $D_i^r(p_i^r(\bar{B})) \leq D_i^r(p_i^r(\hat{B}))$. Hence,
    \begin{align*}
        0 \leq D_i^r(p_i^r(\hat{B})) - D_i^r(p_i^r(\bar{B})) = \frac{\hat{B}_i^r}{p_i^r(\hat{B})} - \frac{\bar{B}_i^r}{p_i^r(\bar{B})}
    \end{align*}
    We now sum over all $i,r$ to deduce that 
    \begin{align*}
        0 \leq \sum_{i,r}D_i^r(p_i^r(\hat{B})) - D_i^r(p_i^r(\bar{B})) = \sum_{i,r}\frac{\hat{B}_i^r}{p_i^r(\hat{B})} - \sum_{i,r}\frac{\bar{B}_i^r}{p_i^r(\bar{B})} = \sum_{i,r}s_i^r - \sum_{i,r}s_i^r = 0
    \end{align*} Hence, we conclude that $D_i^r(p_i^r(\hat{B})) - D_i^r(p_i^r(\bar{B})) = 0$ for all $i,r$ as claimed.
\end{proof}

\begin{remark}[Robustness to Misreporting]
\label{rem:incentives}
    Since allocations are driven by the reported demand functions, it is natural to ask whether an agent can secure more resources by misreporting its demand.
    We consider the capped harmonic example where $D_i^r(p) = \min\{ B_i^r / p, d_i^r \}$ and the two parameters the agent controls: its budget $B_i^r$ and its demand cap $d_i^r$.
    We argue that misreporting these, is not a concern in our setting.

    First, \emph{overreporting} either parameter is either infeasible or useless.
    Budgets are exogenously endowed, so an agent cannot spend credits it was never given.
    Inflating the demand cap $d_i^r$ only registers demand for resources the agent derives no value from, so any over-allocated units are worthless.

    Second, \emph{underreporting} does not help: replacing $B_i^r$ or $d_i^r$ (or both) by smaller values cannot increase the equilibrium allocation of $(i, r)$.
    Intuitively, the share of resources an agent draws from its ancestors tracks the budget it brings to the market, so understating that budget only weakens its claim.
    This is consistent with the monotonicity of feasible prices in the virtual budgets (\cref{lem: key structural lemma}): reducing $B_i^r$ or $d_i^r$ weakly lowers the virtual budget that $(i, r)$ contributes, which weakly lowers prices throughout the tree and, in turn, the demand the market clears for $(i, r)$.
\end{remark}

\printbibliography{}

@article{arrow1954existence,
  title     = {Existence of an equilibrium for a competitive economy},
  author    = {Arrow, Kenneth J and Debreu, Gerard},
  journal   = {Econometrica: Journal of the Econometric Society},
  volume    = {22},
  number    = {3},
  pages     = {265--290},
  year      = {1954},
  publisher = {JSTOR}
}

@inproceedings{banerjee2023robust,
  title     = {Robust Pseudo-Markets for Reusable Public Resources},
  author    = {Banerjee, Siddhartha and Fikioris, Giannis and Tardos, Eva},
  booktitle = {Proceedings of the 24th ACM Conference on Economics and Computation},
  pages     = {241--241},
  year      = {2023}
}

@article{devanur2008market,
  title     = {Market equilibrium via a primal--dual algorithm for a convex program},
  author    = {Devanur, Nikhil R and Papadimitriou, Christos H and Saberi, Amin and Vazirani, Vijay V},
  journal   = {Journal of the ACM (JACM)},
  volume    = {55},
  number    = {5},
  pages     = {1--18},
  year      = {2008},
  publisher = {ACM New York, NY, USA}
}

@article{eisenberg1959consensus,
  title     = {Consensus of subjective probabilities: The pari-mutuel method},
  author    = {Eisenberg, Edmund and Gale, David},
  journal   = {The Annals of Mathematical Statistics},
  volume    = {30},
  number    = {1},
  pages     = {165--168},
  year      = {1959},
  publisher = {JSTOR}
}

@article{elokda2024self,
  title     = {A self-contained karma economy for the dynamic allocation of common resources},
  author    = {Elokda, Ezzat and Bolognani, Saverio and Censi, Andrea and D{\"o}rfler, Florian and Frazzoli, Emilio},
  journal   = {Dynamic Games and Applications},
  volume    = {14},
  number    = {3},
  pages     = {578--610},
  year      = {2024},
  publisher = {Springer}
}

@article{elokda2025carma,
  title     = {CARMA: Fair and efficient bottleneck congestion management via nontradable karma credits},
  author    = {Elokda, Ezzat and Cenedese, Carlo and Zhang, Kenan and Censi, Andrea and Lygeros, John and Frazzoli, Emilio and D{\"o}rfler, Florian},
  journal   = {Transportation Science},
  volume    = {59},
  number    = {2},
  pages     = {340--359},
  year      = {2025},
  publisher = {INFORMS}
}

@inproceedings{fikioris2024incentives,
  title        = {Incentives in dominant resource fair allocation under dynamic demands},
  author       = {Fikioris, Giannis and Agarwal, Rachit and Tardos, {\'E}va},
  booktitle    = {International Symposium on Algorithmic Game Theory},
  pages        = {108--125},
  year         = {2024},
  organization = {Springer}
}

@inproceedings{ghodsi2011dominant,
  title     = {Dominant resource fairness: Fair allocation of multiple resource types},
  author    = {Ghodsi, Ali and Zaharia, Matei and Hindman, Benjamin and Konwinski, Andy and Shenker, Scott and Stoica, Ion},
  booktitle = {8th USENIX symposium on networked systems design and implementation (NSDI 11)},
  year      = {2011}
}

@inproceedings{gorokh2021remarkable,
  title     = {The Remarkable Robustness of the Repeated Fisher Market},
  author    = {Gorokh, Artur and Banerjee, Siddhartha and Iyer, Krishnamurthy},
  booktitle = {Proceedings of the 22nd ACM Conference on Economics and Computation},
  pages     = {562--562},
  year      = {2021}
}

@article{hatfield2013stability,
  title     = {Stability and competitive equilibrium in trading networks},
  author    = {Hatfield, John William and Kominers, Scott Duke and Nichifor, Alexandru and Ostrovsky, Michael and Westkamp, Alexander},
  journal   = {Journal of Political Economy},
  volume    = {121},
  number    = {5},
  pages     = {966--1005},
  year      = {2013},
  publisher = {University of Chicago Press Chicago, IL}
}

@article{hylland1979efficient,
  title     = {The efficient allocation of individuals to positions},
  author    = {Hylland, Aanund and Zeckhauser, Richard},
  journal   = {Journal of Political economy},
  volume    = {87},
  number    = {2},
  pages     = {293--314},
  year      = {1979},
  publisher = {The University of Chicago Press}
}

@inproceedings{lin_et_al,
  author    = {Lin, David X. and Fikioris, Giannis and Banerjee, Siddhartha and Tardos, \'{E}va},
  title     = {{Robust Resource Allocation via Competitive Subsidies}},
  booktitle = {17th Innovations in Theoretical Computer Science Conference (ITCS 2026)},
  pages     = {96:1--96:15},
  series    = {Leibniz International Proceedings in Informatics (LIPIcs)},
  isbn      = {978-3-95977-410-9},
  issn      = {1868-8969},
  year      = {2026},
  volume    = {362},
  editor    = {Saraf, Shubhangi},
  publisher = {Schloss Dagstuhl -- Leibniz-Zentrum f{\"u}r Informatik},
  address   = {Dagstuhl, Germany},
  url       = {https://drops.dagstuhl.de/entities/document/10.4230/LIPIcs.ITCS.2026.96},
  urn       = {urn:nbn:de:0030-drops-253835},
  doi       = {10.4230/LIPIcs.ITCS.2026.96}
}

@article{milgrom1994monotone,
  title     = {Monotone comparative statics},
  author    = {Milgrom, Paul and Shannon, Chris},
  journal   = {Econometrica: Journal of the Econometric Society},
  pages     = {157--180},
  year      = {1994},
  publisher = {JSTOR}
}

@article{ostrovsky2008stability,
  title     = {Stability in supply chain networks},
  author    = {Ostrovsky, Michael},
  journal   = {American Economic Review},
  volume    = {98},
  number    = {3},
  pages     = {897--923},
  year      = {2008},
  publisher = {American Economic Association}
}

@inproceedings{shneidman2005markets,
  title     = {Why Markets Could (But Don't Currently) Solve Resource Allocation Problems in Systems.},
  author    = {Shneidman, Jeffrey and Ng, Chaki and Parkes, David C and AuYoung, Alvin and Snoeren, Alex C and Vahdat, Amin and Chun, Brent N},
  booktitle = {HotOS},
  year      = {2005}
}

@inproceedings{sivanquota,
  title     = {Quota Marketplace: Dynamic Pricing for Efficient Allocation of ML Training Resources},
  author    = {Sivan, Balasubramanian and Leme, Renato Paes and Tiuca, Mihai and McFarlane, Ian and Gkatzelis, Vasilis and Mehta, Nehal and Yeganeh, Soheil Hassas and Mirrokni, Vahab and Vahdat, Amin},
  booktitle = {20th USENIX Symposium on Operating Systems Design and Implementation (OSDI 26)},
  year      = {2026}
}

@inproceedings{stokely2009market,
  title        = {Using a market economy to provision compute resources across planet-wide clusters},
  author       = {Stokely, Murray and Winget, Jim and Keyes, Ed and Grimes, Carrie and Yolken, Benjamin},
  booktitle    = {2009 IEEE International Symposium on Parallel \& Distributed Processing},
  pages        = {1--8},
  year         = {2009},
  organization = {IEEE}
}

@inproceedings{vuppalapati2023karma,
  title     = {Karma: Resource allocation for dynamic demands},
  author    = {Vuppalapati, Midhul and Fikioris, Giannis and Agarwal, Rachit and Cidon, Asaf and Khandelwal, Anurag and Tardos, {\'E}va},
  booktitle = {17th USENIX Symposium on Operating Systems Design and Implementation (OSDI 23)},
  pages     = {645--662},
  year      = {2023}
}

@article{walker1979generalization,
  title     = {A generalization of the maximum theorem},
  author    = {Walker, Mark},
  journal   = {International Economic Review},
  pages     = {267--272},
  year      = {1979},
  publisher = {JSTOR}
}

\appendix
\section{Omitted Proofs of Section \ref{sec:model}}\label{app:proof from sec 2}

\subsection{Eliminating Cycles in \texorpdfstring{$\convGraph$}{the Conversion Graph}}

\LemmaDAG*

\begin{proof}[Proof of \cref{lem:dag}]
    Fix a solution $(f, p, c)$ of $(\agentTree, \convGraph, D, s)$.
    We then create a solution $(\hat p, \hat c, \hat f)$ of $(\agentTree, \hat \convGraph, \hat D, \hat s)$, as follows.
    A price $\hat p_i^{\hat r}$ is equal to $p_i^{r}$ for any $r \in R(\hat r)$; all these prices are equal because all the resources in $R(\hat r)$ are in a cycle and by Conversion Complementarity.
    Flows are $\hat f_i^{\hat r} = \sum_{r \in R(\hat r)} f_i^r$.
    Conversions are $\hat c_i^{(\hat r_1, \hat r_2)} = \sum_{r_1 \in R(\hat r_1)} \sum_{r_2 \in R(\hat r_2)} \hat c_i^{(r_1, r_2)}$.
    Because all the new allocation-related variables (demands, flows, conversions, supplies) are the summation of the old ones, while the prices stay the same, it is not hard to show that the new solution is feasible.

    Now fix a solution $(\hat p, \hat c, \hat f)$ of $(\agentTree, \hat \convGraph, \hat D, \hat s)$ and create a solution $(f, p, c)$ of $(\agentTree, \convGraph, D, s)$, as follows.
    A price $p_i^{r}$ is equal to $p_i^{\hat r}$ for all $r \in R(\hat r)$.
    Flows are defined recursively from bottom to top: let $\hat r$ such that $r \in R(\hat r)$, then
    \begin{equation*}
        f_i^r = \hat f_i^{\hat r} \frac{
            p_i^r\qty(D_i^r(p_i^r) + \sum_{j \in \ch(i)}f_j^r)
        }{
            \sum_{r' \in R(\hat r)} p_i^{r'}\qty(D_i^{r'}(p_i^{r'}) + \sum_{j \in \ch(i)}f_j^{r'})
        }
    \end{equation*}
    if the denominator is $0$ then $f_i^r = 0$.
    There are multiple ways to define the conversions.
    Because between two strongly connected components $R(\hat r_1), R(\hat r_2)$ of $\convGraph$ we can convert resources in multiple ways, we can set $c_i^{r_1, r_2} = \hat c_i^{\hat r_1, \hat r_2}$ for some representatives $r_1 \in R(\hat r_1)$ and $r_2 \in R(\hat r_2)$ and set the rest of these cross-component conversions to $0$.
    Then these cross-component conversions can be routed to the other nodes of the component, since there is a path from $r_1$ to every other node in $R(\hat r_1)$ and all the prices in $R(\hat r_1)$ are the same.
    Inside a single component, in $(\agentTree, \hat \convGraph, \hat D, \hat s)$, the supply $s_i^{\hat r}$ that is used for the demand $\hat D_i^{\hat r}(\hat p_i^{\hat r})$ is being consumed by a single node; in the full graph this might need to be distributed to different nodes $(i, r)$ for $r \in R(\hat r)$.
    This again can be routed via paths in $R(\hat r)$, since this is a strongly connected component.
\end{proof}

\subsection{Isolating Degenerate Subgraphs} \label{ssec:app:isolating_degenerate}

\begin{remark}[Minimality in \cref{def:model:degenerate}]\label{rem:degenerate_minimality}
The minimality clause in \cref{def:model:degenerate} is essential: an absorbing subgraph can have more aggregate supply than demand yet still not be self-sufficient.
Consider $\mathcal G$ with three nodes where nodes $1$ and $2$ are children of node $3$ (for simplicity we do not differentiate between agents and resources), with demands $D_1(0) = D_2(0) = D_3(0) = 2$ and supplies $s_1 = 10$, $s_2 = 1$, $s_3 = 0$.
The total supply ($11$) exceeds the total demand ($6$), yet $\mathcal G$ is not self-sufficient: node $2$ has a supply deficit that node $1$'s excess cannot cover, since resources cannot flow from $1$ to $2$.
Hence, $\{1\}$ is degenerate but $\{1,2,3\}$ is not.
\end{remark}

\begin{lemma}\label{lem: prices are decreasing}
    Fix a instance $(\agentTree, \convGraph, D, s)$. Consider a solution $(f, p, c)$ for the instance. If $$\sum_{j \in \subTree(i) } p_j^r D_j^r(p_j^r) > 0 \quad\text{ or } \quad s_i^r > 0$$ for all $(i, r)$, then $p^r_i \leq p^r_{\parent(i)}$ for all $(i, r)$.
\end{lemma}
\begin{proof}
We assume that $p^r_i > p^r_{\parent(i)}$.  Consider $r_1 \in \outEdges(r)$ such that $c_i^{(r, r_1)} > 0$. By Conversion Complementarity, $p^r_i = p^{r_1}_i$. Since, $(r,r_1) \in \convGraph$,  $p_{\parent(i)}^{r_1} \leq p_{\parent(i)}^r$. Hence, $p_{\parent(i)}^{r_1} < p^{r_1}_i$. If there exists a resource $r_2 \in \outEdges(r_1)$ such that $c_i^{(r_1, r_2)} > 0$, then we can repeat the argument to achieve $p_{\parent(i)}^{r_2} < p^{r_2}_i$. We continue this process until we get to a resource $r^*$ such that $c_i^{(r^*, r')} = 0$ for all $r' \in \outEdges(r^*)$. We then have that $p_{\parent(i)}^{r^*} < p^{r^*}_i$. Since prices are non-negative, $p^{r^*}_i > p^{r^*}_{\pi(i)} \geq 0$. By Flow Conservation on $(i, r^*)$, we observe that 
    $$D_i^r(p_i^{r^*}) + \sum_{j \in \ch(i)} f_{j}^{r^*} = s_i^{r^*} + f_{i}^{r^*} + \sum_{r' \in \inEdges(r)} c_i^{(r', {r^*})} - \sum_{r' \in \outEdges(r^*)} c_i^{({r^*}, r')} =  s_i^{r^*} + f_{i}^{r^*} + \sum_{r' \in \inEdges(r)} c_i^{(r', {r^*})} $$ where we use the fact that $\sum_{r' \in \outEdges(r^*)} c_i^{({r^*}, r')} = 0$. Furthermore, Money Conservation tells us that 

    $$p_{\parent(i)}^{r^*}f_{i}^{r^*} = p_i^{r^*}\left(D_i^{r^*}(p_i^{r^*}) +\sum_{j \in \ch(i)} f_{j}^{r^*}\right) = p_i^{r^*}\left(s_i^{r^*} + f_{i}^{r^*} + \sum_{r' \in \inEdges(r)} c_i^{(r', {r^*})} \right) $$
Then, 
    $$(p_{\parent(i)}^{r^*} - p_i^{r^*})f_{i}^{r^*} =  p_i^{r^*}\left(s_i^{r^*}+ \sum_{r' \in \inEdges(r)} c_i^{(r', {r^*})} \right)$$ By assumption, $p_{\parent(i)}^{r^*} - p_i^{r^*} < 0$ but the r.h.s. is non-negative. Hence, $f_{i}^{r^*} = 0$ and $s_i^{r^*} = 0$. However, by \cref{lem: isolate flow lemma}, we deduce that $\sum_{j \in \subTree(i) } p_j^{r^*} D_j^{r^*}(p_j^{r^*}) = p_{\pi(i)}^{r^*}\; f^{r^*}_{i} = 0$, a contradiction.

\end{proof} 

This result immediately implies that,  in a feasible solution, if a pair $(i, r)$ has a small price then any other pair reachable from that node must also have as small a price.

\begin{restatable}{lemma}{LemmaPropagationOfZeroPrices} \label{lem:propagation_of_zero_prices}
    Fix an instance $(\agentTree, \convGraph, D, s)$ and a feasible solution $(f, p, c)$. Assume $$\sum_{j \in \subTree(i) } p_j^r D_j^r(p_j^r) > 0 \quad\text{ or } \quad s_i^r > 0$$ for all $(i, r)$. Fix $\epsilon \geq 0$.
    For all $p_i^r \leq  \epsilon$, $p_j^{r'} \leq \epsilon$ it holds that for any $(j, r')$ that is reachable from $(i, r)$ in the graph $\agentTree \times \convGraph$.
\end{restatable}

\begin{proof}
Let $p_i^r \leq  \epsilon$. By Conversion Complementary, $p_i^{r'} \leq p_i^r \leq \epsilon$ for all $r' \in \outEdges(r)$. By \cref{lem: prices are decreasing}, there exists a feasible solution in which $p^r_j \leq p^r_{i} \leq \epsilon$ for all $j \in \ch(i)$.
\end{proof}

\begin{proposition}\label{prop:smaller_price_or_zero}
    Fix an instance $(\agentTree, \convGraph, D, s)$ and a feasible solution $(f, p, c)$.
    If for any $(i, r)$ and $(i', r')$ that is reachable by $(i, j)$ in $\agentTree \times \convGraph$.
    Then either $p^{r'}_{i'} \le p^r_i$ or setting $p^{r'}_{i'}$ to $0$ produces a feasible solution where the demand, incoming/outgoing flows/conversions of $(i', r')$ stay the same.
\end{proposition}

\begin{proof}
    We notice that the proposition trivially holds for $i' = i$ by the Conversion rule.
    We now prove it for any agent $i' \in \ch(i)$.
    We start with a sink resource $r'$.
    Assume towards a contradiction that $p^{r'}_{i'} > p^r_i$.
    By \cref{lem:propagation_of_zero_prices} this implies that
    \begin{equation*}
        \sum_{j \in \subTree(i') } p_j^{r'} D_j^r(p_j^{r'}) = 0
        \quad \text{ or } \quad
        s_{i'}^{r'} = 0
    \end{equation*}
    which also implies that $D_{i'}^r(p_{i'}^{r'}) = 0$ since $p^{r'}_{i'} > 0$.
    In addition, by the Price Conservation axiom, \cref{lem: isolate flow lemma}, we also get that any $r'$ outgoing flow from $i'$ must be $0$, by the fact that $\sum_{j \in \subTree(i') } p_j^{r'} D_j^r(p_j^{r'}) = 0$ (the incoming money from $i'$ children is $0$) and that $p^{r'}_{i'} > 0$.
    This proves that $(i' ,r')$ is not using any resources: the demand, outgoing flow, and outgoing conversions (since $r'$ is a sink) are $0$.
    Therefore, by Flow Conversion and $p^{r'}_{i'} > 0$ everything incoming must also be $0$.
    Therefore, if we were to set $p^{r'}_{i'} = 0$ instead, nothing would change: the demand remains the same, there are no conversion rules to be broken, and the incoming flow is still $0$.

    Now fix a node $(i', r')$ and assume the proposition for all $r'' \in \outEdges(r')$ and that we have set these prices to $0$ if they were more than $p^r_i$ (in other words, they are not at most $p^r_i$). 
    Now consider that $p_{i'}^{r'} > p^r_i$.
    We use the same argument as before: by \cref{lem:propagation_of_zero_prices} there is no demand at $(i', r')$, there are no outgoing conversions since the prices of all $r'' \in \outEdges(r')$, and there is no outgoing flow.
    This means by Flow Conservation that the incoming flow and incoming conversions are $0$, making it inconsequential if we set $p_{i'}^{r'}$ to $0$.
\end{proof}

\TheoremDegeneratePrices*

We break the theorem into the following propositions.

\begin{proposition} \label{prop:no_degenerate_has_positive_prices}
    If no degenerate subgraph exists, then all the prices are positive.
\end{proposition}

\begin{proof}
    Towards a contradiction, assume that for some node there is a feasible solution whose price is $0$.
    Let $\mathcal G$ be the nodes $(i, r)$ whose prices are $0$, along with every node that is reachable by these nodes.
    By \cref{prop:smaller_price_or_zero}, we can assume that there is a feasible solution where the prices are $0$ for all nodes in $\mathcal G$.
    Now notice that for any $(i, r) \in \mathcal G$ and $(\pi(i), r) \notin \mathcal G$ it has to hold that $f_i^r = 0$: by Money Conservation, either $f_i^r = 0$ or $p_{\pi(i)}^r = 0$, but the last is impossible by definition of $\mathcal G$.
    Also notice that for every $(i, r) \in \mathcal G$ and $(i, r') \notin \mathcal G$ with $(r', r) \in \convGraph$, it has to hold that $c_i^{(r', r)} = 0$; otherwise, by Conversion Complementarity, it has to hold that $p_i^{r'} = 0$, which contradicts the definition of $\mathcal G$.
    This means that there is no incoming flow or conversion into $\mathcal G$.
    Since $\mathcal G$ is absorbing, there is no outgoing flow of conversion out of $\mathcal G$.
    This implies that if we sum up Flow Conservation for all $(i, r) \in \mathcal G$ we get
    \begin{equation*}
        \sum_{(i, r) \in \mathcal G} D_i^r(0)
        \le
        \sum_{(i, r) \in \mathcal G} s_i^r
    \end{equation*}
    which directly implies our assumption that there is no degenerate subgraph.
\end{proof}

\begin{proposition} \label{prop:degenerate_has_zero_prices}
    Fix a degenerate subgraph $\mathcal G$ and a feasible solution $(f, p, c)$.
    Then, setting the prices of $p_i^r$ to $0$ for $(i, r) \in \mathcal G$ remains a feasible solution.
\end{proposition}

\begin{proof}
    Let $\mathcal G_0$ be the subset of $\mathcal G$ whose prices are $0$, along with any nodes reachable by them.
    By \cref{prop:smaller_price_or_zero}, we can assume that all the prices of $\mathcal G_0$ are $0$, without changing anything.
    Now let $\mathcal G_+ = \mathcal G \setminus \mathcal G_0$.
    Because $\mathcal G_+$ has positive prices, there is no flow or conversions between $\mathcal G_+$ and $\mathcal G_0$.
    Also, because $\mathcal G$ is absorbing, there is no outgoing flow or conversion from $\mathcal G_+$ in general.
    This lets us use Flow Conservation on all $(i, r) \in \mathcal G_+$ with equality; summing them all up, we get
    \begin{equation*}
        \sum_{(i, r) \in \mathcal G_+} D_i^r(p_i^r)
        =
        \sum_{(i, r) \in \mathcal G_+} s_i^r
        + F_{\text{in}}
        + C_{\text{in}}
    \end{equation*}
    where $F_{\text{in}}$ and $C_{\text{in}}$ is the total incoming flow and conversion into $\mathcal G_+$ from outside $\mathcal G$.
    Because $F_{\text{in}} + C_{\text{in}} \ge 0$ and by non-increasingess of $D_i^r(\cdot)$ we have that
    \begin{equation*}
        \sum_{(i, r) \in \mathcal G_+} D_i^r(0)
        \ge
        \sum_{(i, r) \in \mathcal G_+} s_i^r
    \end{equation*}
    This is a contradiction: by $\sum_{(i, r) \in \mathcal G} D_i^r(0) \le \sum_{(i, r) \in \mathcal G} s_i^r$ we get that
    \begin{equation*}
        \sum_{(i, r) \in \mathcal G_0} D_i^r(0)
        \le
        \sum_{(i, r) \in \mathcal G_0} s_i^r
    \end{equation*}
    which implies that $\mathcal G_0$ cannot be a strict non-empty subset of $\mathcal G$ because then $\mathcal G$ would not be degenerate---recall \cref{def:model:degenerate} requires that a subgraph has to be minimal that satisfies the above property.
    This proves that either
    (a) $\mathcal G_+ = \emptyset$, i.e., we can set the prices of all nodes in $\mathcal G$ to $0$ without changing anything, or
    (b) $\mathcal G_+ = \mathcal G$ which gives that
    \begin{equation*}
        \sum_{(i, r) \in \mathcal G_+} s_i^r
        \ge
        \sum_{(i, r) \in \mathcal G_+} D_i^r(0)
        \ge
        \sum_{(i, r) \in \mathcal G_+} D_i^r(p_i^r)
        \ge
        \sum_{(i, r) \in \mathcal G_+} s_i^r
        + F_{\text{in}}
        + C_{\text{in}}
    \end{equation*}
    which proves that $F_{\text{in}} + C_{\text{in}} = 0$.
    This proves there is no incoming flow or conversion into $\mathcal G$.
    It also proves that each agent gets her demand at price $0$, which means we can change all the prices to $0$ without affecting anything (all the flows and conversions in $\mathcal G$ remain feasible if we change \textit{all} the prices in $\mathcal G$ to $0$).
\end{proof}

\begin{proof}[Proof of \cref{thm:degenerate_and_prices}]
    Follows from \cref{prop:no_degenerate_has_positive_prices,prop:degenerate_has_zero_prices}.
\end{proof}

\begin{lemma}\label{lem: flow in a subgraph}
   Fix an instance $(\agentTree, \convGraph, D, s)$ and an absorbing subgraph $\mathcal{G} \subset \agentTree \times \convGraph$. Let $\mathcal{E}_{\text{in}}(\mathcal{G})$ be the set of $(i, r) \in \mathcal{G}$ such that $(\pi(i), r) \in \mathcal{G}^c$ and $\mathcal{F}_{\text{in}}(\mathcal{G})$ be the set of $(i, r, r')$ such that $(r, r') \in \convGraph$, $(i, r) \in\mathcal{G}^c$ and $(i, r') \in\mathcal{G}$. 
   For a solution $(f, p, c)$ to be feasible, it must be that 
   \begin{equation}\label{eq: flow sum over subtree}
        \sum_{(i,r) \in \mathcal{G}}D_i^r(p_i^r) - \sum_{(i,r) \in \mathcal{G}}s_i^r \leq \sum_{(i, r) \in \mathcal{E}_{\text{in}}(\mathcal{G})}f_{i}^r + \sum_{(i, r', r) \in \mathcal{F}_{\text{in}}(\mathcal{G}) } c_i^{(r', r)} 
    \end{equation}
    Furthermore, the inequality is equality unless $p_i^r = 0$ for some $(i, r) \in \mathcal{G}$.
\end{lemma}

\begin{proof}
    We sum the Flow Conservation equations over $(i, r) \in \mathcal{G}$:
    $$\sum_{(i,r) \in \mathcal{G}}D_i^r(p_i^r) + \sum_{(i,r) \in \mathcal{G}}\left(\sum_{j \in \ch(i)} f_{j}^r - f_{i}^r\right) + \sum_{(i,r) \in \mathcal{G}}\left(\sum_{r' \in \outEdges(r)} c_i^{(r, r')} - \sum_{r' \in \inEdges(r)} c_i^{(r', r)}\right) \leq \sum_{(i,r) \in \mathcal{G}}s_i^r.$$

    By complementarity, all the inequalities are equal if $p_i^r > 0$ for all $(i, r) \in \mathcal{G}$. Now, consider the sum
    $\sum_{(i,r) \in \mathcal{G}}\left(\sum_{j \in \ch(i)} f_{j}^r - f_{i}^r\right)$. The only terms that do not cancel out are those in $\mathcal{E}_{\text{in}}(\mathcal{G})$ and $\mathcal{E}_{\text{out}}(\mathcal{G})$ and these terms appear with negative signs, and positive signs respectively. Similarly, we consider the sum $\left(\sum_{r' \in \outEdges(r)} c_i^{(r, r')} - \sum_{r' \in \inEdges(r)} c_i^{(r', r)}\right)$ and made an analogous observation.

    Hence, we notice that the terms that do not cancel out exactly correspond to the (directed) edges between $\mathcal{G}$ and $\mathcal{G}^c$. Thus, rearranging the terms, we achieve \cref{eq: flow sum over subtree}.
\end{proof}

\subsection{Identifying Degenerate Subgraphs}\label{sec: identify degen graph}
\begin{proof}[Proof of \cref{thm:computation_of_degenerate}]
    We can reduce the problem of deciding non-degeneracy in a min cut problem. Specifically, consider a graph $\mathcal{H}$ consisting of $\agentTree \times \convGraph$ and a source and sink node. The capacity of each edge in $\agentTree \times \convGraph$ is $\infty$. Add an edge from the source to node $(i, r)$ with capacity $s_i^r$ and add an edge from node $(i, r)$ to the sink with capacity $D_i^r(0)$.

    Consider a cut in this graph. Let $\mathcal{G}$ be the set of vertices with the source in this cut. If there is an edge from $\mathcal{G}$ to $\mathcal{G}^c$ in $\agentTree \times \convGraph$, this cut has value $\infty$. Hence, in any min cut finite cut $\mathcal{G}$ is downwards closed. Furthermore, the value of this cut is given by 
    $$\sum_{(i, r) \in \mathcal{G}} D_i^r(0) + \sum_{(i, r) \not\in \mathcal{G}}s_i^r = \sum_{(i, r) \in \mathcal{G}} D_i^r(0) - \sum_{(i, r) \in \mathcal{G}}s_i^r +  \sum_{(i, r) \in \agentTree \times \convGraph}s_i^r.$$ Note that there is a cut with value $\sum_{(i, r) \in \agentTree \times \convGraph}s_i^r$ (the cut with the source isolate). Hence, if the value of the min cut is $\sum_{(i, r) \in \agentTree \times \convGraph}s_i^r$,

    $$\sum_{(i, r) \in \agentTree \times \convGraph}s_i^r \leq \sum_{(i, r) \in \mathcal{G}} D_i^r(0) - \sum_{(i, r) \in \mathcal{G}}s_i^r +  \sum_{(i, r) \in \agentTree \times \convGraph}s_i^r$$ then $\sum_{(i, r) \in \mathcal{G}} D_i^r(0) - \sum_{(i, r) \in \mathcal{G}}s_i^r \geq 0$ for all downwards closed graphs. 
    
    Assume, instead that there exists a cut with smaller value. Then, for the corresponding $\mathcal{G}$, 
    $$\sum_{(i, r) \in \agentTree \times \convGraph}s_i^r > \sum_{(i, r) \in \mathcal{G}} D_i^r(0) - \sum_{(i, r) \in \mathcal{G}}s_i^r +  \sum_{(i, r) \in \agentTree \times \convGraph}s_i^r.$$ Thus, $\mathcal{G}$ is degenerate. Hence, we can decide whether the graph has a strictly degenerate subgraph and identify such a graph.

    Assume the min cut value is $\sum_{(i, r) \in \agentTree \times \convGraph}s_i^r$.
    To decide whether there exists a degenerate subgraph $\mathcal{G}$, we recompute the flow value after updating changing the capacity of the source to $(i',r')$ edge from $s_i^r$ to $s_i^r + 1$ (one at a time). Let $\tilde{s}$ be the update capacities. Assume there exist $(i, r)$ such that the min cut value is still $\sum_{(i, r) \in \agentTree \times \convGraph}s_i^r$. Consider the graph $\mathcal{G}$ corresponding to the min cut. It must be that $(i', r') \in \mathcal{G}$, otherwise that value of the min cut is 
    $$\sum_{(i, r) \in \mathcal{G}} D_i^r(0) + \sum_{(i, r) \not\in \mathcal{G}}\tilde{s}_i^r = \sum_{(i, r) \in \mathcal{G}} D_i^r(0) + \sum_{(i, r) \not\in \mathcal{G}}s_i^r + 1 \geq 1 + \sum_{(i, r) \not\in \agentTree\times \convGraph}s_i^r.$$
    We then see that 
    $$\sum_{(i, r) \in \mathcal{G}} D_i^r(0) + \sum_{(i, r) \not\in \mathcal{G}}\tilde{s}_i^r =\sum_{(i, r) \in \agentTree\times \convGraph}s_i^r$$ and $\mathcal{G} \neq \emptyset$ which implies that $\mathcal{G}$ is degenerate.
    
    We now claim that if the min cut value increases for all $(i',r')$ changed, the graph is non-degenerate. Indeed, fix $\mathcal{G}$ and select $(i, r) \in \mathcal{G}$. Then,  
    $$\sum_{(i, r) \in \mathcal{G}} D_i^r(0) + \sum_{(i, r) \not\in \mathcal{G}}s_i^r \geq \sum_{(i, r) \in \agentTree\times \convGraph}s_i^r + \epsilon.$$ Thus, $\mathcal{G}$ is non-degenerate.

    Thus, we have demonstrated a procedure to deduce where the instance is degenerate or non-degenerate and find such a subgraph. To find a minimally degenerate subgraph, we take our degenerate subgraph, remove a node with no incoming edges and check if the resulting graph is also degenerate. If so, we repeat this process. If not, we do the same for all other nodes of the graph with no incoming edges. If the graph is non-degenerate after removing any of these nodes we can conclude that no subgraph of this graph is degenerate (since no subgraph can contain all nodes with no incoming edges and be downwards closed without just being the whole graph).

    Since min cuts can be found efficiently, we conclude that this algorithm can be carried out efficiently.
\end{proof}

\subsection{Simplifying the Flow Variables}
\label{ssec:app:simplifying_flows}

\begin{lemma}\label{lem: isolate flow lemma}
    Fix an instance $(\agentTree, \convGraph, D, s)$. Consider a solution $(f, p, c)$ for the instance. The solution $(f, p, c)$ satisfies Money Conservation if and only if, for all non-root agents $i \in [n]$ and resources $r \in [m]$,
    $$p_{\pi(i)}^r\; f^r_{i} = \sum_{j \in \subTree(i)} p_{j}^r\;D_j^r(p_{j}^r)$$
    where $\subTree(i) \subseteq [n]$ are the agents who are reachable by $i$ in $\agentTree$.
\end{lemma}
\begin{proof}
    We show the forward direction and the reverse direction follows by an identical argument. We show this result inductively. Consider a leaf $i \in [n]$. Then, $\ch(i') = \emptyset$. Then, Money Conservation of $i'$ implies that  
    $$p_{\pi(i)}^r\;f_{i}^r = p_{i}^r\; D_i^r(p_{i}^r) + p_{i}^r\sum_{j \in \ch(i)} f_{j}^r = p_{i}^r\; D_{i}^r(p_{i}^r)$$
    Hence, the base case holds. Fix non-root $i \in [n]$ and assume the result holds for all children of $i$. The result then follows by applying the induction hypothesis to the Money Conservation equation for the pair $(i, r)$:
    \begin{align*}
        p_{\pi(i)}^r\;f_{i}^r &= p_{i}^r\; D_i^r(p_{i}^r) + \sum_{j \in \ch(i)} p_{i}^r\;f_{j}^r\\
        &= p_{i}^r\; D_i^r(p_{i}^r) + \sum_{j \in \ch(i)} \sum_{k \in \subTree(j) } p_{k}^r\;D_k^r(p_{k}^r)= \sum_{k \in \subTree(i)} p_{k}^r\;D_k^r(p_{k}^r)
    \end{align*}
    where we use the observation that $\bigcup_{j \in \ch(i)} \subTree(j) \cup \{i\} = \subTree(i)$ and the terms in the union are disjoint since $\agentTree$ is a tree. Thus, by induction, the result holds for all non-leaf agents.
\end{proof}

\begin{proof}[Proof of \cref{lem: problem reformulation}]
    Consider a feasible solution $(f, p, c)$. By \cref{lem: isolate flow lemma}, for any non-root $i \in [n]$, 
    $f^r_{i} = \frac{1}{p_{\parent(i)}^r}\;\sum_{j \in \subTree(i)  } p_{j}^r\;D_j^r(p_{j}^r).$ Furthermore, for any $i \in [n]$ and $r \in [m]$,
    $$p_i^r\sum_{i' \in \ch(i)}\; f^r_{i'} = \sum_{i' \in \ch(i)}\sum_{j \in \subTree(i') } p_{j}^r\;D_j^r(p_{j}^r) = \sum_{j \in \subTree(i)\setminus\{i\} } p_{j}^r\;D_j^r(p_{j}^r) $$
    
    We now replace the flow variables in the Flow Conservation equation for $(i, r)$ where $i$ is not the root to achieve
    $$ D_i^r(p_i^r) + \frac{1}{p_i^r}\sum_{j \in \subTree(i)\setminus\{i\} } p_{j}^r\;D_j^r(p_{j}^r) - \frac{1}{p_{\parent(i)}^r}\;\sum_{j \in \subTree(i) } p_{j}^r\;D_j^r(p_{j}^r) = s_i^r + \sum_{r' \in \inEdges(r)} c_i^{(r', r)} - \sum_{r' \in \outEdges(r)} c_i^{(r, r')}$$
    Rearranging terms, we get \eqref{eq: payment eq} for non-root nodes.
    A similar calculation for the root nodes yields the full \eqref{eq: payment eq}.

    We now prove the final part of the lemma.
    Assume $(p, c)$ satisfies \cref{eq: payment eq}.
    By \cref{lem: isolate flow lemma}, if there exists $(f, p, c)$ that is feasible, it must be that, for any non-root $i \in [n]$, 
    $f^r_{i} = \frac{1}{p_{\parent(i)}^r}\;\sum_{j \in \subTree(i) } p_{j}^r\;D_j^r(p_{j}^r)$. consider this solution. By \cref{lem: isolate flow lemma}, Money Conservation holds. By manipulating \cref{eq: payment eq} as we did in the previous case, we can see that Flow Conservation holds. Hence, $f$ is the unique set of flow values such that $(f, p, c)$ is a feasible solution.
\end{proof}

\section{Omitted Proof of Section \ref{sec:solution}}

\begin{proof}[Proof of \cref{thm:harmonic:main}]
    We first show that, for the solution of program \eqref{eq: optimization problem}, we can construct conversions satisfying \cref{eq:harmonic:flow,eq:harmonic:conversion}.
    Define
    \begin{equation*}
        \beta^r_i = \sum_{j \in \subTree(i) } B_j^r
        \qquad\textrm{ and }\qquad
        \sigma^r_i = \frac{\ind{i \neq \text{root}}}{p_{\parent(i)}^r}\sum_{j \in \subTree(i) } B_j^r+ s_i^r
    \end{equation*}
    
    We consider the KKT conditions for the program.
    Let $c_i^{(r, r')}$ be the dual variable for the constraint $p_i^r \leq p_i^{r'}$ for $(r, r') \in \convGraph$.
    The Lagrangian $L(p_i, c_i)$ is then given by:
    $$L(p_i, c_i) = \sum_{r \in [m]}\beta^r_i\log(p_i^r) - \sum_{r \in [m]}p_i^r\sigma_i^r +  \sum_{(r, r') \in \convGraph} c_i^{(r, r')}(p_i^r - p_i^{r'}).$$

    We now see that for every $r \in [m]$
    \begin{align*}
        \frac{\partial L(p_i, c_i)}{\partial p_i^r}
        =
        \frac{\beta_i^r}{p_i^r} - \sigma_i^r + \sum_{r'\in \outEdges(r)}c_i^{(r, r')} - \sum_{r'\in \inEdges(r)}c_i^{(r', r)}
        .
    \end{align*}
    Setting the above to $0$ gives \eqref{eq:harmonic:flow}, which is exactly what the stationarity of the optimal solution, $\frac{\partial L(p_i, c_i)}{\partial p_i^r} = 0$, implies.
    Hence, the Price Conservation equation holds.
    On the other hand, primal and dual feasibility imply that $c_i^{(r, r')} \geq 0$ and $ p_i^r \leq p_i^{r'}$ for all $(r, r') \in \convGraph$ and complementary slackness implies that one of these inequalities is tight.
    Hence, Conversion Complementarity holds, proving that the optimal primal and dual solutions of program \eqref{eq: optimization problem} satisfy \cref{eq:harmonic:flow,eq:harmonic:conversion}.

    We now observe that the KKT conditions are sufficient for optimality, because our objective is differentiable and concave, and the feasibility region is a polytope that has an interior (therefore satisfying Slater's condition).
    Since the KKT conditions are equivalent to \cref{eq:harmonic:flow,eq:harmonic:conversion}, any solution to these axioms induces primal and dual variables satisfying primal and dual feasibility, stationarity, and complementary slackness.
    This makes them optimal for \cref{eq: optimization problem}.

    We now prove the uniqueness of the optimal solution.
    Let $f^r(p_i^r)$ be the optimization objective with respect to $p_i^r$.
    Assume there are two optimal solutions $\hat p_i$ and $\tilde p_i$ that differ in some coordinate $r$.
    We notice that if $\beta_i^r > 0$ then this cannot happen: $f^r(p_i^r)$ is strongly concave with respect to $p_i^r$ and $\sum_i f^r(p_i^r)$ is concave, making any convex combination of $\hat p_i$ and $\tilde p_i$ have a strictly higher objective.
    Therefore, $\hat p_i$ and $\tilde p_i$ must differ only in coordinates where $\beta_i^r = 0$.
    Let $\hat p_i \wedge \tilde p_i$ be the pairwise minimum of the two vectors.
    Notice that this solution is also feasible.
    The difference in objective value between $\hat p_i \wedge \tilde p_i$ and $\hat p_i$ must be non-negative by optimality of $\hat p_i$:
    \begin{equation*}
        -\sum_{r: \beta_i^r = 0} \qty(\hat p_i^r - \min\{\hat p_i^r, \tilde p_i^r\}) \sigma_i^r
        \ge
        0
    \end{equation*}

    The similar argument for $\tilde p_i^r$, adding with the above gives
    \begin{equation*}
        -\sum_{r: \beta_i^r = 0} \qty(\max\{\hat p_i^r, \tilde p_i^r\} - \min\{\hat p_i^r, \tilde p_i^r\}) \sigma_i^r
        \ge
        0
    \end{equation*}

    The above proves that it must be $\sigma_i^r = 0$ for all resources $r$ where $\hat p_i$ and $\tilde p_i$ differ.
\end{proof}

\section{Omitted Proofs of Section \ref{sec:capped}}
\label{app: proofs for key lemmas}

We first restate and prove \cref{thm: main correctness thm}.

\CappedMainTheorem*

\begin{proof}[Proof of \cref{thm: main correctness thm}]
    By \cref{lem:capped:invariant} (see also that lemma's proof) the virtual bids are decreasing in $t$.
    This means that $0 \le \tilde B_i^r[t] \leq \tilde B_i^r[t-1]$, so the $\tilde{B}_i^r[t]$ form a bounded monotone sequence.
    Thus $\lim_{t \rightarrow \infty}\tilde {B}_i^r[t] = \bar{B}_i^r$ exists.
    In addition, by \cref{lem: key structural lemma; price continuity}, $p_i^r( \tilde B[t])$ converges to $p_i^r(\bar B[t])$. 
    
    We now claim that at the limit, \cref{eq:invariant} holds in equality: for all $(i, r)$,
    \begin{equation*}
        \bar B_i^r
        =
        p_i^r(\bar B) D_i^r\qty(p_i^r(\bar B))
    \end{equation*}

    Notice that by the update rule \cref{eq:update_rule}, we have that
    \begin{equation*}
        \tilde B_i^r[t]
        =
        p_i^r(\tilde B_i^r[t-1]) D_i^r\qty(p_i^r(\tilde B_i^r[t-1]))
    \end{equation*}
    If the function $D_i^r(\cdot)$ is continuous for positive prices, then, by taking the limit of the above equality as $t \to \infty$, we get exactly what we want (also using the fact that the algorithm produces positive prices \cref{sec:app:positive_prices_BDA}).
    We now prove continuity by the two basic properties of the demand functions:
    if $p \mapsto D(p)$ is weakly decreasing and $p \mapsto p\,D(p)$ is weakly increasing, then $D(\cdot)$ is continuous.
    Fix a $p > 0$ and $\e > 0$; then
    \begin{equation*}
        D(p) \ge D(p + \e)
        \quad\text{ and }\quad
        p\,D(p) \ge (p + \e) \, D(p + \e)
    \end{equation*}
    which proves that
    \begin{equation*}
        \frac{p}{p + \e} D(p)
        \le
        D(p + \e)
        \le
        D(p)
    \end{equation*}
    which proves that $\lim_{\e \to 0^+} D(p + \e) = D(p)$.
    On the other hand, for $p > p - \e > 0$ we can similarly prove that
    \begin{equation*}
        D(p)
        \le
        D(p - \e)
        \le
        \frac{p}{p - \e} D(p)
    \end{equation*}
    which proves that $\lim_{\e \to 0^+} D(p - \e) = D(p)$.
    Combining both limits proves the continuity for positive prices.
\end{proof}

Next we include the omitted proofs of the lemma of that section.

\begin{lemma}\label{lem: topkis application to supply}
    Fix $B$ and $p_{\pi(i)}$. Let $\mathcal{H} = \{x \in \R^m: C^r \geq x^r\}$ for some constants  $C^r$. Let $$F_i(p)
            =
            \sum_{r \in [m]} \qty(
                \log(p^{r})
                \sum_{j \in \subTree(i) }  B_j^{r}
                -
                p^{r} \qty(
                    s_{i}^{r}
                    +
                    \frac{1}{p_{\pi(i)}^{r}} \sum_{j \in \subTree(i) }  B_j^{r}
                ) ).$$ 
            Then, $\argmax_{p \in \mathcal{K} \cap \mathcal{H}} F_i(p) \preceq \argmax_{p \in \mathcal{K}} F_i(p)$.
\end{lemma}
\begin{proof}
    Consider $$F_i(p; z)
            =
            \sum_{r \in [m]} \qty(
                \log(p^{r})
                \sum_{j \in \subTree(i) } B_j^{r}
                -
                p^{r} \qty(
                    s_{i}^{r} - z^r
                    +
                    \frac{1}{p_{\pi(i)}^{r}} \sum_{j \in \subTree(i) } B_j^{r}
                )
            ).$$ 
    
    We first apply \cref{thm:topkis} to deduce that $\argmax_{p \in \mathcal{K}} F_i(p; z)$ is weakly increasing in $z$. As before, $\mathcal{K}$ is a lattice. Furthermore, all the cross derivatives of $p^r$ and $p^{r'}$ are zero. Finally, $\frac{\partial F_i(p; z)}{\partial p^r z^{r'}} = \ind{r = r'} \geq 0$.

    We consider the KKT conditions applied to  problem $\argmax_{p \in \mathcal{K} \cap \mathcal{H}} F_i(p)$. Fix $\lambda$, the optimal dual variables for the constraint in 
    $\mathcal{H}$.

    The Lagrangian is given by 
    $$L(p, \theta, \lambda) = \sum_{r \in [m]} \qty(
                \log(p^{r})
                \sum_{j \in \subTree(i) } \hat B_j^{r}
                -
                p^{r} \qty(
                    s_{i}^{r}
                    +
                    \frac{1}{p_{\pi(i)}^{r}} \sum_{j \in \subTree(i) } \hat B_j^{r}
                ) ) + \sum_{(r, r') \in \convGraph}\theta^{(r, r')}(p^{r'} - p^r) + \lambda^r(C^r - p^r).$$ Stationarity implies that 
    $$\frac{\partial L(p, \theta, \lambda)}{\partial p^r} = \sum_{j \in \subTree(i) } \hat B_j^{r}\left(\frac{1}{p^r} - \frac{\ind{i \text{ is not root}}}{p_{\parent(i)}^r}\right) - (s_i^r + \lambda^r) - \sum_{(r', r) \in \convGraph}\theta^{(r, r')} + \sum_{(r, r') \in \convGraph}\theta^{(r, r')}.$$ Notice that these are exactly the stationarity conditions for $\argmax_{p \in \mathcal{K}} F_i(p; -\lambda)$. Furthermore, the optimal primal and dual values for $\argmax_{p \in \mathcal{K} \cap \mathcal{H}} F_i(p)$ satisfy feasibility and complimentary slackness for $\argmax_{p \in \mathcal{K}} F_i(p; -\lambda)$ (as these are a subset of the conditions for $\argmax_{p \in \mathcal{K} \cap \mathcal{H}} F_i(p)$). Hence, we deduce that $\argmax_{p \in \mathcal{K} \cap \mathcal{H}} F_i(p) = \argmax_{p \in \mathcal{K}} F_i(p; -\lambda)$ for some $\lambda \geq 0$. Finally, as we have shown, 
    $$\argmax_{p \in \mathcal{K} \cap \mathcal{H}} F_i(p) = \argmax_{p \in \mathcal{K}} F_i(p; -\lambda) \preceq \argmax_{p \in \mathcal{K}} F_i(p; 0) = \argmax_{p \in \mathcal{K}} F_i(p).$$
\end{proof}

\subsection{Proof that Algorithm \ref{algo:BDA} produces Positive Prices} \label{sec:app:positive_prices_BDA}

Recall that $\mathcal{B}$ is the set of budgets $\hat B$ such that $\sum_{j \in \subTree(i) } \hat{B}_j^r > 0$ for all $(i, r)$ such that $s_i^r = 0$. In this section, we assume that  $\tilde B[t],\bar B \in \mathcal{B}$ for all $t$.  In \cref{subsec: justify uniqueness}, we argue that this is without loss of generality.

\begin{lemma}\label{lem:prices are positive}
    Consider a non-degenerate instance of the capped harmonic problem. In the Budget Descent Algorithm (\cref{algo:BDA}), for all $(i, r) \in [n] \times [m]$, $p_i^r(\tilde B[t]) > 0$ at all $t \in\N$ and $p_i^r(\bar{B}) > 0$ where $\bar{B} = \lim_{t\rightarrow \infty}\tilde B[t]$.
\end{lemma}
\begin{proof}   
    We note that since $p_i^r(\cdot)$ is continuous and non-decreasing as a function of $\hat{B}_{i'}^{r'}$ for all $(i', r')$, $p_i^r(\tilde B[t])$ is a non-increasing sequence wrt. $t$. Let $$\mathcal{G}_{\epsilon} = \{(i, r) \in [n]\times [m]: p_i^r(\tilde B[t]) < \epsilon \;\; \text{for some }t \in \N\}$$ be the set of prices that eventual become less than some $\epsilon > 0$. We note that $\mathcal{G}_{\epsilon'} \subset \mathcal{G}_{\epsilon}$ for $\epsilon' < \epsilon$.
    Hence, the sets $\mathcal{G}_{\epsilon}$ form a collection of nesting sets. Hence, this set cannot change for infinitely many $\epsilon > 0$. Hence, there exists $\epsilon > 0$ such that $\mathcal{G}_{\epsilon'} = \mathcal{G}_{\epsilon}$ for all $0 < \epsilon' \leq \epsilon$. We will now show that for this $\epsilon$, $\mathcal{G}_{\epsilon} = \emptyset$. 
    
    Note that there must exists $t$ such that the set of $(i, r)$ such that $p_i^r(\tilde B[t]) < \epsilon$ is $\mathcal{G}_{\epsilon}$. By the same argument as \cref{lem:propagation_of_zero_prices}, the set of $(i, r)$ such that $p_i^r(\tilde B[t]) < \epsilon$ can be taken to be absorbing, i.e., every $(i', r')$ reachable by $(i, r)$ satisfies $p_i^r(\tilde B[t]) < \epsilon$. Hence, we assume $\mathcal{G}_{\epsilon}$ is absorbing. Furthermore, the conversions into this graph must be zero (otherwise, the prices of the node the conversion is coming from must also be small). We apply \cref{lem: flow in a subgraph} to conclude that
    $$\sum_{(i,r) \in \mathcal{G}_{\epsilon}}\frac{\tilde B_i^r[t]}{p_i^r(\tilde B[t])} - \sum_{(i,r) \in \mathcal{G}_{\epsilon}}s_i^r \leq \sum_{(i, r) \in \mathcal{E}}f_{i}^r = \sum_{(i, r) \in \mathcal{E}_{\text{in}}(\mathcal{G}_\epsilon) : i \neq\text{root}}f_{i}^r $$ where we recall that $\mathcal{E}_{\text{in}}(\mathcal{G}_\epsilon)$ is the set of $(i, r) \in \mathcal{G}_{\epsilon}$ such that $(\pi(i), r) \not\in \mathcal{G}_{\epsilon}$. We note that the last inequality follows from the fact that $f_{i}^r = 0$ for $i$ the root.

    By the definition of $\mathcal{E}_{\text{in}}(\mathcal{G}_\epsilon)$, we have that $(i, r) \in \mathcal{E}_{\text{in}}(\mathcal{G}_\epsilon)$ implies that  $p_{\pi(i)}^r(\tilde B[t]) \geq \epsilon > 0$. Thus, \cref{lem: isolate flow lemma} now tells us that 
    \begin{align*}
        \sum_{(i,r) \in \mathcal{G}_{\epsilon}}\frac{\tilde B_i^r[t]}{p_i^r(\tilde B[t])} - \sum_{(i,r) \in \mathcal{G}_{\epsilon}}s_i^r &\leq \sum_{(i, r) \in \mathcal{E}_{\text{in}}(\mathcal{G}_\epsilon) : i \neq \text{root}} \frac{1}{p_{\pi(i)}^r} \sum_{j \in \subTree(i)}  \tilde B_j^r[t]\\
        &\leq \frac{1}{\epsilon}\sum_{(i, r) \in \mathcal{E}_{\text{in}}(\mathcal{G}_\epsilon) : i \neq \text{root}} \sum_{j \in \subTree(i)}  \tilde B_j^r[t]\\
        &\leq  \frac{1}{\epsilon} \sum_{(i, r) \in \mathcal{G}_{\epsilon}}  \tilde B_j^r[t] 
    \end{align*} where the last inequality follows from noticing that for all $(j, r) \in \mathcal{G}$, there is a unique $i$ such that $j \in \subTree(i)$ and $(i, r) \in \mathcal{E}_{\text{in}}(\mathcal{G}_\epsilon)$. Note that this bound will hold for all $t' > t$. 
    
    We now claim that $\tilde B_i^r[t] = o(1)$ for all $(i, r) 
    \in \mathcal{G}_\epsilon$. Indeed, if this were not the case, $\tilde B_i^r[t] > c$ for all $t$ for some constant $c$. We, however, know that $p_i^r(\tilde B[t]) \rightarrow 0$ because $(i, r) \in \mathcal{G}_{\epsilon'}$ for all $\epsilon' < \epsilon$. Thus $\frac{\tilde B_i^r[t]}{p_i^r(\tilde B[t])} \geq \frac{c}{p_i^r(\tilde B[t])} \rightarrow \infty$. This cannot hold since by the inequality above, 
    $$\frac{\tilde B_i^r[t]}{p_i^r(\tilde B[t])} \leq \sum_{(i,r) \in \mathcal{G}_{\epsilon}}\frac{\tilde B_i^r[t]}{p_i^r(\tilde B[t])} \leq   \frac{1}{\epsilon} \sum_{(i, r) \in \mathcal{G}_{\epsilon}}  \tilde B_j^r[t] + \sum_{(i,r) \in \mathcal{G}_{\epsilon}}s_i^r \leq \frac{1}{\epsilon} \sum_{(i, r) \in \mathcal{G}_{\epsilon}}  \tilde B_j^r[0] + \sum_{(i,r) \in \mathcal{G}_{\epsilon}}s_i^r < \infty,$$ proves the claim.

    We now note that by \cref{lem:capped:invariant}, $\tilde B_i^r[t] \geq p_i^r(\tilde B[t])\cdot D_i^r(p_i^r(\tilde B[t]))$. Hence, 
    $$\sum_{(i,r) \in \mathcal{G}_{\epsilon}}D_i^r(p_i^r(\tilde B[t])) - \sum_{(i,r) \in \mathcal{G}_{\epsilon}}s_i^r \leq \sum_{(i,r) \in \mathcal{G}_{\epsilon}}\frac{\tilde B_i^r[t]}{p_i^r(\tilde B[t])} - \sum_{(i,r) \in \mathcal{G}_{\epsilon}}s_i^r \leq \frac{o(1)}{\epsilon}.$$ Taking the limit on both sides of the inequality, we have that 
    $$\sum_{(i,r) \in \mathcal{G}_{\epsilon}}D_i^r(0) \leq \sum_{(i,r) \in \mathcal{G}_{\epsilon}}s_i^r.$$ Since, the instance is non-degenerate, it must be that $\mathcal{G}_{\epsilon} = \emptyset$.
\end{proof}

\subsection{Justification of Uniqueness Assumption}\label{subsec: justify uniqueness}

We now justify our assumption that $\tilde B[t],\bar B \in \mathcal{B}$ in \cref{subsec: proof of structural result} for all $t$. Given a virtual budget vector $\hat B$, we will say that $(i, r)$ is a \emph{pass-through node} with respect to $\hat B$ if $\hat B^r_j = 0$ for all $j \in \subTree(i)$ and $s_i^r = 0$.

\begin{proposition}
    In the general demand setting of the Budget Descent Algorithm, $(i, r)$ is a pass-through node wrt. $\tilde B[t]$ for some $t$ (or $\bar B$) if and only if $(i, r)$ is a pass-through node with respect to $\tilde B[0]$.
\end{proposition}
\begin{proof}
   We first consider the reverse case. If $(i, r)$ is a pass-through node wrt. $\tilde B[0]$, then $\tilde B_j^r[t] \leq \tilde B_j^r[0] = 0$ for all $j \in \subTree(i)$ and $s_i^r = 0$. Hence, $(i, r)$ is a pass-through node wrt. $\tilde B[t]$ (and $\bar B$). 

   We now assume $(i, r)$ is a pass-through node wrt. $\tilde B[t]$ for some $t$. By \cref{lem:prices are positive}, we have that $p_j^r(\tilde B[t]) > 0$ for all $j \in \subTree(i)$. Furthermore, $$p_j^r(\tilde B[t])\cdot  D_j^r(p_j^r(\tilde B[t])) = \tilde B_j^r[t+1]  \leq \tilde B_j^r[t] = 0.$$ Hence, $D_j^r(p) = 0$ for $p = p_j^r(\tilde B[t])$.  Since $D_j^r$ is non-increasing, $D_j^r(x) = 0$ for $x > p$. Hence, $\tilde B_j^r[0] = \lim_{p \rightarrow \infty} p D_j^r(p) =  0$.   Thus, $(i, r)$ is a pass-through node wrt. $\tilde B[0]$. The same argument works for $\bar B$ instead of $\tilde B[t]$. 
   
\end{proof}

By this result, we can define $\mathcal{Q} \subset [n] \times [m]$, the set of all pass-through nodes at each iteration of the Budget Descent Algorithm. Consider $(i, r) \in \mathcal{Q}$ and iteration $t$. Let $p \in \R^{n \times m}$ be a set of optimal prices for harmonic instance with virtual budget $\tilde B[t]$.

We now consider the problem in \cref{eq: optimization problem} which we use to compute the prices $p$:
\begin{align*}
    \max_{p_i^1, \cdots, p_i^m \geq 0}&  \sum_{r' \in [m]}\log(p_i^{r'})\sum_{j \in \subTree(i) }  \tilde B_j^{r'}[t]- \sum_{{r'} \in [m]}p_i^{r'}\left(\frac{\ind{i \neq \text{root}}}{p_{\parent(i)}^{r'}}\sum_{j \in \subTree(i) }  \tilde B_j^{r'}[t] + s_i^{r'}\right)\\
    & \text{s.t.} \quad p_i^{r''}\geq p_i^{r'}\;\; \forall (r'', r') \in \convGraph
\end{align*}

We observe that if $(i, r)$ is a pass-through node, the $(i, r)$ term of the sum in the objective function is 0. Hence, the objective is unaffected by the value of $p_i^r$. We now consider an alternative problem where we remove the dependence on the $p_i^r$ variable entirely by considering a new conversion graph $\convGraph'$ where we remove $r$ but add an edge between each $r' \in \inEdges(r)$ and $r'' \in \outEdges(r)$. We now consider the problem:
\begin{align*}
    \max_{(p^{r'}_i)_{r' \neq r} \succeq 0}&  \sum_{r' \in [m]: r' \neq r}\log(p_i^{r'})\sum_{j \in \subTree(i) }  \tilde B_j^{r'}[t]- \sum_{r' \in [m]: r' \neq r}p_i^{r'}\left(\frac{\ind{i \neq \text{root}}}{p_{\parent(i)}^{r'}}\sum_{j \in \subTree(i) }  \tilde B_j^{r'}[t] + s_i^{r'}\right)\\
    & \text{s.t.} \quad p_i^{r''}\geq p_i^{r'}\;\; \forall (r'', r') \in \convGraph'
\end{align*}

We claim that if $p_i$ is an optimal price vector in the original problem, then $(p^{r'}_i)_{r' \neq r}$ is optimal in the new problem. This holds because the objective functions for the problems are the same and any value feasible $p_i$ is the original problem implies that $(p^{r'}_i)_{r' \neq r}$ is feasible in the new problem (and vice versa). On the other hand, if $(p^{r'}_i)_{r' \neq r}$ is optimal in the new problem, $p_i$ is an optimal price vector in the original for any $p_i^r$ where $$\max_{r' \in \outEdges(r)}p_i^{r'}  \leq p_i^r \leq \min_{r' \in \inEdges(r)}p_i^{r'} .$$ We note that $\max_{r' \in \outEdges(r)} p_i^{r'} \leq \min_{r' \in \inEdges(r)}p_i^{r'}$ by the construction of $\convGraph'$.

We note that $p_i^r$ does not appear in any of the optimization problem for any of its children because for any $k$, a child of $i$,  $0 \leq \sum_{j \in \subTree(k) } \tilde B_j^{r}[t] \leq \sum_{j \in \subTree(i) }  \tilde B_j^{r}[t]= 0$. Hence, to find the optimal prices for all the agents, it will suffice to eliminate the path-through node $(i, r)$ as we have and set the price $p_i^r$ after computing all other prices. We can repeat this process for all $(i,r) \in \mathcal{Q}$. 

Hence, we can compute the prices for all non-pass-through nodes (uniquely). Thus, all the arguments in \cref{subsec: proof of structural result} will go through without issue for the optimization problems induced by the non-pass-through nodes. Then, at the end of the process, we can compute an optimal price for the pass-through nodes at any iteration. Furthermore, since the set of pass-through nodes is fixed from the initialization of the algorithm, we can assume WLOG that the instance has no pass-through nodes.

\section{Proof of Convergence Rate of the Budget Descent Algorithm}\label{app: convergence rate}
In this section, we prove \cref{thm: main convergence result}:
\ConvergenceResult*

This result will follow from two result. The first result bounds the distance between the limiting budgets in the perturbed problem $\bar{B}[\delta] = \lim_{t\rightarrow \infty} \tilde{B}[t; \delta]$ and the limiting budgets in the original problem $\bar{B} = \bar{B}[0]$:

We first assume the correctness of \cref{lem: bounding limiting values,lem: bounding limiting value and iterate}.

\begin{proof}
    We apply \cref{lem: bounding limiting values,lem: bounding limiting value and iterate} to see that for any $\delta$ and $t$,
    \begin{align*}
        \sum_{i,r}\left|D_i^r(p_i^r(\bar{B})) - D_i^r(p_i^r(\tilde{B}[t; \delta]))\right| &\leq \sum_{i,r}\left|D_i^r(p_i^r(\bar{B})) - D_i^r(p_i^r(\bar{B}[\delta]))\right|\\
        &\quad\quad\quad+ \sum_{i,r}\left|D_i^r(p_i^r(\bar{B}[\delta])) - D_i^r(p_i^r(\tilde{B}[t; \delta]))\right|\\
        &\leq \left(\max_{i,r}\frac{\tilde{B}_i^r[0; \delta]}{\bar{B}_i^r[\delta]}-1\right)\left(1 - \frac{\delta}{1 + \delta}\right)^t\left(\sum_{i,r} s_i^r\right) + \delta\sum_{i,r}\frac{\overline{E_i^r}}{p_i^r(\bar{B})} 
    \end{align*}

    For any $t \geq \frac{\log(\delta)}{\log\left(1 - \frac{\delta}{1 + \delta}\right)}$,

    \begin{align*}
        \sum_{i,r}\left|D_i^r(p_i^r(\bar{B})) - D_i^r(p_i^r(\tilde{B}[t; \delta]))\right| &\leq \delta\left[\left(\max_{i,r}\frac{\tilde{B}_i^r[0; \delta]}{\bar{B}_i^r[\delta]}-1\right)\left(\sum_{i,r} s_i^r\right) + \sum_{i,r}\frac{\overline{E_i^r}}{p_i^r(\bar{B})} \right]
    \end{align*}

    It then suffices to take $\delta = \epsilon$. We note that $\frac{1+\epsilon}{\epsilon}\log\left(\frac{1}{\epsilon}\right) \geq \frac{\log(\epsilon)}{\log\left(1 - \frac{\epsilon}{1 + \epsilon}\right)}$ for small $\epsilon$.    
\end{proof}

We now prove \cref{lem: bounding limiting values,lem: bounding limiting value and iterate}.

\subsection{Bounding Distance Between Limiting Budgets} Towards proving \cref{lem: bounding limiting values}, we argue that the virtual budget in iteration $t$ of the algorithm are monotone with respect to $\delta$:
\begin{lemma}
    For all $\delta \leq \delta'$, $\tilde{B}[t; \delta] \preceq \tilde{B}[t; \delta']$. Hence, $\bar{B}[\delta] \preceq \bar{B}[\delta']$.
\end{lemma}
\begin{proof}
We show this inductively. The base case is true since $\tilde{B}_i^r[0; \delta] = \max_{p} E_i^r(p; \delta) = (1+\delta)\; \overline{E_i^r}$. Assume $\tilde{B}[t; \delta] \preceq \tilde{B}[t; \delta']$.  Then,
\begin{align*}
    \tilde{B}_i^r[t+1; \delta] &= E_i^r(p_i^r(\tilde{B}[t; \delta]); \delta)\\
    &\leq E_i^r(p_i^r(\tilde{B}[t; \delta']); \delta)\\
    &\leq E_i^r(p_i^r(\tilde{B}[t; \delta']); \delta') = \tilde{B}_i^r[t+1; \delta']
\end{align*} where we apply the monotonicity of $E_i^r$ and $p_i^r$.
\end{proof}

We now prove the result.
\BoundingLimitings*
\begin{proof}
    We note that $\bar{B}[\delta]$ satisfies $E_i^r(\bar{B}[\delta]; \delta) =  \bar{B}_i^r[\delta]$.
    We see that \begin{align*}
        0 \leq D_i^r(p_i^r(\bar{B})) - D_i^r(p_i^r(\bar{B}[\delta])) &= \frac{\bar{B}_i^r}{p_i^r(\bar{B})} - \frac{\bar{B}_i^r[\delta] - \delta\overline{E_i^r}}{p_i^r(\bar{B}[\delta])}\\
        &= \frac{\bar{B}_i^r}{p_i^r(\bar{B})} - \frac{\bar{B}_i^r[\delta]}{p_i^r(\bar{B}[\delta])} + \frac{\delta\overline{E_i^r}}{p_i^r(\bar{B}[\delta])}\\
        \sum_{i,r}\left|D_i^r(p_i^r(\bar{B})) - D_i^r(p_i^r(\bar{B}[\delta]))
        \right|&= \sum_{i,r}\frac{\bar{B}_i^r}{p_i^r(\bar{B})} - \sum_{i,r}\frac{\bar{B}_i^r[\delta]}{p_i^r(\bar{B}[\delta])} + \sum_{i,r}\frac{\delta\overline{E_i^r}}{p_i^r(\bar{B}[\delta])}\\
        &=\sum_{i,r}\frac{\delta\overline{E_i^r}}{p_i^r(\bar{B}[\delta])}
    \end{align*} where we use that fact that $\sum_{i,r}\frac{B_i^r}{p_i^r(B)} = \sum_{i,r}s_i^r$ for all $B$.
     The result then follows from $\bar{B}_i^r = \bar{B}_i^r[0] \leq \bar{B}_i^r[\delta]$ for $\delta > 0$. 
\end{proof}

\subsection{Bounding Distance Between Limiting Budget and \texorpdfstring{$t$}{t}-th iterate budgets}

To prove \cref{lem: bounding limiting value and iterate}, we will track our algorithms progress by bounded the maximum ratio between  $\tilde{B}_i^r[t; \delta]$ and $\bar{B}_i^r[\delta]$. Recall that $\mu_t(\delta)$ is defined to be the smallest value such that $$\tilde{B}[t; \delta] \preceq \mu_t(\delta)\cdot \bar{B}[\delta].$$ We will now show that $\mu_t(\delta)$ converges to 1 exponentially fast. First, we need the following proposition:

\begin{proposition}
    For any demand functions satisfying \cref{assumption:demands} and $\alpha \geq 1$, $$E_i^r(\alpha p; \delta) \leq \left(\alpha - (\alpha - 1)\cdot \frac{\delta}{1 + \delta}\right) E_i^r(p; \delta).$$
    
\end{proposition}
\begin{proof}
    We note that $E_i^r(\alpha p) = \alpha p \cdot D_i^r(\alpha p) \leq \alpha p \cdot D_i^r(p) = \alpha E_i^r(p)$ where we use the fact that $\alpha \geq 1$ and $D_i^r$ is decreasing. We now see that
    \begin{align*}
        E_i^r(\alpha p; \delta) &= E_i^r(\alpha p) + \delta \overline{E_i^r}\\
        &\leq \alpha E_i^r( p) + \alpha\delta \overline{E_i^r} - (\alpha -  1)\delta \overline{E_i^r}\\
        &\leq \alpha  E_i^r(p; \delta)  - (\alpha -  1)\delta \overline{E_i^r}\\
        &\leq \left(\alpha    - (\alpha -  1)\frac{\delta \overline{E_i^r}}{E_i^r( p; \delta)}\right)E_i^r( p; \delta)\\
        &\leq \left(\alpha    - (\alpha -  1)\cdot \frac{\delta}{1+\delta}\right)E_i^r( p; \delta)\\
    \end{align*}
\end{proof}

We now show the exponential convergence rate of $\mu_t(\delta)$:
\begin{lemma}
     For all $t$ and $\delta > 0$, $$\mu_t(\delta) - 1 \leq \left(\max_{i,r}\frac{\tilde{B}_i^r[0; \delta]}{\bar{B}_i^r[\delta]} - 1\right)\left(1 - \frac{\delta}{1 + \delta}\right)^t.$$
 \end{lemma}
 \begin{proof}
     We see that following
     \begin{align*}
        \tilde{B}[t+1; \delta] &= E_i^r(p_i^r(\tilde{B}[t; \delta]); \delta)\\
         &\leq E_i^r(p_i^r(\mu_t(\delta)\cdot \bar{B}[\delta]); \delta)\\
         &= E_i^r(\mu_t(\delta)\cdot p_i^r(\bar{B}[\delta]); \delta)\\
         &\leq  \left(\mu_t(\delta)    - (\mu_t(\delta) -  1)\frac{\delta}{1 + \delta}\right)E_i^r( p(\bar{B}[\delta]); \delta)\\
         &\leq  \left(\mu_t(\delta)    - (\mu_t(\delta) -  1)\frac{\delta}{1 + \delta}\right)\bar{B}_i^r[\delta]
     \end{align*} where we use the monotonicity of $p_i^r$ and $E_i^r$ and the homogeneity of $p_i^r$ ie. $p_i^r(\alpha B) = \alpha p_i^r( B)$ (\cref{prop: positively homogeneous}).

     Note that by the definition of $\mu_{t+1}(\delta)$, $$\mu_{t+1}(\delta) \leq \mu_t(\delta)    - (\mu_t(\delta) -  1)\frac{\delta}{1 + \delta}.$$ Subtracting 1 from both sides of the inequality and rearranging, we see that
     $\mu_{t+1}(\delta)$, $$(\mu_{t+1}(\delta) - 1) \leq (\mu_t(\delta) -  1)\left(1-\frac{\delta}{1 + \delta}\right).$$
     By induction we can now see that 
     $$(\mu_{t}(\delta) - 1) \leq (\mu_0(\delta) -  1)\left(1-\frac{\delta}{1 + \delta}\right)^t.$$ 
     It is immediate to see that $$\mu_0(\delta)= \max_{i,r}\frac{\tilde{B}_i^r[0; \delta]}{\bar{B}_i^r[\delta]}.$$
 \end{proof}

We now use this exponential convergence rate bound to show  \cref{lem: bounding limiting value and iterate}.
\BoundingIterate*
\begin{proof}
We see that
    \begin{align*}
        D_i^r(p_i^r(\bar{B}_i^r[\delta])) - D_i^r(\tilde{B}_i^r[t;\delta]) &= \frac{\bar{B}_i^r[\delta] - \delta \overline{E_i^r}}{p_i^r(\bar{B}[\delta])} - \frac{\tilde{B}_i^r[t+1; \delta] - \delta\overline{E_i^r}}{p_i^r(\tilde{B}[t;\delta])}\\
        &\leq \frac{\bar{B}_i^r[\delta] - \delta\overline{E_i^r}}{p_i^r(\bar{B}[\delta])} - \frac{\bar{B}_i^r[\delta] - \delta\overline{E_i^r}}{p_i^r(\tilde{B}[t;\delta])}\\
        &= (\bar{B}_i^r[\delta] - \delta\overline{E_i^r})\left(\frac{p_i^r(\tilde{B}[t;\delta]) - p_i^r(\bar{B}[\delta])}{p_i^r(\bar{B}[\delta]) \cdot p_i^r(\tilde{B}[t;\delta])}\right)\\
     \end{align*}
         
We now apply the definition of $\mu_t(\delta)$ and the monotonicity of $p_i^r$:

         \begin{align*}
        0 \leq D_i^r(p_i^r(\bar{B}_i^r[\delta])) - D_i^r(p_i^r(\tilde{B}_i^r[t;\delta])) &\leq (\bar{B}_i^r[\delta] - \delta\overline{E_i^r})\left(\frac{\mu_t(\delta) \cdot p_i^r(\bar{B}[\delta]) - p_i^r(\bar{B}[\delta])}{p_i^r(\bar{B}[\delta]) \cdot p_i^r(\tilde{B}[t;\delta])}\right)\\
        &\leq (\bar{B}_i^r[\delta] - \delta\overline{E_i^r})\left(\frac{\mu_t(\delta) - 1}{p_i^r(\tilde{B}[t;\delta])}\right)\\
        &\leq (\bar{B}_i^r[\delta] - \delta\overline{E_i^r})\left(\frac{\mu_t(\delta) - 1}{p_i^r(\bar{B}[\delta])}\right)\\
        &= D_i^r(p_i^r(\bar{B}_i^r[\delta])) \cdot (\mu_t(\delta) - 1)\\
        &= D_i^r(p_i^r(\bar{B}_i^r[\delta])) \cdot\left(\max_{i,r}\frac{\tilde{B}_i^r[0; \delta]}{\bar{B}_i^r[\delta]}-1\right)\left(1 - \frac{\delta}{1 + \delta}\right)^t.
    \end{align*}
We now sum over $i,r$ and see that 

\begin{align*}
        \sum_{i,r}\left|D_i^r(p_i^r(\bar{B}_i^r[\delta])) - D_i^r(\tilde{B}_i^r[t;\delta])\right| &\leq \left(\max_{i,r}\frac{\tilde{B}_i^r[0; \delta]}{\bar{B}_i^r[\delta]}-1\right)\left(1 - \frac{\delta}{1 + \delta}\right)^t\sum_{i,r}D_i^r(p_i^r(\bar{B}_i^r[\delta]))\\
        &\leq \left(\max_{i,r}\frac{\tilde{B}_i^r[0; \delta]}{\bar{B}_i^r[\delta]}-1\right)\left(1 - \frac{\delta}{1 + \delta}\right)^t\left(\sum_{i,r} s_i^r - \delta\sum_{i,r}\frac{ \overline{E_i^r}}{p_i^r(\bar{B}[\delta])} \right)\\
        &\leq \left(\max_{i,r}\frac{\tilde{B}_i^r[0; \delta]}{\bar{B}_i^r[\delta]}-1\right)\left(1 - \frac{\delta}{1 + \delta}\right)^t\left(\sum_{i,r} s_i^r \right)
    \end{align*}
\end{proof}

\section{Simulations on the Allocation of AI Accelerators}
\label{sec:simulations}

Finally, we demonstrate the convergence of our algorithms on a real dataset obtained from the market used by Google to allocate AI accelerators, such as TPUs and GPUs. The resource types correspond to TPU chips of various generations and in different geographic locations. When running a machine learning (ML) training job, an engineering team specifies the types of TPU chips it needs: a specific generation and location. However, a team also has the option of specifying a ``global'' chip, meaning they are flexible and can run their jobs on any chip of that generation, regardless of its location. This flexibility leads to a very simple, star-shaped conversion graph: for each generation $X$ and geographic location $L$, we have a directed conversion $(X, L) \rightarrow (X, \gl)$, where $\gl$ stands for the global location. While conversions between generations are theoretically possible, they are rare in practice, so we solve a separate market clearing for each chip generation.
We analyze four market runs (one corresponding to each chip generation) in a single snapshot of the market.

\begin{figure}[htbp]
    \centering
    \begin{subfigure}[b]{\figwidth}
        \centering
        \includegraphics[width=\linewidth]{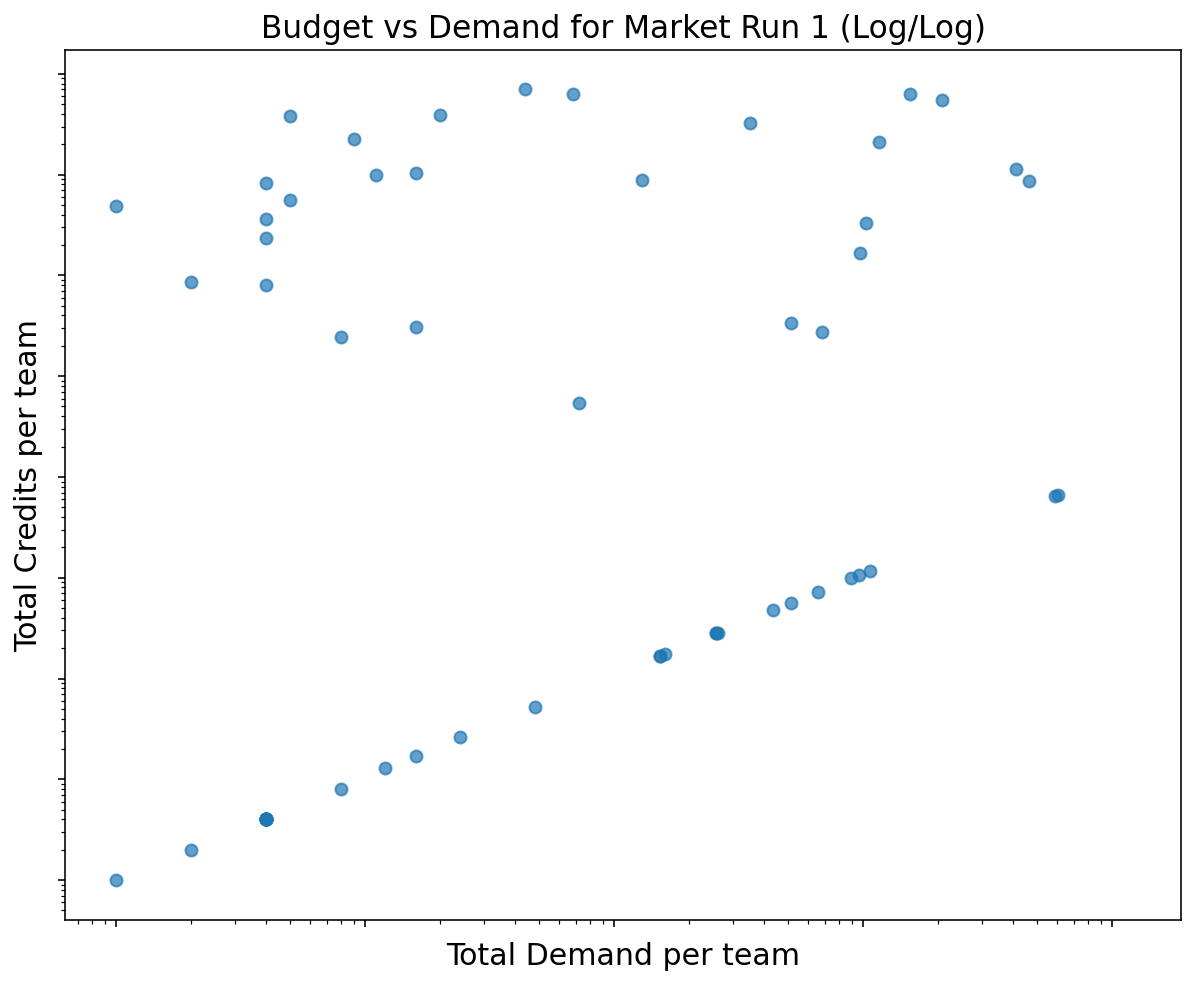}
    \end{subfigure}
    \hfill
    \begin{subfigure}[b]{\figwidth}
        \centering
        \includegraphics[width=\linewidth]{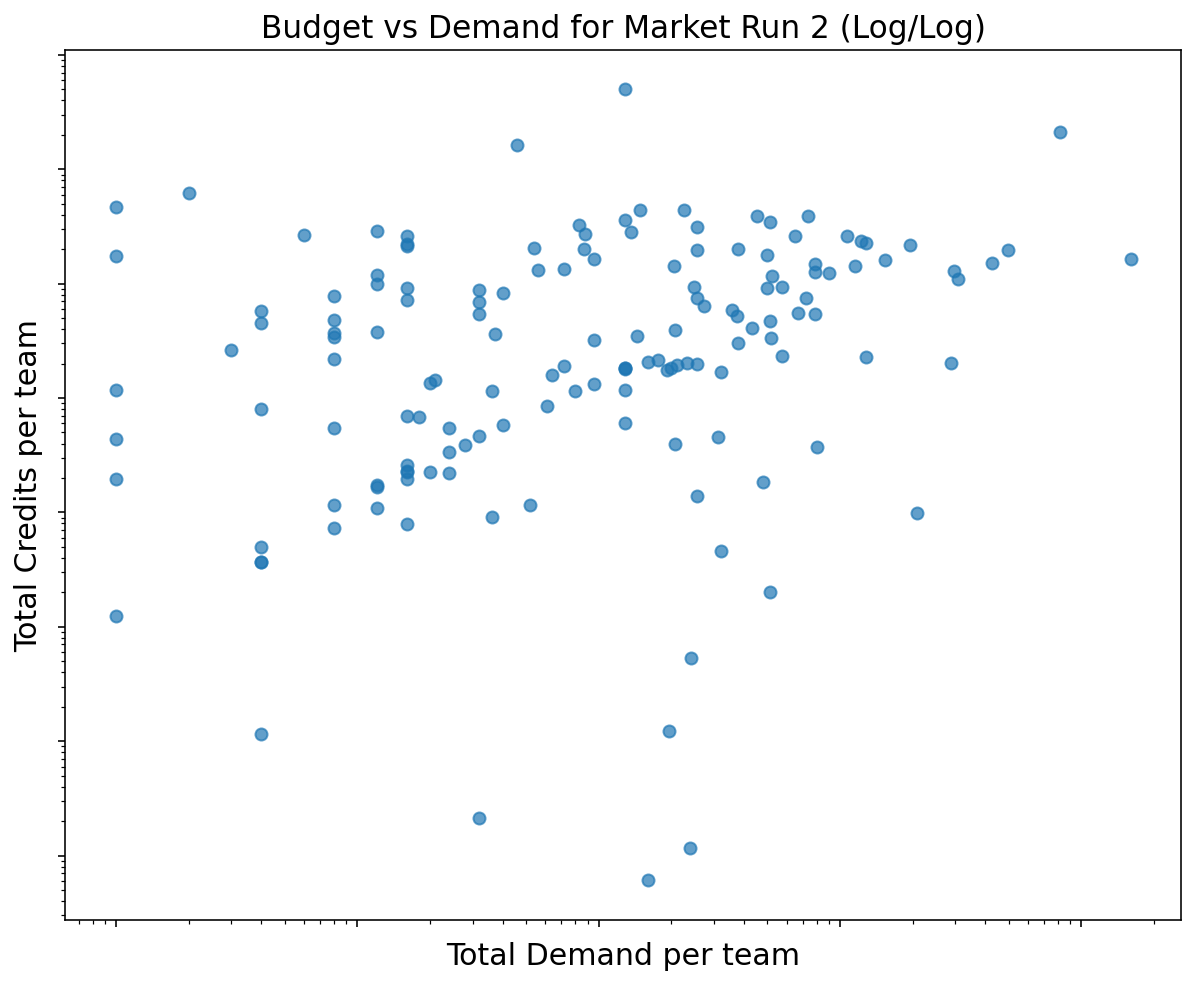}
    \end{subfigure}

    \vspace{1em}

    \begin{subfigure}[b]{\figwidth}
        \centering
        \includegraphics[width=\linewidth]{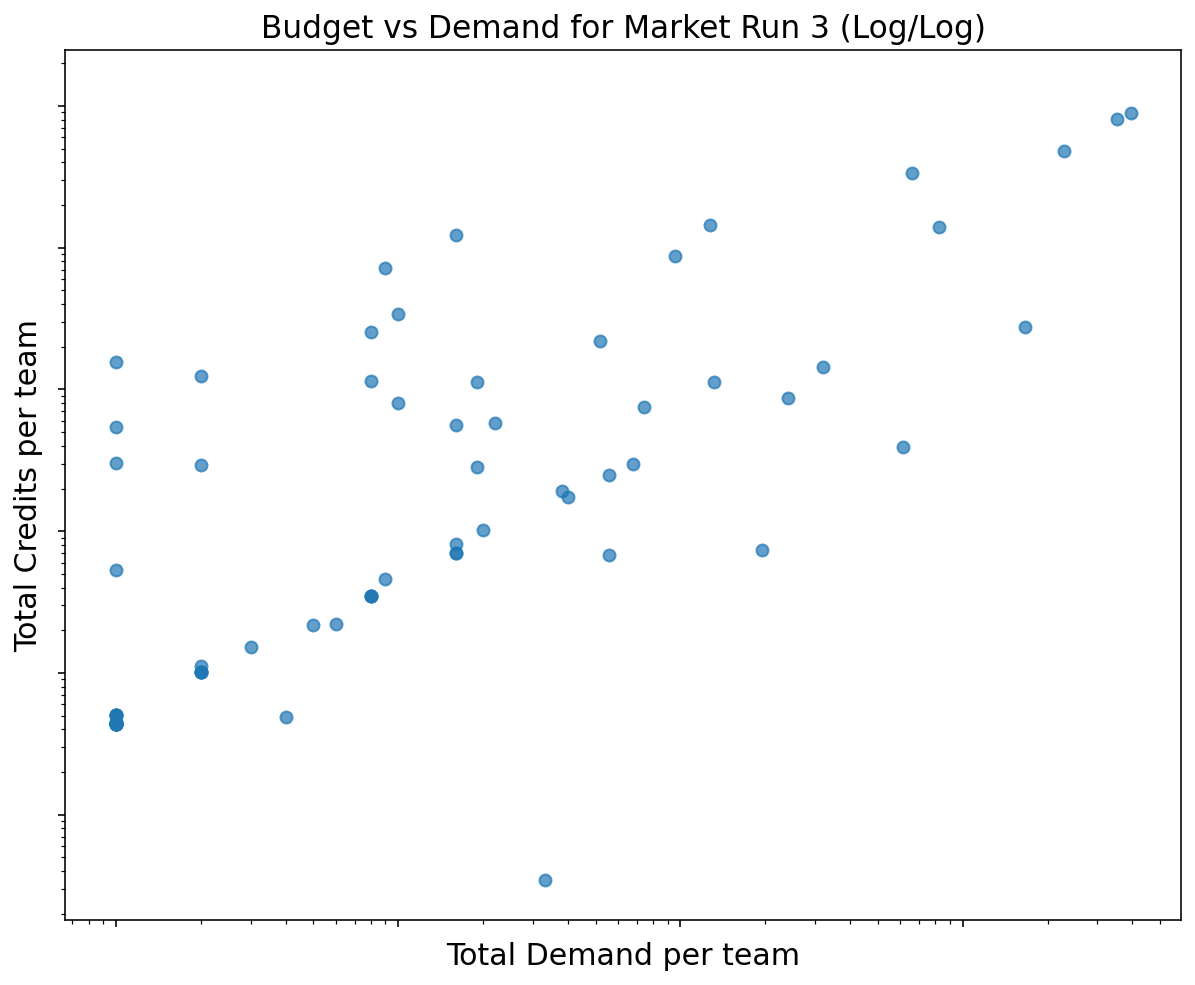}
    \end{subfigure}
    \hfill
    \begin{subfigure}[b]{\figwidth}
        \centering
        \includegraphics[width=\linewidth]{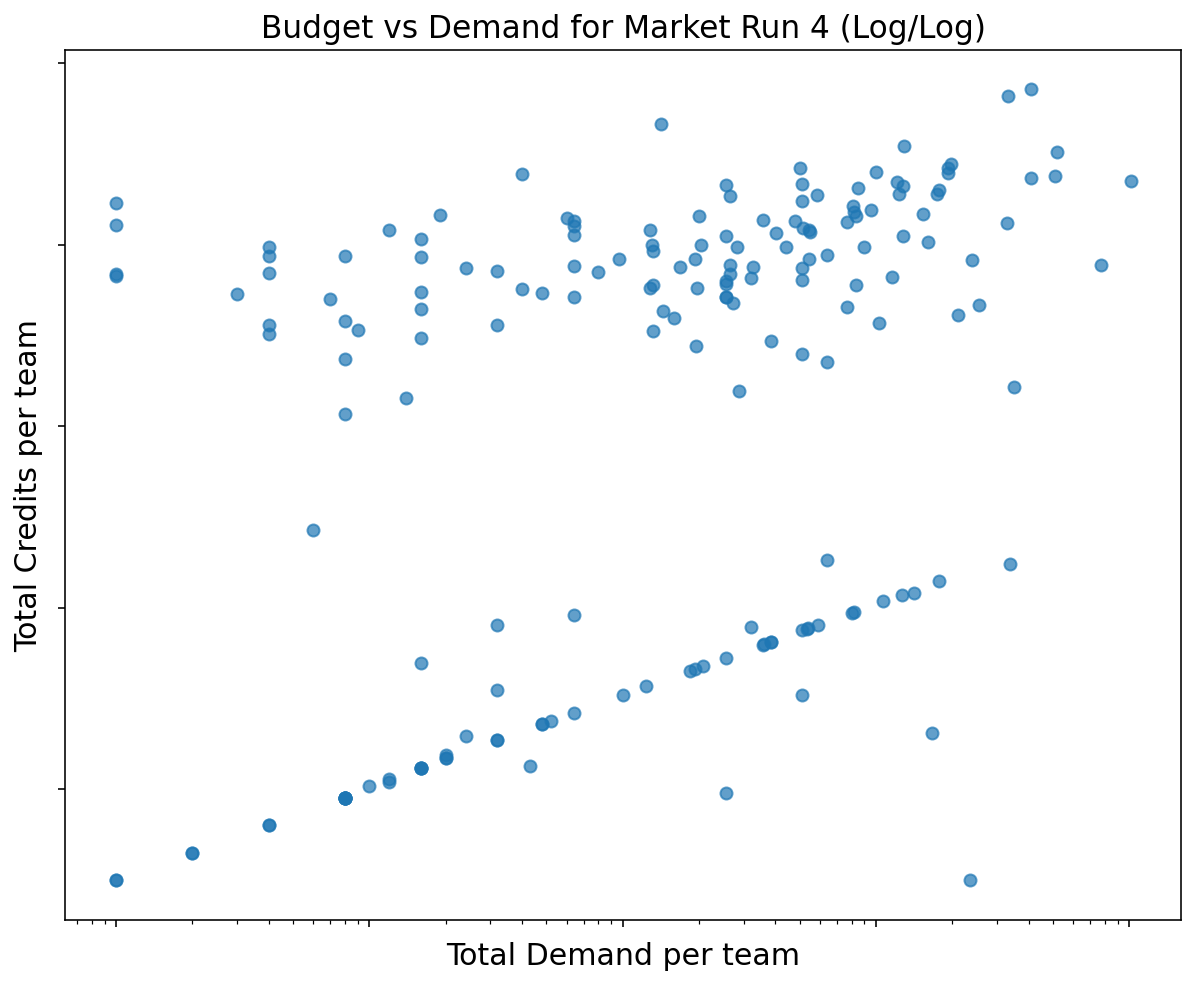}
    \end{subfigure}

    \caption{Credits and Demand Scatterplot for different Market Runs corresponding to the allocation of different resources. Axes are omitted for confidentiality.}
    \Description{}
    \label{fig:bid_demand_plot}
\end{figure}

The agent hierarchy $\agentTree$ is a three-level tree where the root represents the entire company, the first level corresponds to business units, and the leaves correspond to engineering teams within those business units.
All the resource demand originates at the leaves.
Each team at leaf $i$ has a maximum physical demand for chips $d_i^r$ and a total amount of credits $B_i^r$ available, for each chip type $r$, making the demand function $D_i^r(p) = \min\{ B_i^r / p, d_i^r \}$.
Figure \ref{fig:bid_demand_plot} shows the distribution of $(B_i^r, d_i^r)$ for all four markets.
The maximum demand $d_i^r$ reflects the ML jobs that the team is currently running or trying to start.
The amount of credits $B_i^r$ allocated to each team is determined by a top-down planning process that is exogenous to the market.
Similarly, each team has settings that control what fraction of their credits is allocated towards the purchase of each resource type, which is also exogenous.
The supply is distributed across the tree, but the vast majority of it is either pooled company-wide at the root or localized to each business unit at the first level.

In production, the market clears every minute and provisions access to resources between market runs. In this paper, we abstract away the dynamic aspects of the market to focus on computing the static equilibrium for a single snapshot.

\subsection{Convergence Results}

To track the convergence of the Budget Descent Algorithm (\cref{algo:BDA}), we measure the relative discrepancy of the virtual budgets at each iteration. For every agent $i$ and resource $r$, let $\hat B_i^r[t]$ be the active virtual budget at iteration $t$, and $p_i^r[t]$ be the computed prices.
The optimal target budget for that iteration, dictated by the demand caps, is $\min\{ B_i^r, d_i^r \cdot p_i^r[t] \}$.
We compute the error ratio for every $i, r$:
$ \text{Error}_i^r = \hat B_i^r[t] / \min\{ B_i^r, d_i^r \cdot p_i^r[t] \} .$

In a perfect equilibrium, this ratio is exactly $1$ for all agents and resources.
Because our algorithm strictly descends the virtual budgets, they approach the target monotonically from above, meaning the error ratio is always $\ge 1$, as shown in Figure~\ref{fig:budget_descent}.
Each subfigure corresponds to one of the four market runs and plots the maximum, average, and minimum error ratios across the entire hierarchy over $200$ iterations.
As the logarithmic plots demonstrate, the maximum and average errors tightly and smoothly converge to $1$.
This validates that the algorithm efficiently and reliably finds the exact equilibrium for the large-scale capped harmonic market.

We also show a different statistic throughout this iterative process.
Because the solution we calculate is approximate, there will be a discrepancy between the payment an agent is expecting and is actually receiving, or an agent will be asked for more resources than she has available.
For this reason, we demonstrate the following.
First, the \textit{Payment Error}: the total amount of error in the Payment Conservation axiom, normalized by the total amount of budgets in the system.
Second, the \textit{Allocation Error}: the total excess of resources that an agent is asked to provide, normalized by the total amount of resources in the system.
We showcase these two in \cref{fig:budget_descent_pay_alloc}, where both errors decrease over time.
This confirms that as we run more and more iterations of the Budget Descent Algorithm, the quality of the solution improves.

\begin{figure}[htbp]
    \centering
    \begin{subfigure}[b]{\figwidth}
        \centering
        \includegraphics[width=\linewidth]{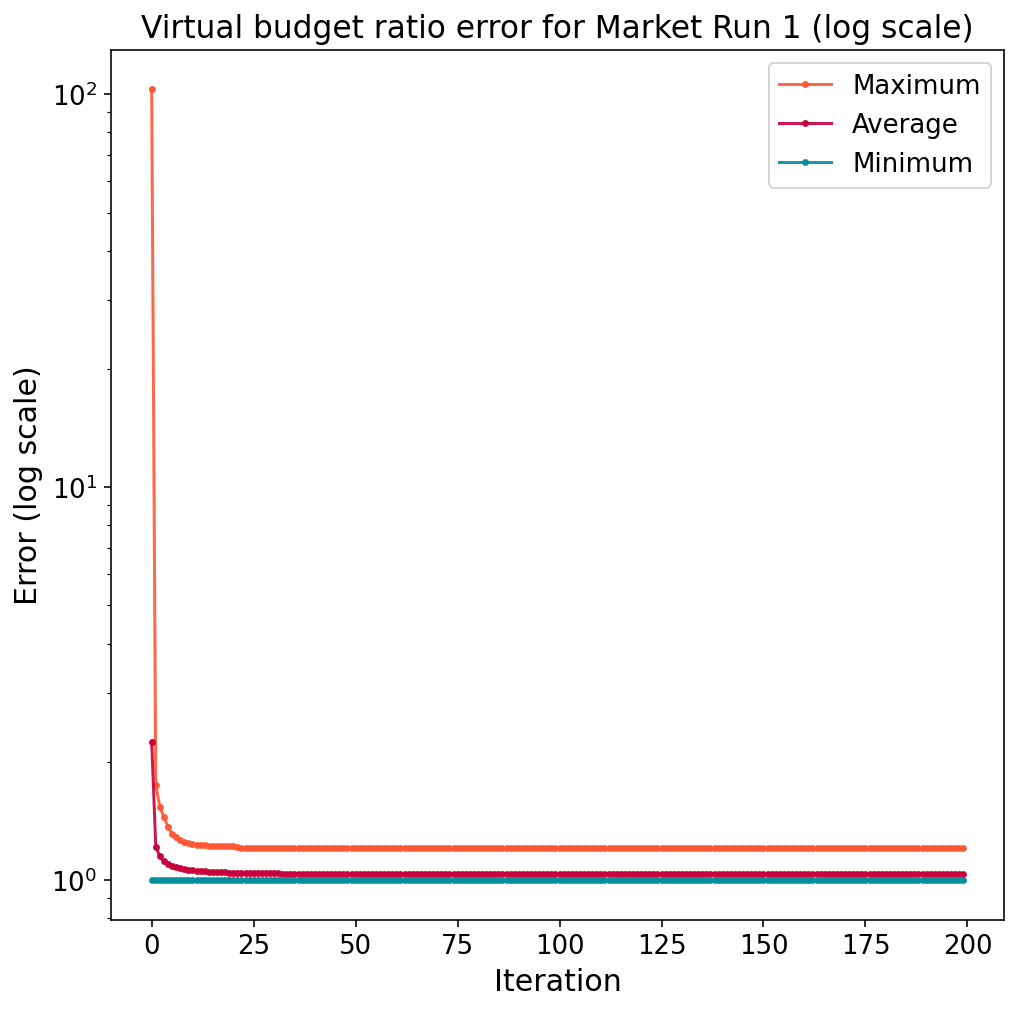}
    \end{subfigure}
    \hfill
    \begin{subfigure}[b]{\figwidth}
        \centering
        \includegraphics[width=\linewidth]{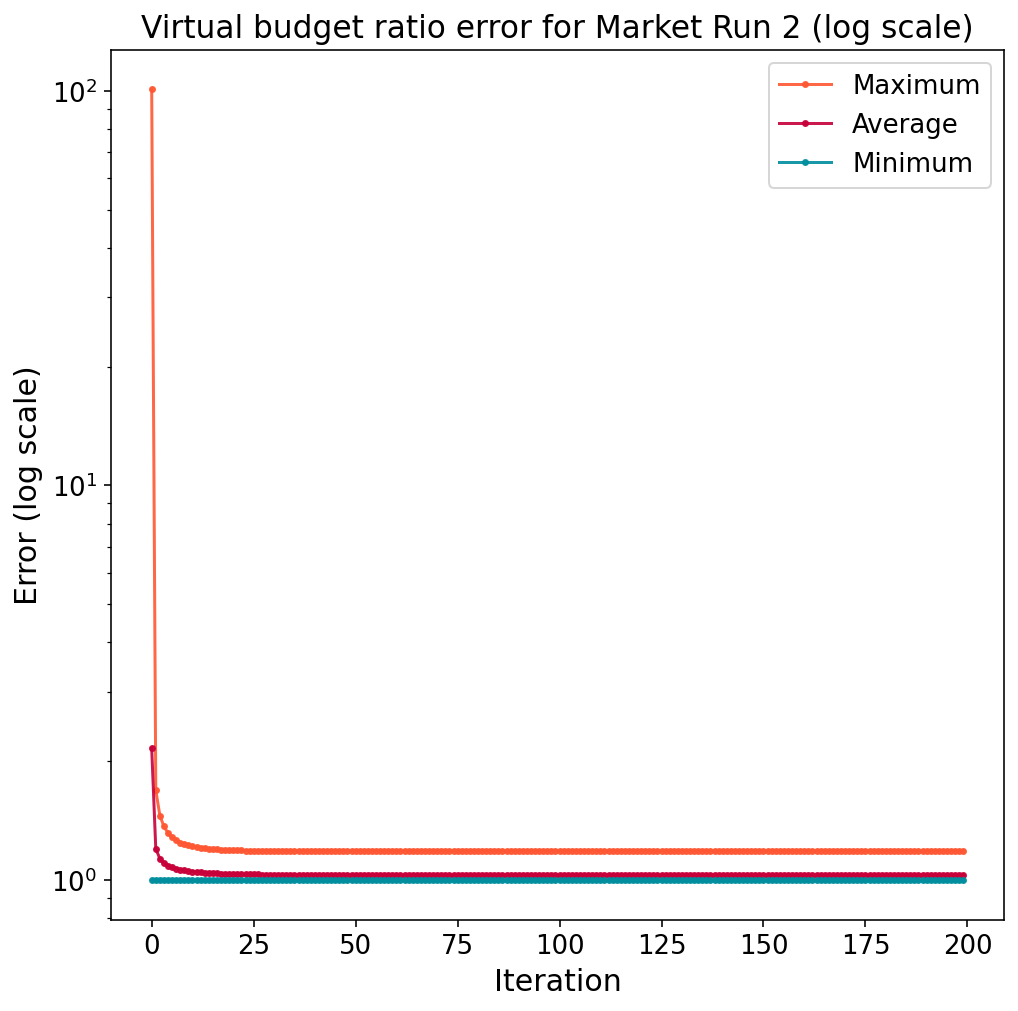}
    \end{subfigure}

    \vspace{1em}

    \begin{subfigure}[b]{\figwidth}
        \centering
        \includegraphics[width=\linewidth]{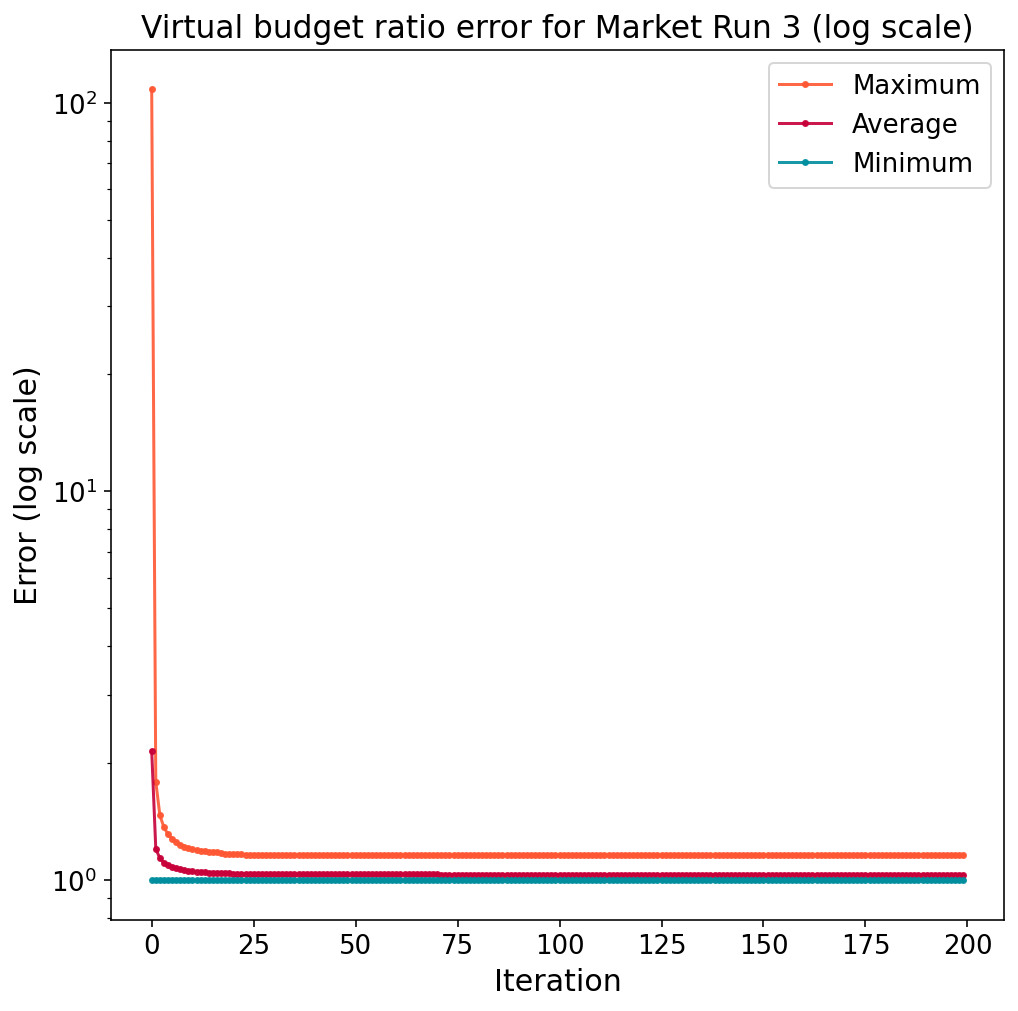}
    \end{subfigure}
    \hfill
    \begin{subfigure}[b]{\figwidth}
        \centering
        \includegraphics[width=\linewidth]{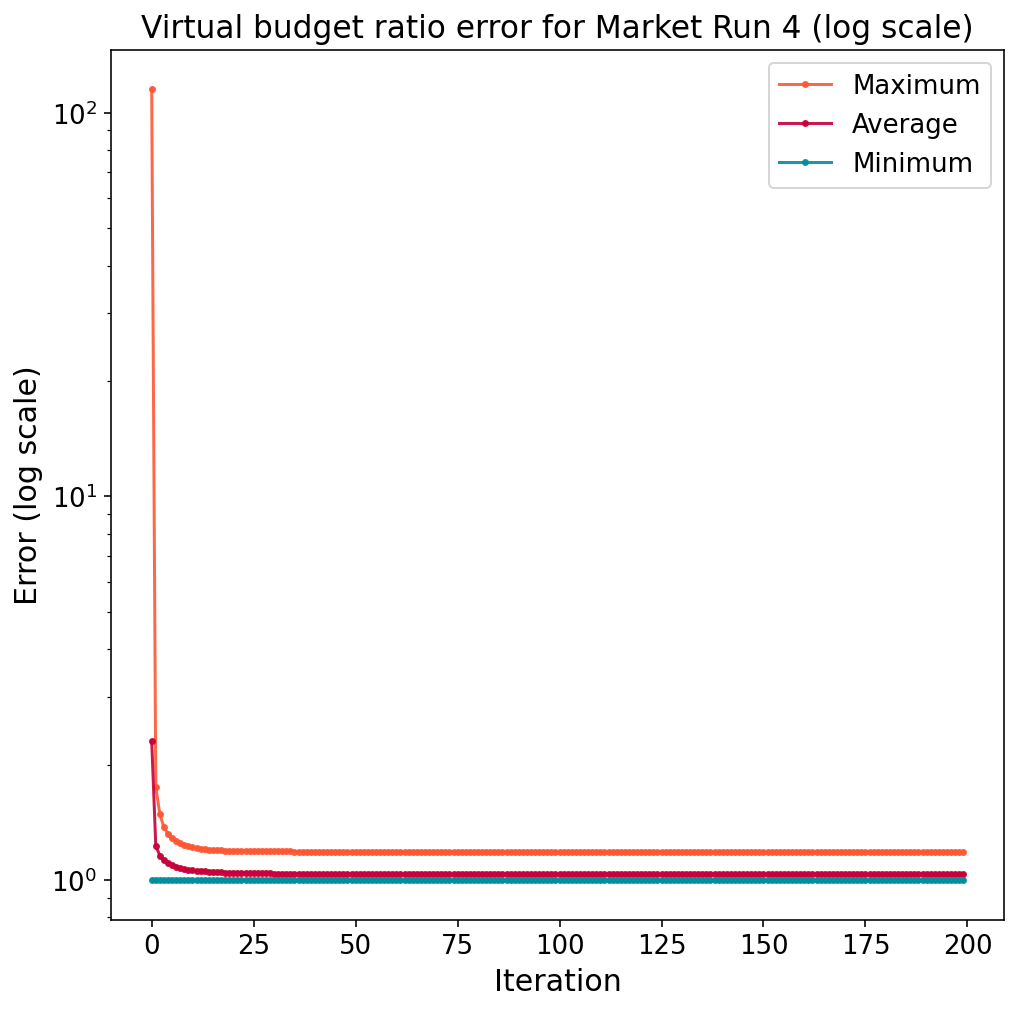}
    \end{subfigure}

    \caption{Error virtual budgets of the budget descent algorithm as a function of iterations}
    \Description{}
    \label{fig:budget_descent}
\end{figure}

\begin{figure}[htbp]
    \centering
    \begin{subfigure}[b]{\figwidth}
        \centering
        \includegraphics[width=\linewidth]{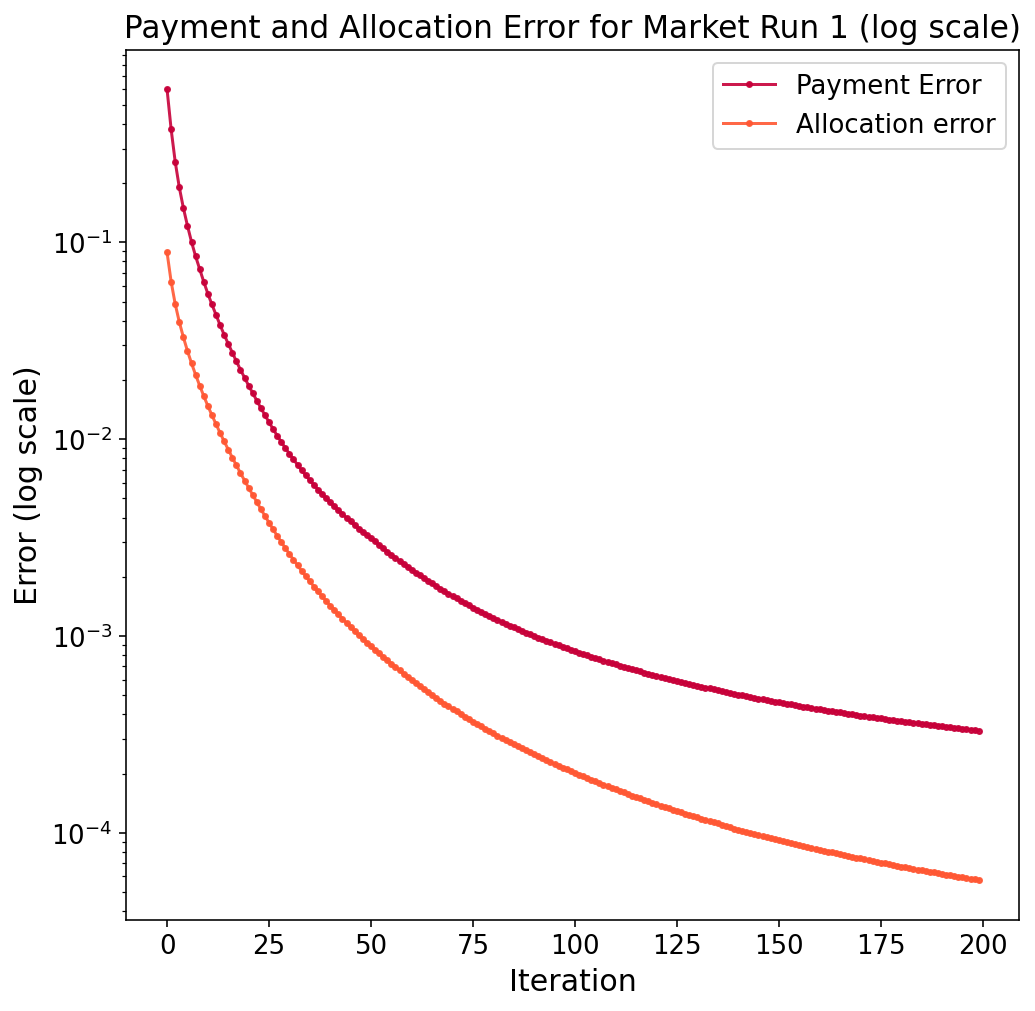}
    \end{subfigure}
    \hfill
    \begin{subfigure}[b]{\figwidth}
        \centering
        \includegraphics[width=\linewidth]{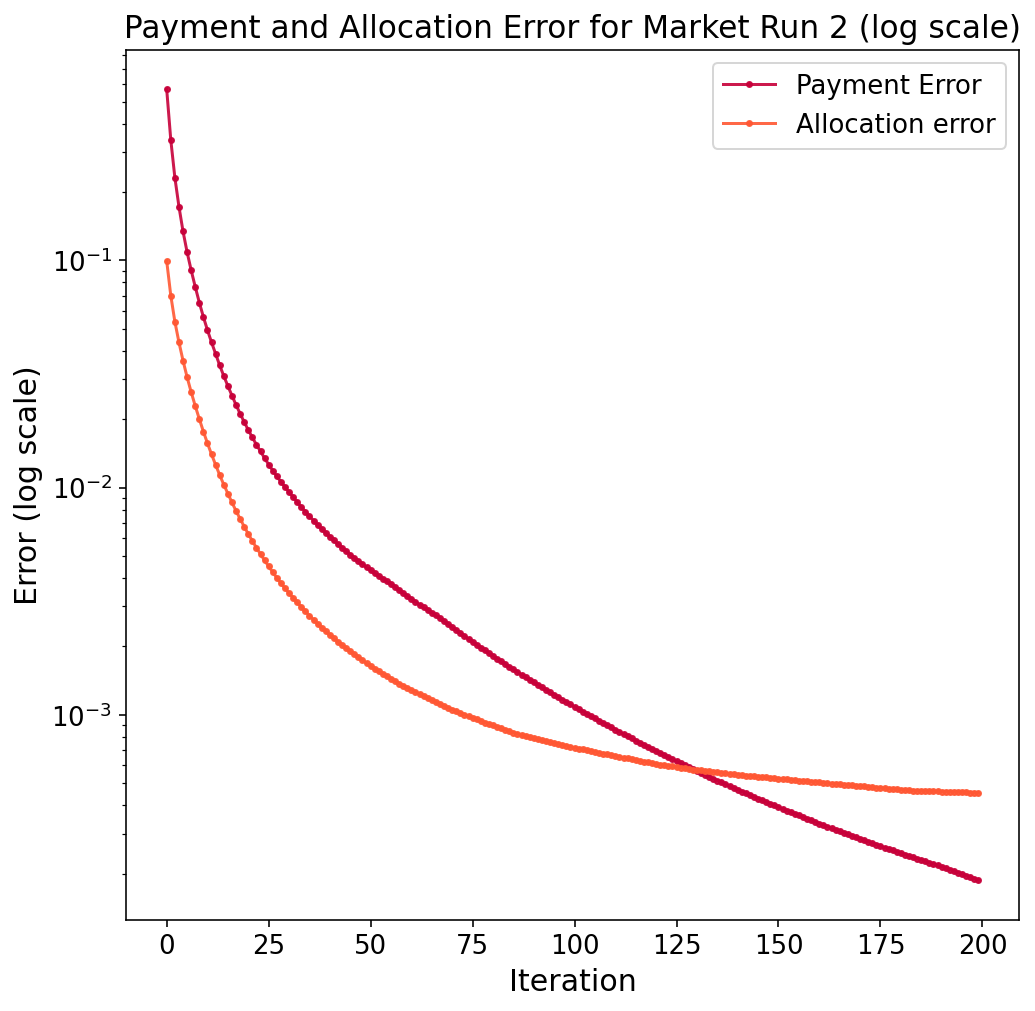}
    \end{subfigure}

    \vspace{1em}

    \begin{subfigure}[b]{\figwidth}
        \centering
        \includegraphics[width=\linewidth]{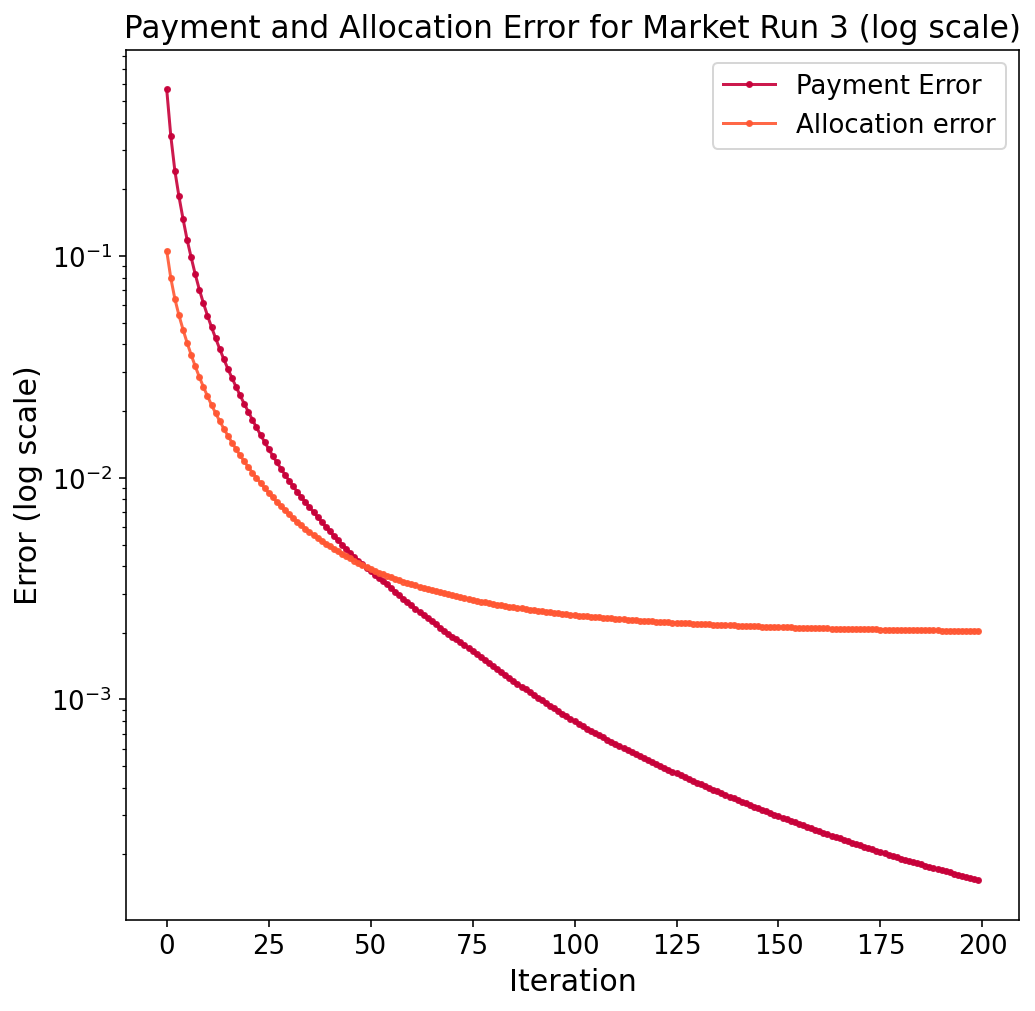}
    \end{subfigure}
    \hfill
    \begin{subfigure}[b]{\figwidth}
        \centering
        \includegraphics[width=\linewidth]{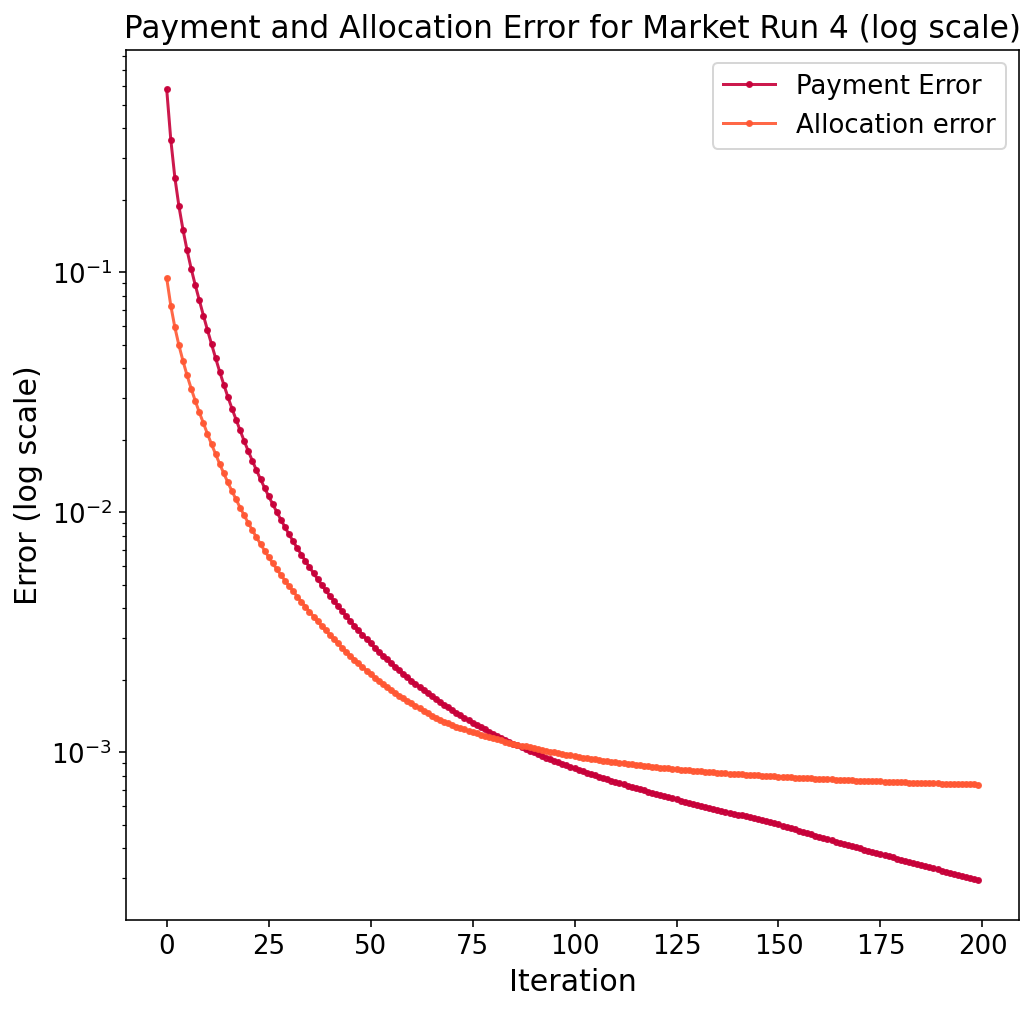}
    \end{subfigure}

    \caption{Relative error in allocations and payments of the Budget Descent Algorithm as a function of iterations}
    \Description{}
    \label{fig:budget_descent_pay_alloc}
\end{figure}

\subsection{Algorithm Implementation and Combinatorial Solver}
\label{ssec:experiments:algo_implementation}

Our implementation of the Budget Descent Algorithm is similar to what is outlined in \cref{algo:BDA}.
The only substantial change, that offers a major computational advantage, is that we do not rely on a generic, heavy-weight convex optimization solver for the inner loop of the algorithm. The conversion graph present in the data for this market is a star-graph, which allows us to solve the uncapped harmonic subproblem exactly, using a highly efficient combinatorial algorithm.
Because all specific locations can independently satisfy the global demand, we can determine the prices for each agent via a greedy pooling approach which we describe below.

\bfparagraph{Notation and Subproblem Formulation}
The resources are $[m]\cup\{0\}$, where $0$ is the global resource to which all the other resources can be converted: the edges of $\convGraph$ are $\qty\big( (r, 0) )_{r\in[m]}$.
Consider a fixed agent $i$ during a single iteration of the Budget Descent algorithm.
To compute the local prices $p_i^r$, we first define the agent's effective virtual budget and effective supply.
Let $\beta_i^r = \sum_{j \in \subTree(i)} \hat{b}_j^r$ denote the effective virtual budget at resource $r$ by agent $i$ and its entire sub-tree.
Let $\sigma_i^r$ denote the total effective supply of resource $r$ available to agent $i$: $\sigma_i^r = s_i^r$ if $i$ is the root and $\sigma_i^r = s_i^r + \frac{\beta_i^r}{p_{\pi(i)}^r}$ otherwise.
Then the subproblem we have, dropping the $i$ subscript for simplicity, is
\begin{equation} \label{eq:experiments:harmonic_optimization}
    \max_{p^0, p^1, \ldots, p^m} \quad
    \sum_{r \in [m] \cup \{0\}} \qty( \beta^r \log p^r - \sigma^r p^r )
    \qquad
    \textrm{such that } \quad
    p^r \ge p^0 \ge 0 \qquad \forall r \in [m]
\end{equation}

\bfparagraph{The Greedy Pooling Algorithm.}
Because all specific locations can independently satisfy the global demand, we can determine the exact optimal prices of \cref{eq:experiments:harmonic_optimization} via a greedy pooling approach:

\begin{enumerate}
    \item \textbf{Calculate Unconstrained Prices:} Compute the price of each specific resource except the global one $r \in [m]$, assuming no restriction on the prices.
    These are simply $\tilde p^r = \beta^r / \sigma^r$.
    
    \item \textbf{Sort:} Relabel the resources so that they are in ascending order: $\tilde p^1 \le \tilde p^2 \le \ldots \le \tilde p^{m}$. 
    
    \item \textbf{Greedy Pooling:}
    Initialize a resource pool $\mathcal P = \{0\}$, initially only containing the global resource with (temporary) solution $\tilde p^0 = \beta^0 / \sigma^0$. 
    
    Iterate $r$ from $1$ to $m$:
    \begin{itemize}
        \item If $\tilde p^r < \tilde p^0$, allow the conversion of resource $r$ to the global resource.
        Therefore, merge it into the pool: $\mathcal P \gets R \cup \{ r \}$ and set $\tilde p^0 \gets ({ \sum_{r' \in \mathcal P}\beta^{r'} })/({ \sum_{r' \in \mathcal P}\sigma^{r'} })$.
        
        \item Otherwise, if $\tilde p^r \ge \tilde p^0$, halt the pooling process.
    \end{itemize}
    
    \item \textbf{Output:} 
    Set the prices of the pooled resources, $p^r = \tilde p^0$ for $r \in \mathcal P$. 
    For unpooled resources $r \notin \mathcal{P}$, set the unrestricted prices $p^r = \tilde p^r$.
\end{enumerate}

By replacing generic numerical optimization solvers with this exact $\order{m \log m}$ combinatorial subroutine, we drastically accelerate the running time of each iteration,
enabling the outer Budget Descent iterations to seamlessly scale to extensive organizational trees and high-dimensional resource spaces.

\bfparagraph{Correctness of the Greedy Pooling Algorithm.}
We want to solve the following problem:
\begin{equation} \label{eq:experiments:harmonic_optimization2}
    \max_{p^0, p^1, \ldots, p^m} \quad
    \sum_{r \in [m] \cup \{0\}} \qty( \beta^r \log p^r - \sigma^r p^r )
    \qquad
    \textrm{such that } \quad
    p^r \ge p^0 \ge 0 \qquad \forall r \in [m]
\end{equation}

Consider the Lagrangian with multiplier $c^{(r, 0)}$ for $r \in [m]$
\begin{equation*}
    L(p, c) =
    \sum_{r \in [m] \cup \{0\}} \qty(\big. \beta^r \log p^r - \sigma^r p^r )
    +
    \sum_{r \in [m]} c^{(r, 0)} \qty(\big. p^r - p^0 )
\end{equation*}

Since the problem is concave and the feasible region has an interior (i.e., Slater's condition holds), all we have to find is a solution where the primal and dual optimality conditions hold: for $r \in [m]$,
\begin{equation*}
    \frac{\partial L}{\partial p^r}
    =
    \frac{\beta^r}{p^r} - \sigma^r + c^{(r, 0)}
    \le
    0
    \quad\bot\quad
    p^r \ge 0
    \quad\implies\quad
    p^r = \max\qty{
        \frac{\beta^r}{\sigma^r - c^{(r, 0)}}, 0
    }
\end{equation*}
and for $r = 0$
\begin{equation*}
    \frac{\partial L}{\partial p^0}
    =
    \frac{\beta^0}{p^0} - \sigma^0 - \sum_{r \in [m]} c^{(r, 0)}
    \le
    0
    \quad\bot\quad
    p^0 \ge 0
    \quad\implies\quad
    p^r = \max\qty{
        \frac{\beta^0}{\sigma^0 + \sum_{r \in [m]} c^{(r, 0)}}, 0
    }
\end{equation*}

Rearrange the $[m]$ resources such that $\frac{\beta^1}{\sigma^1} \le \frac{\beta^2}{\sigma^2} \le \ldots \le \frac{\beta^m}{\sigma^m}$.
Also let $\mathcal P = \{0\} \cup [r^*]$ as in the algorithm: $r^*$ is such that
for all $r \le r^*$
\begin{equation} \label{eq:987}
    \frac{
        \sum_{r' \in \{0\} \cup [r-1]}\beta^{r'}
    }{
        \sum_{r' \in \{0\} \cup [r-1]}\sigma^{r'}
    }
    >
    \frac{\beta^{r}}{\sigma^{r}}
\end{equation}
and for $r^*$
\begin{equation*} \label{eq:654}
    \frac{
        \sum_{r' \in \{0\} \cup [r^*]}\beta^{r'}
    }{
        \sum_{r' \in \{0\} \cup [r^*]}\sigma^{r'}
    }
    \le
    \frac{\beta^{r^*+1}}{\sigma^{r^*+1}}
\end{equation*}

We now set for $r \in [r^*]$
\begin{equation*}
    p^r
    =
    \frac{\beta^r}{\sigma^r - c^{(r, 0)}}
    =
    \frac{\sum_{r \in \{0\} \cup [r^*]}\beta^r}{\sum_{r \in \{0\} \cup [r^*]}\sigma^r}
    \quad\implies\quad
    c^{(r, 0)} = \sigma^r - \beta^r \frac{\sum_{r \in \{0\} \cup [r^*]}\sigma^r}{\sum_{r \in \{0\} \cup [r^*]}\beta^r}
    \ge
    0
\end{equation*}
where the inequality is implied by \cref{eq:987}.
The above makes
\begin{equation*}
    p^0
    =
    \frac{\beta^0}{\sigma^0 + \sum_{r\in[r^*]} c^{(r,0)}}
    =
    \frac{
        \beta^0
    }{
        \sigma^0 + \sum_{r\in[r^*]} \qty( \sigma^r - \beta^r \frac{\sum_{r \in \{0\} \cup [r^*]}\sigma^r}{\sum_{r \in \{0\} \cup [r^*]}\beta^r} )
    }
    =
    \frac{
        \sum_{r' \in \{0\} \cup [r^*]}\beta^{r'}
    }{
        \sum_{r' \in \{0\} \cup [r^*]}\sigma^{r'}
    }
\end{equation*}
as needed by the algorithm.
This completes the correctness.

\end{document}